\documentclass[11pt]{article}
\usepackage[margin=1in]{geometry}
\usepackage{amsthm}

\usepackage{cite}
\usepackage{graphicx}
\usepackage[tight,footnotesize]{subfigure}
\usepackage[cmex10]{amsmath}
\usepackage{amsfonts}

\newtheorem{assumption}{Assumption}
\newtheorem{lemma}{Lemma}
\newtheorem{theorem}{Theorem}
\newtheorem{proposition}{Proposition}

\theoremstyle{definition}
\newtheorem{remark}{Remark}

\newcommand{\diag}{\operatorname{diag}}

\newcommand{\calG}{\mathcal{G}}
\newcommand{\calV}{\mathcal{V}}
\newcommand{\calE}{\mathcal{E}}
\newcommand{\calN}{\mathcal{N}}

\renewcommand{\Re}{\operatorname{Re}}

\begin{document}

\title{\bf Broadcast Gossip Algorithms for Consensus on Strongly Connected Digraphs}


\author{Wu~Shaochuan
        and~Michael~G.~Rabbat\footnotemark[0]
}

\renewcommand{\thefootnote}{\fnsymbol{footnote}}
\footnotetext{Wu Shaochuan is with the Department
of Electrical and Information Engineering, Harbin Institute of Technology, Harbin, 150080 China e-mail: scwu@hit.edu.cn. This work was conducted while W.S.~was visiting McGill University. 

Michael G. Rabbat is with the Department of Electrical and Computer Engineering, McGill University, Montreal, QC, H3A 0E9 Canada, e-mail: michael.rabbat@mcgill.ca.

The work of W.S.~was partially supported by the National Science Foundation of China (NSFC) grant No.~61201147. The work of M.G.R.~was partially supported by grants from the Natural Sciences and Engineering Research Council of Canada (NSERC) and the \emph{Minist\`{e}re du D\'{e}veloppement \'{e}conomique, de l'Innovation et de l'Exportation} (MDEIE).}

\maketitle

\begin{abstract}
We study a general framework for broadcast gossip algorithms which use companion variables to solve the average consensus problem. Each node maintains an initial state and a companion variable. Iterative updates are performed asynchronously whereby one random node broadcasts its current state and companion variable and all other nodes receiving the broadcast update their state and companion variable. We provide conditions under which this scheme is guaranteed to converge to a consensus solution, where all nodes have the same limiting values, on any strongly connected directed graph. Under stronger conditions, which are reasonable when the underlying communication graph is undirected, we guarantee that the consensus value is equal to the average, both in expectation and in the mean-squared sense. Our analysis uses tools from non-negative matrix theory and perturbation theory. The perturbation results rely on a parameter being sufficiently small. We characterize the allowable upper bound as well as the optimal setting for the perturbation parameter as a function of the network topology, and this allows us to characterize the worst-case rate of convergence. Simulations illustrate that, in comparison to existing broadcast gossip algorithms, the approaches proposed in this paper have the advantage that they simultaneously can be guaranteed to converge to the average consensus and they converge in a small number of broadcasts.
\end{abstract}

\section{Introduction} \label{sec:intro}

Gossip algorithms are an attractive solution for information processing in applications such as distributed signal processing~\cite{DI10}, networked control~\cite{OL07}, and multi-robot systems~\cite{ME10}. They are attractive because they require little infrastructure; nodes iteratively pass messages with their immediate neighbors in a network until they reach a consensus on the solution. Consequently, there is little overhead associated with forming and maintaining specialized routes, and there are no bottlenecks or single points of failure. 

Broadcast gossip algorithms, introduced in~\cite{AY08a,AY08b,AY09}, are especially attractive for use in wireless networks. Unlike the majority of existing gossip algorithms, where messages are either asynchronously exchanged between pairs of nodes or where nodes synchronously exchange and process messages with all of their neighbors, in broadcast gossip algorithms nodes asynchronously broadcast a message and the message contents are immediately processed by all neighbors receiving it. By exploiting the broadcast nature of wireless communications, broadcast gossip algorithms are more efficient (they converge after fewer transmissions) than other gossip algorithms~\cite{AY09}. However, previously proposed broadcast gossip algorithms either converge to a consensus on a random solution~\cite{AY09}, which may not be acceptable in practical applications, or they do not have theoretical guarantees~\cite{FR09}.

In this article we propose and analyze a family of broadcast gossip algorithms for strongly connected directed graphs. If the network is symmetric (undirected) or if nodes know their out-degree, these algorithms are guaranteed to converge to the average consensus both in expectation and in the mean-squared sense. In more general settings, the algorithms are still guaranteed to converge to a specific solution which is a convex combination of the initial values at all nodes in the network (but not necessarily the average). We give a precise characterization of this solution in terms of the algorithm parameters. Our analysis combines tools and techniques from non-negative matrix theory and matrix perturbation theory. Along these lines, we derive an upper bound on the perturbation parameter under which convergence is guaranteed, and we derive an expression for the optimal value of the perturbation parameter.

\subsection{Related Work}

\emph{Broadcast gossip algorithms} (BGAs) are introduced in the series of papers by Aysal et al.~\cite{AY08a,AY08b,AY09}. The BGAs proposed there involve nodes asynchronously transmitting a scalar-valued message, and each time a node receives a message from its neighbors it performs an update by forming a convex combination of the received value with their own previous value. Then, when it is a given node's turn to broadcast next (as determined by a random timer, in the asynchronous model~\cite{TS86,BE97}), the node broadcasts its current value. In~\cite{AY08a,AY08b,AY09} it is shown that, when executed over an undirected graph (i.e., one with symmetric links) such an algorithm converges to a consensus solution almost surely. The updates of this algorithm are linear and can be expressed as a random, time-varying matrix acting on the vector containing the state values at each node. Unlike conventional pairwise or synchronous gossip algorithms, the matrices in~\cite{AY08a,AY08b,AY09} corresponding to the update when a particular node transmits cannot be viewed as the transition matrix of a reversible Markov chain, and so the average is not preserved from iteration to iteration. Consequently, although the consensus value is equal to the average of the initial values at every node in expectation, for any particular sample path (where the randomness is in the sequence determining the order in which nodes broadcast) the consensus value is randomly distributed about the average of their initial values but is not precisely equal to it. 

Subsequent recent work~\cite{FR11} investigates related BGAs, demonstrating that their convergence properties are robust even when the broadcasts from different nodes may interfere at a receiver. A broadcast-based algorithm has also been proposed for solving distributed convex optimization problems~\cite{NE11}.

A modified BGA is proposed by Franceschelli et al.~\cite{FR09,FA11}, where nodes maintain a companion (or surplus) variable in addition to the state variable they seek to average. By careful accounting of both the companion and state variables, a conservation principle is established, and simulation results suggest that the algorithm with companion variables converges to the average consensus for all sample paths, not just in expectation. However, no proof of convergence or theoretical convergence rate analysis is available for the algorithm of~\cite{FR09,FA11}. 

Recent work of Cai and Ishii~\cite{CA11,CA12} analyzes related distributed averaging algorithms on directed graphs that use companion variables. The two types of algorithms analyzed in~\cite{CA11,CA12} involve asynchronous pairwise updates and synchronous updates. They make use of tools from matrix perturbation theory, and the work in the present article can be seen as generalizing the results in~\cite{CA11,CA12} for broadcast gossip updates.

\subsection{Contributions and Paper Organization}

The contributions of this article are as follows. In Section~\ref{sec:framework} we propose a general framework for broadcast gossip algorithms over directed graphs using companion variables. For this framework we determine conditions on the algorithm parameters under which convergence to a consensus is guaranteed both in expectation (in  Section~\ref{sec:convergenceInExpectation}) and in the mean squared sense (in Section~\ref{sec:convergenceSecondMoment}). 

We then consider two specific instances of the general framework in Sections~\ref{sec:ubga} and \ref{sec:bbga}. In one instance, which we refer to as unbiased broadcast gossip algorithms (cf.~Section~\ref{sec:ubga}), the consensus value is guaranteed to be the average of the initial values. In the other instance (biased broadcast gossip algorithms, Sec.~\ref{sec:bbga}), the consensus value is no longer the average of the initial values, but it depends on the stationary distribution of a Markov chain associated with the algorithm parameters. The unbiased algorithm requires that each node be aware of its out-degree, the number of nodes that receive its broadcasts. This is a reasonable assumption in networks where connectivity is symmetric, but it may not be reasonable in networks with directed edges. In particular, if there are directed edges, then there is no immediate feedback link, making it more challenging for a node to identify the out-neighbors that receive its broadcasts. This motivates further study of the biased BGAs, which are more practical in such scenarios because they do not require that nodes know their out-degree. 

Our analysis of the general framework makes use of tools from matrix perturbation theory. In particular, the way in which the information in the companion variables is incorporated back into the main state variables depends on a parameter which can be viewed as controlling the extent to which a baseline linear system is perturbed. For sufficiently small values of the perturbation parameter, the algorithm is guaranteed to converge. In Section~\ref{sec:supremum} we determine a tight upper bound on the allowable values for the perturbation parameter for biased broadcast gossip algorithms. This bound constitutes an improvement over previous bounds along these lines because it explicitly takes into account the structure of the graph through spectral properties of a corresponding graph Laplacian matrix. In addition to determining this bound, we identify a topology-dependent optimal value for the perturbation parameter in Section~\ref{sec:optimal}, and we obtain an expression for the resulting second largest eigenvalue which governs the worst-case rate of convergence.

Simulation results, reported in Section~\ref{sec:perf}, demonstrate that the proposed broadcast gossip algorithms fare well compared to the existing algorithms~\cite{AY09,FR09}. The algorithm of~\cite{AY09} converges quickly but can converge to a consensus value which is very far from the average. The algorithm of~\cite{FR09} converges to the average consensus but requires significantly more iterations than the algorithm of~\cite{AY09}. The algorithms proposed here converge quickly and they can be made to converge to the average consensus. We conclude in Section~\ref{sec:conclusion}.

\subsection{Notation}

Before proceeding, we summarize some of the notation used in this article. Let $x \in \mathbb{R}^n$ be a $n$-dimensional column vector. The Euclidean norm of $x$ is denoted by $\|x\|_2$. Let $A$ be a $n \times n$ matrix with real-valued entries. Let $[A]_{i,j}$ denote the entry in the $i$th row and $j$th column of $A$; we also write $A_{i,j}$ when there is no ambiguity. The $\infty$-norm of $A$ is given by $\|A\|_{\infty} = \max_i \sum_{j=1}^n |A_{i,j}|$, the largest absolute row sum, and the $1$-norm of $A$ is give by $\|A\|_1 = \max_j \sum_{i=1}^n |A_{i,j}|$, the largest absolute column sum. The spectral radius of $A$ is the largest modulus of an eigenvalue of $A$ and is denoted by $\rho(A) = \max_{i=1,\dots,n} |\lambda_i(A)|$, where $\lambda_1(A), \dots, \lambda_n(A)$ are the eigenvalues of $A$. For the vector $x \in \mathbb{R}^n$, let $\diag(x)$ denote a $n \times n$ diagonal matrix with $[\diag(x)]_{i,i} = x_i$. For a $n \times n$ matrix $A$, let $\diag(A)$ denote a $n$-dimensional column vector with $[\diag(A)]_i = A_{i,i}$.

\section{Framework for Broadcast Gossip Algorithms} \label{sec:framework}

\subsection{Network Model} \label{sec:networkModel}

Let $\calG = (\calV, \calE)$ be a directed graph which represents the network connectivity, where $\calV = \{1,\dots,n\}$ is the set of nodes and $\calE \subseteq \calV \times \calV$ is the set of directed edges. The network contains a directed edge $(i,j) \in \calE$ if and only if node $i$ receives messages transmitted by node $j$. Let $\calN^+_i = \{j \in \calV \colon (i,j) \in \calE\}$ and $\calN^-_i = \{k \in \calV \colon (k,i) \in \calE\}$ denote the set of in-neighbors and out-neighbors, respectively, of node $i$. 
For the rest of this paper we make the following assumption.

\begin{assumption} \label{stronglyConnected}
The graph $\mathcal{G}$ is strongly connected; i.e., for any pair of nodes $i,j \in \mathcal{V}$, there exists a sequence of nodes $i = i_0, i_1, i_2, \dots, i_m = j$ such that $(i_{\ell-1}, i_{\ell}) \in \calE$ for all $\ell = 1,\dots,m$.
\end{assumption}

\subsection{Distributed Averaging} \label{sec:problemStatement}

The goal of any broadcast gossip algorithm is to accomplish distributed averaging. Each node $i \in \calV$ initially has a value, $x_i(0) \in \mathbb{R}$, and the goal is for all nodes to compute the average, $\frac{1}{n} \sum_{i =1}^n x_i(0)$. In general, distributed averaging algorithms seek to achieve consensus on the average while only allowing messages to be passed between neighboring nodes, as defined by the communication graph $\calG$. In broadcast gossip algorithms we make the following additional restrictions. Each node has a unique id, which corresponds to its index in the set $\{1,\dots,n\}$. When a node transmits a message, the message is received by all of its out-neighbors. The receiving nodes may know the id of the transmitter, but the transmitter will not, in general, know the ids of the receivers. Equivalently, each node knows the ids of its in-neighbors but not its out-neighbors.

\subsection{Asynchronous Time Model} \label{sec:timeModel}

Following~\cite{AY09}, we adopt the standard asynchronous time model~\cite{BE97}. Each node runs a clock which ticks according to an independent rate $1/n$ Poisson process. When node $k$'s clock ticks it initiates a broadcast gossip update, the details of which are described in the subsection that follows. Since the clocks at each node are independent, this model is equivalent to running a single, global Poisson clock which ticks at rate $1$, and assigning each tick uniformly and independently to one node in $\calV$. In the sequel we use the variable $t \in \{1, 2, \dots\}$ to index the ticks of this global Poisson clock. Each global clock tick corresponds to one update or iteration.

\subsection{Broadcast Gossip Updates} \label{sec:bgupdates}

Similar to previous broadcast gossip algorithms with companion variables~\cite{FR09,CA11}, every node $i$ maintains two variables, $x_i(t)$ and $y_i(t)$. The first variable, $x_i(t)$, is the estimate of the average at node $i$ after $t$ iterations, and it is initialized to $x_i(0)$, the same initial value from Sec.~\ref{sec:problemStatement}. The second variable, $y_i(t)$, is the companion variable at node $i$ after $t$ iterations, and it is initialized to $y_i(0) = 0$. The companion variables (called ``surplus'' variables in~\cite{CA12}), play the role of compensating for asymmetric updates made to $x_i(t)$, and if they are updated carefully, the companion variables can be used to ensure that consensus is achieved on the average.

When a node's clock ticks, it initiates an update by broadcasting its current state and companion value. Suppose that the $t+1$st global clock tick occurs at node $k$. Then $k$ broadcasts the values $x_k(t)$ and $y_k(t)$, and all nodes $j \in \calN^-_k$ which receive this information set
\begin{align}
x_j(t+1) &= (1 - a_{j,k}) x_j(t) + a_{j,k} x_k(t) + \epsilon d_j^{(k)} y_j(t) \label{eqn:updatexj} \\
y_j(t+1) &= a_{j,k} (x_j(t) - x_k(t)) + (1 - \epsilon d^{(k)}_j) y_j(t) + b_{j,k} y_k(t), \label{eqn:updateyj}
\end{align}
where the values of the algorithm parameters $a_{j,k}$, $b_{j,k}$, $d^{(k)}_j$, and $\epsilon > 0$ will be specified below. The transmitting node $k$ sets
\begin{align}
x_k(t+1) = x_k(t) \label{eqn:updatexk} \\
y_k(t+1) = 0, \label{eqn:updateyk}
\end{align}
and all other nodes $i \notin k \cup \calN_k^-$ keep
\begin{align}
x_i(t+1) = x_i(t) \label{eqn:updatexi} \\
y_i(t+1) = y_i(t). \label{eqn:updateyi}
\end{align}
Note that the nodes need not be aware of the global clock index $t$ to implement this protocol; they can simply update two local registers (one for $x_i$ and one for $y_i$) when they either broadcast a message or receive a broadcast. Below we continue to keep track of the global clock index for the purposes of analysis.

Different choices of the parameters $a_{j,k}$, $b_{j,k}$, $d^{(k)}_j$, and $\epsilon$ lead to different broadcast gossip algorithms with different properties; we will examine two particular choices of interest in Sections~\ref{sec:ubga} and \ref{sec:bbga}. Note that the seminal broadcast gossip algorithm of~\cite{AY09} is recovered by setting $\epsilon = 0$ and $a_{j,k} = \gamma$ for all $(j,k) \in \calE$. The broadcast gossip algorithm of~\cite{FR09} does not directly fit the form considered here, since in~\cite{FR09}, the receiving nodes also use $y_k(t)$ to calculate $x_j(t+1)$.

The broadcast gossip updates \eqref{eqn:updatexj}--\eqref{eqn:updateyi} are linear, and below we will use tools from linear algebra, spectral graph theory, and matrix perturbation theory to analyze their convergence properties. To this end, we introduce some additional notation. Let $A$ and $B$ be $n \times n$ matrices with entries $[A]_{i,j} = a_{i,j}$ and $[B]_{i,j} = b_{i,j}$, respectively, satisfying 
\begin{align}
\begin{cases} 0 < a_{i,j} \le 1 & \text{ if } (i,j) \in \calE \\ a_{i,j} = 0 & \text{ if } (i,j) \notin \calE, \end{cases} \label{eqn:constraints_a}
\end{align}
and
\begin{align}
\begin{cases} 0 < b_{i,j} \le 1 & \text{ if } (i,j) \in \calE \\ b_{i,j} = 0 & \text{ if } (i,j) \notin \calE. \end{cases} \label{eqn:constraints_b}
\end{align}
The matrices are graph-conformant in the sense that they have non-zero entries in locations corresponding to the edges of $\calG$.

We write $e_k \in \mathbb{R}^n$ for the $k$th canonical vector---the vector with all entries equal to $0$ except for the $k$th entry, which is equal to $1$. We also write $\mathbf{1}$ (respectively $\mathbf{0}$) for a $n$-dimensional vector with all entries equal to $1$ (respectively $0$).

Define $A_k = A e_k e_k^T$ and $B_k = B e_k e_k^T$. One can verify that $B_k$ is a $n \times n$ matrix, the $k$th column of $B_k$ is identical to that of $B$, and all other entries of $B_k$ are zero (and similar properties hold for $A_k$ in relation to $A$). It also follows directly from the definitions of $A_k$ and $B_k$ that $A = \sum_{k \in \calV} A_k$ and $B  = \sum_{k \in \calV} B_k$. 

The matrices $A$ and $B$ can be viewed as weighted adjacency matrices of the graph $\calG$ (possibly assigning different weights to each edge). From this view, the matrices $A_k$ and $B_k$ correspond to weighted adjacency matrices of a graph $\calG_k$ obtained from $\calG$ by eliminating all edges except those of the form $(i,k)$ for some $i \in \calV$, i.e., by retaining only those edges emanating from $k$. Thus, $\calG_k$ represents the graph of active edges when node $k$ transmits.

Finally, with this view of $A_k$ as a weighted adjacency matrix on $\calG_k$, let $L_k = \diag(A_k \mathbf{1}) - A_k$ denote the corresponding (directed) graph Laplacian. It follows from the definition of $L_k$ that $L_k \mathbf{1} = \mathbf{0}$. It also follows from the definition of $A_k$ that $\sum_{k \in V} L_k = \diag(A \mathbf{1}) - A \stackrel{\text{def}}{=} L$, where $L$ is the Laplacian corresponding to the graph with weighted adjacency matrix $A$.

The remaining algorithm parameters to discuss are $d^{(k)}_j$ and $\epsilon$. Let $d^{(k)} \in [0,1]^n$ denote a vector with values satisfying
\begin{align}
\begin{cases} d^{(k)}_j > 0 & \text{ if } (j, k) \in \calE \\ d^{(k)} = 0 & \text{ if } (j,k) \notin \calE, \end{cases} \label{eqn:constraints_d}
\end{align}
and let $D_k = \diag(d^{(k)})$ denote a diagonal matrix with $[D_k]_{i,i} = d^{(k)}_i$. The positive weights $\epsilon d^{(k)}_j$ determine the amount of $j$'s own surplus it injects into an update of $x_j(t+1)$ when $j$ receives a broadcast from node $k$. The parameter $\epsilon > 0$ will be treated as a perturbation parameter in our analysis below, and through this analysis we will obtain: 1) an upper bound on how large $\epsilon$ can be made while still ensuring convergence, as well as 2) an indication of how $\epsilon$ affects the rate of convergence.

Define the $2n \times 2n$ matrix $W_k$ to be
\begin{equation}
W_k = \begin{bmatrix}
I - L_k & \epsilon D_k \\
L_k & S_k - \epsilon D_k
\end{bmatrix}, \label{eqn:W_k}
\end{equation}
where $S_k = I - e_k e_k^T + B_k$. The general broadcast gossip updates \eqref{eqn:updatexj}--\eqref{eqn:updateyi} can be compactly written as
\begin{equation}
\begin{bmatrix}
x(t+1) \\
y(t+1)
\end{bmatrix}
= W(t) \begin{bmatrix}
x(t) \\
y(t)
\end{bmatrix}, \label{eqn:update}
\end{equation}
where $W(t)$ is a random matrix with $W(t) = W_k$ when node $k$ transmits at iteration $t$. In the asynchronous time model, the random matrices $W(t)$, $t = 1, \dots$, are independent and identically distributed, and $W(t) = W_k$ with probability $1/n$ for all $k \in \calV$. 

\section{Convergence in Expectation} \label{sec:convergenceInExpectation}

Next, we focus on identifying properties that the parameters $a_{i,j}$, $b_{i,j}$, $d^{(k)}_j$, and $\epsilon$ must satisfy in order to guarantee that the iterations \eqref{eqn:update} converge in expectation. Since $L_k$ contains some negative entries, $W_k$ is not nonnegative, and so standard results from nonnegative matrix analysis and the study of Markov chains are not sufficient to guarantee convergence in expectation. Our approach will make use of a combination of techniques from the theory of nonnegative matrices and perturbation theory.

Taking the conditional expectation of \eqref{eqn:update} with respect to the random node that broadcasts at each iteration, given the initial values $x(0)$ and $y(0)$, we obtain
\begin{align}
\mathbb{E}\left( \begin{bmatrix} x(t+1) \\ y(t+1)\end{bmatrix} \Bigg| \begin{bmatrix} x(0) \\ y(0)\end{bmatrix}\right) &= \mathbb{E}\left(\prod_{t'=0}^t W(t')\right) \begin{bmatrix} x(0) \\ y(0)\end{bmatrix} \\
&= \overline{W}^{t+1} \begin{bmatrix} x(0) \\ y(0)\end{bmatrix}, \label{eqn:expectedUpdate}
\end{align}
where $\overline{W} = \frac{1}{n} \sum_{k \in \calV} W_k$. One can verify that $W_k [\mathbf{1}^T \ \mathbf{0}^T]^T =  [\mathbf{1}^T \ \mathbf{0}^T]^T$ for all $k \in \calV$ since $L_k \mathbf{1} = \mathbf{0}$, and so $[\mathbf{1}^T \ \mathbf{0}^T]^T$ is also a right eigenvector of $\overline{W}$ corresponding to the eigenvalue $1$. The main result of this section is the following.

\begin{theorem} \label{thm:convergenceInExpectation}
In addition to the constraints \eqref{eqn:constraints_a}, \eqref{eqn:constraints_b}, and \eqref{eqn:constraints_d} imposed on the algorithm parameters above, suppose that $\|B\|_{\infty} \le 1$ or $\|B\|_1 \le 1$. Then under the assumption that $\calG$ is strongly connected (Assumption~\ref{stronglyConnected}), there exists a value $\eta > 0$ such that if $\epsilon \in (0,\eta]$, then $1$ is a simple eigenvalue of $\overline{W}$ with corresponding left eigenvector $[w_1^T \ w_2^T]^T$ normalized such that $[w_1^T \ w_2^T] [\mathbf{1}^T \ \mathbf{0}^T]^T = w_1^T \mathbf{1} = 1$, and
\begin{equation}
\lim_{t \rightarrow \infty} \mathbb{E}\left( \begin{bmatrix} x(t) \\ y(t)\end{bmatrix} \Bigg| \begin{bmatrix} x(0) \\ y(0)\end{bmatrix}\right) =  \begin{bmatrix} \big(w_1^T x(0)\big) \mathbf{1} \\ \mathbf{0} \end{bmatrix}.
\end{equation}
\end{theorem}

\begin{remark} \label{remark1}
As a consequence of Theorem~\ref{thm:convergenceInExpectation}, the broadcast gossip updates \eqref{eqn:updatexj}--\eqref{eqn:updateyi} will converge to the average consensus if and only if $w_1 = \frac{1}{n} \mathbf{1}$. From the expression for $\overline{W}$ derived below (see eqn.~\eqref{eqn:Wbar}), it turns out that this is only possible if $w_2 = \frac{1}{n} \mathbf{1}$ and $\mathbf{1}^T B = \mathbf{1}^T$.
\end{remark}

The rest of this section is devoted to the proof of Theorem~\ref{thm:convergenceInExpectation}.

\subsection{Preliminaries and The Plan}

From \eqref{eqn:W_k}, we find that the expected update matrix $\overline{W}$ has the form
\begin{equation}
\overline{W} = 
\underbrace{\begin{bmatrix} 
I - \overline{L} & 0 \\
\overline{L} & \overline{S}
\end{bmatrix}}_{\stackrel{\text{def}}{=} W_0}
+ \epsilon 
\underbrace{\begin{bmatrix}
0 & \overline{D} \\
0 & -\overline{D}
\end{bmatrix}}_{\stackrel{\text{def}}{=} E} , \label{eqn:Wbar}
\end{equation}
where, recalling that $L = \diag(A \mathbf{1}) - A$ is the Laplacian of the graph with weighted adjacency matrix $A$, we have
\begin{align}
\overline{L} &= \frac{1}{n} \sum_{k \in \calV} L_k = \frac{1}{n} L \\
\overline{D} &= \frac{1}{n} \sum_{k \in \calV} D_k = \diag\left(\sum_{k \in \calV} d^{(k)}\right) \\
\overline{S} &= \frac{1}{n} \sum_{k \in \calV} S_k = \left(1 - \frac{1}{n}\right) I + \frac{1}{n}B.
\end{align}
Using the expression \eqref{eqn:Wbar} for $\overline{W}$, one can verify the statement made in Remark~\ref{remark1} above.

From \eqref{eqn:Wbar}, it is evident that $\overline{W}$, can be viewed as a perturbed version of the matrix $W_0$. The proof of Theorem~\ref{thm:convergenceInExpectation} involves first characterizing the eigenvalues of $W_0$ using concepts from the theory of nonnegative matrices. Then results from perturbation theory can be used to determine the eigenvalues of $\overline{W}$ as a function of $\epsilon$ and the eigenvalues of $W_0$ and $E$. Before proceeding, we briefly review background material from nonnegative matrix theory and perturbation theory.

\subsection{Background}

Recall that a matrix $F$ is called \emph{nonnegative} if all of its entries are greater than or equal to zero. A square nonnegative matrix $F$ is \emph{primitive} if there exists a positive integer $k$ such that all entries of $F^k$ are strictly positive. If $F$ corresponds to the weighted adjacency matrix of a strongly connected graph, then it is irreducible and thus primitive~\cite{ME01}.

Next we recall some definitions and results from perturbation theory~\cite{SE04}.

\begin{lemma}[\cite{SE04} Sec.~2.4] \label{property1}
Suppose that a matrix $F(\epsilon)$ is continuously differentiable (entry-wise) with respect to the perturbation parameter $\epsilon$. Then the eigenvalues of $F(\epsilon)$ are continuous functions of $\epsilon$.
\end{lemma}

An eigenvalue of $F(\epsilon)$ is called \emph{stable} if it does not depend on $\epsilon$. An eigenvalue of $F(\epsilon)$ is called \emph{semi-simple} if its algebraic multiplicity is equal to its geometric multiplicity.

\begin{lemma}[\cite{SE04} Sec.~2.8] \label{property2}
Suppose that $F(\epsilon) = F_0 + \epsilon E$. Let $\lambda_0$ be a semi-simple double eigenvalue of $W_0$ with corresponding left eigenvectors $u_1$ and $u_2$ and right eigenvectors $v_1$ and $v_2$ normalized such that $u_1^T v_1 = 1$ and $u_2 ^T v_2 = 1$. Then, for $\epsilon > 0$, the eigenvalue $\lambda_0$ bifurcates into two distinct eigenvalues $\lambda_{0,1}(\epsilon)$ and $\lambda_{0,2}(\epsilon)$ of $F(\epsilon)$. The bifurcation is given in the form of the power series
\begin{align}
\lambda_{0,1}(\epsilon) &= \lambda_0 + \epsilon \lambda' + o(\epsilon) \\
\lambda_{0,2}(\epsilon) &= \lambda_0 + \epsilon \lambda'' + o(\epsilon),
\end{align}
where $\lambda'$ and $\lambda''$ are the eigenvalues of the $2 \times 2$ matrix,
\begin{equation}
\begin{bmatrix}
u_1^T E v_1 & u_1^T E v_2 \\
u_2^T E v_1 & u_2^T E v_2
\end{bmatrix}. \label{eqn:property2}
\end{equation}
\end{lemma}

\subsection{Eigenvalues of $W_0$} \label{sec:eigenvaluesW_0}

In order to apply the perturbation results mentioned above, we need to first identify the eigenvalues of $W_0$. Observe, from \eqref{eqn:Wbar}, that $W_0$ is block diagonal, and so the eigenvalues of $W_0$ are the collective eigenvalues of $I - \overline{L}$ and $\overline{S}$.

\begin{lemma} \label{lemma1}
The matrix $I - \overline{L}$ is primitive, its largest eigenvalue is $1$, and all other eigenvalues of $I - \overline{L}$ have moduli strictly less than $1$.
\end{lemma}

\begin{proof}
Recall that $a_{i,j} \in [0,1]$ according to the constraints \eqref{eqn:constraints_a}, and so $I - \overline{L}$ is a square non-negative matrix. Moreover, all diagonal elements of $I - \overline{L}$ are strictly positive. Since $\calG$ is strongly connected, it follows that $I - \overline{L}$ is irreducible, and combining these facts gives that $I - \overline{L}$ is primitive. Since $I - \overline{L}$ is primitive and $(I - \overline{L})\mathbf{1} = \mathbf{1}$, we have $\rho(I - \overline{L}) \ge 1$. Also, its spectral radius is bounded by $\rho(I - \overline{L}) \le \| I - \overline{L} \|_{\infty} = 1$. Thus its largest eigenvalue is $1$ and, according to the Perron-Frobenius Theorem~\cite{Seneta}, all other eigenvalues of $I - \overline{L}$ are strictly less than $1$.
\end{proof}

Based on Lemma~\ref{lemma1} we know that $W_0$ has at least one eigenvalue equal to $1$. Next we need to determine the eigenvalues of $\overline{S} = (1 - \frac{1}{n}) I + \frac{1}{n}B$. If $\lambda_i(B)$ is an eigenvalue of $B$, then $\lambda_i(\overline{S}) = 1 - \frac{1}{n} + \frac{1}{n}\lambda_i(B)$ is an eigenvalue of $\overline{S}$, and so the real task is to characterize the eigenvalues of $B$. If all eigenvalues of $B$ have magnitude less than $1$, then all eigenvalues of $\overline{S}$ are also less than $1$, and so $1$ is a simple eigenvalue of $W_0$. On the other hand, if $1$ is an eigenvalue of $B$ then it is also an eigenvalue of $\overline{S}$, in which case $1$ is a multiple eigenvalue of $W_0$.

Under the assumptions of Theorem~\ref{thm:convergenceInExpectation} we have that $\|B\|_\infty \le 1$ or $\|B \|_1 \le 1$. Since $\rho(B) \le \min\{\|B\|_\infty, \|B\|_1\}$, it follows that the largest eigenvalue of $B$ is no larger than $1$. Moreover, it follows from Assumption~1 and \eqref{eqn:constraints_b} that $B$ corresponds to the weighted adjacency matrix of a strongly connected digraph, and hence $B$ is primitive. Thus, $\overline{S} = (1 - \frac{1}{n})I + \frac{1}{n}B$ is also primitive and its diagonal entries are all positive. Then, from the Perron-Frobenius Theorem, the largest eigenvalue $\lambda_1(B)$ of $B$ is simple and all other eigenvalues of $B$ have magnitude strictly less than $\lambda_1(B)$. 

It turns out that under the condition $\| B \|_\infty \le 1$ or $\| B \|_1 \le 1$ there are two possible cases: either $\lambda_1(\overline{S}) < 1$ or $\lambda_1(\overline{S}) = 1$. These cases are captured by the two following lemmas.

\begin{lemma} \label{lemma4}
Suppose that $B$ is either row stochastic, column stochastic, or doubly stochastic. Then $1$ is a simple eigenvalue of $\overline{S}$ and all other eigenvalues of $\overline{S}$ have moduli strictly less than $1$.
\end{lemma}

Lemma~\ref{lemma4} follows from standard arguments in the theory of nonnegative matrices~\cite{ME01,Seneta}.

\begin{lemma} \label{lemma3}
Suppose that either
\begin{equation}
\max_j \sum_{i=1}^n b_{i,j} \le 1 \quad \text{ and } \quad \min_j \sum_{i=1}^n b_{i,j} < 1, \label{eqn:lemma3condition1}
\end{equation}
or
\begin{equation}
\max_i \sum_{j=1}^n b_{i,j} \le 1 \quad \text{ and } \quad \min_i \sum_{j=1}^n b_{i,j} < 1. \label{eqn:lemma3condition2}
\end{equation}
Then the modulus of the largest eigenvalue of $\overline{S}$ is strictly less than $1$.
\end{lemma}

\begin{proof}
We will prove the lemma for the case where \eqref{eqn:lemma3condition1} holds; the proof for the other case follows by a similar argument where $\overline{S}$ is replaced with $\overline{S}^T$. If $\max_j \sum_{i=1}^n b_{i,j} < 1$ then the claim follows since
\begin{equation}
\rho(\overline{S}) \le \|\overline{S}\|_1 = 1 - \frac{1}{n} + \frac{1}{n} \max_j \sum_{i=1}^n b_{i,j} < 1.
\end{equation}
If
\begin{equation}
\min_j \sum_{i=1}^n b_{i,j} < \max_j \sum_{i=1}^n b_{i,j} = 1, \label{eqn:lemma3assumption}
\end{equation}
then suppose, to arrive at a contradiction, that $\rho(\overline{S}) = 1$.
Let $u$ denote the right eigenvector of $\overline{S}$ for which $\overline{S} u = u$, and let $u$ be normalized such that $\mathbf{1}^T u = 1$. Then $Bu = u$ also, and summing over the entries of $u$ we find that
\begin{align}
\sum_{i=1}^n u_i &= \sum_{i=1}^n \left(\sum_{j=1}^n b_{i,j} u_j\right) \\
&= \sum_{j=1}^n \left(\sum_{i=1}^n b_{i,j} \right) u_j \\
&< \sum_{j=1}^n u_j,
\end{align}
where the inequality follows from the assumption \eqref{eqn:lemma3assumption}. Since this is a contradiction, it must be true that all eigenvalues of $\overline{S}$ are strictly less than $1$ when either \eqref{eqn:lemma3condition1} or \eqref{eqn:lemma3condition2} holds.
\end{proof}

To summarize, in this subsection we have shown that $1$ is an eigenvalue of $W_0$ with multiplicity at least $1$. If $B$ is sub-stochastic (i.e., if elements of some rows and columns sum to a value less than $1$), then $1$ is a simple eigenvalue of $W_0$. On the other hand if $B$ is a (row-, column-, or doubly) stochastic matrix, then $1$ is an eigenvalue of $W_0$ with multiplicity $2$. In the next subsection we apply tools from perturbation theory to characterize the eigenvalues of $\overline{W}$.

\subsection{Perturbation Analysis}

The proof that the broadcast gossip iterations converge in expectation hinges on showing that $1$ is a simple eigenvalue of $\overline{W}$ with appropriate corresponding eigenvectors. When $1$ is a simple eigenvalue of $\overline{W}$ then the analysis is straightforward. 

\begin{proposition} \label{prop1}
Suppose that either of the conditions \eqref{eqn:lemma3condition1} or \eqref{eqn:lemma3condition2} of Lemma~\ref{lemma3} hold. Then there exists a number $\eta > 0$ such that if $\epsilon \in (0, \eta]$ then $1$ is a simple eigenvalue of the matrix $\overline{W} = W_0 + \epsilon E$ in \eqref{eqn:Wbar} and the moduli of all other eigenvalues of $\overline{W}$ are strictly less than $1$.
\end{proposition}

\begin{proof}
From the discussion in Section~\ref{sec:eigenvaluesW_0}, we know that the eigenvalues of $\overline{W}$ are the collective eigenvalues of $W_0$ and $\overline{S}$. Under the conditions of the proposition, Lemmas~\ref{lemma1} and \ref{lemma3} provide that $1$ is a simple eigenvalue of $W_0$ and the moduli of all other eigenvalues of $W_0$ are strictly less than $1$. Observe that $\overline{W} = W_0 + \epsilon E$ depends on $\epsilon$ in a continuous manner, and so the eigenvalues of $\overline{W}$ are continuous functions of the perturbation parameter $\epsilon$. Therefore, there exists $\eta > 0$ so that $1$ is a simple eigenvalue of $\overline{W}$ and the moduli of all other eigenvalues of $\overline{W}$ are strictly less than $1$ when $\epsilon \in (0, \eta]$.
\end{proof}

When $1$ is a double eigenvalue of $\overline{W}$ we arrive at the same conclusion, but the proof requires a bit more effort.

\begin{proposition} \label{prop2}
Suppose that $B$ is either row stochastic, column stochastic, or doubly stochastic. Then there exists a number $\eta > 0$ such that if $\epsilon \in (0, \eta]$ then $1$ is a simple eigenvalue of the matrix $\overline{W}$ and the moduli of all other eigenvalues of $\overline{W}$ are strictly less than $1$.
\end{proposition}

\begin{proof}
The proof follows from a generalization of an argument in~\cite{CA12}. Under the conditions of the proposition, Lemmas~\ref{lemma1} and \ref{lemma4} provide that $1$ is an eigenvalue of $W_0$ with multiplicity $2$. One can verify that $1$ is a semi-simple eigenvalue of $W_0$ since there exist two linearly independent right eigenvectors $u_1$ and $u_2$, with corresponding linearly independent left eigenvectors $v_1$ and $v_2$. These eigenvectors are given by
\begin{equation}
u_1 = \begin{bmatrix} \mathbf{0} \\ \frac{1}{w^T q} q\end{bmatrix}, \quad
u_2 = \begin{bmatrix} \mathbf{1} \\ \frac{-1}{w^T q} q\end{bmatrix}, \quad
v_1 = \begin{bmatrix} w \\ w\end{bmatrix}, \quad
v_2 = \begin{bmatrix} p \\ \mathbf{0}\end{bmatrix},
\end{equation}
where $w$ is the eigenvector of $\overline{S}$ satisfying $w^T \overline{S} = w^T$ and $w^T \mathbf{1} = 1$, and where $p$ and $q$ are the left and right eigenvectors of $I - \overline{L}$ corresponding to the eigenvalue $1$, normalized so that $p^T \mathbf{1} = 1$ and $\mathbf{1}^T q = 1$. Note that these eigenvectors exist as a consequence of Lemmas~\ref{lemma1} and \ref{lemma4}. One can verify that
\begin{align}
v_1^T u_1 = v_2^T u_2 = 1, \quad
v_1^T u_2 = v_2^T u_1 = 0.
\end{align}
Also note that all entries of the vectors $w$, $p$, and $q$ are positive since $I - \overline{L}$ and $\overline{S}$ are primitive matrices.

Applying Lemma~\ref{property2} with these values, we find that the semi-simple eigenvalue $1$ bifurcates into two eigenvalues
\begin{align}
\lambda_{1,1}(\epsilon) &= 1 + \epsilon \lambda' + o(\epsilon) \\
\lambda_{1,2}(\epsilon) &= 1 + \epsilon \lambda'' + o(\epsilon),
\end{align}
where $\lambda'$ and $\lambda''$ are the eigenvalues of the matrix
\begin{equation}
\begin{bmatrix}
v_1^T E u_1 & v_1^T E u_2 \\
v_2^T E u_1 & v_2^T E u_2
\end{bmatrix}
=
\begin{bmatrix}
0 & 0 \\
\frac{1}{w^T q} p^T \overline{D} q & \frac{-1}{w^T q} p^T \overline{D} q
\end{bmatrix}.
\end{equation}
Thus, we have $\lambda' = 0$. Also, since $p$, $q$, $w$, and the diagonal entries of $\overline{D}$ are all strictly positive, we have $\lambda'' = \frac{-1}{w^T q} p^T \overline{D} q < 0$. It follows that $\lambda_{1,1} = 1$ is the stable eigenvalue of $\overline{W}$ corresponding to the right eigenvector $[\mathbf{1}^T \ \mathbf{0}^T]^T$. Moreover, for sufficiently small $\epsilon$, we have $\lambda_{1,2}(\epsilon) < 1$ since $d \lambda_{1,2}(\epsilon) / d \epsilon \vert_{\epsilon = 0} = \lambda'' < 0$. Therefore, there must exist a positive constant $\eta_1$ so that $|\lambda_{1,2}(\epsilon)| < |\lambda_{1,1}(\epsilon)|$ when $\epsilon \in (0, \eta_1]$. In addition. the eigenvalues of $\overline{W}$ are continuous functions of $\epsilon$, so the moduli of $\lambda_{1,1}(\epsilon) = \lambda_1(\overline{W})$ and $\lambda_{1,2}(\epsilon) = \lambda_{2}(\overline{W})$ will dominate the moduli of all other eigenvalues of $\overline{W}$ provided that $\epsilon > 0$ is sufficiently small; i.e., there exists an $\eta_2 > 0$ such that $\max_{i \ne 1,2} |\lambda_i(\overline{W})| < 1$ when $\epsilon \in (0, \eta_2]$. Therefore, when $\epsilon \in (0, \min\{\eta_1, \eta_2\}]$, then $1$ is a simple eigenvalue of $\overline{W}$ and the moduli of all other eigenvalues are strictly less than $1$.
\end{proof}

We are now ready to complete the proof of Theorem~\ref{thm:convergenceInExpectation}.

\subsection{Proof of Theorem~\ref{thm:convergenceInExpectation}}

Suppose that $1$ is a simple eigenvalue of a matrix $\overline{W}$ and the moduli of all other eigenvalues are strictly less than $1$. Let $u$ and $v$ denote the left and right eigenvectors of $\overline{W}$ corresponding to the eigenvalue $1$, normalized so that $u^T v = 1$. Then it is known~\cite{CA12,XI04} that $\lim_{t \rightarrow \infty} \overline{W}^t = u v^T$.

From Propositions~\ref{prop1} and \ref{prop2} we know that, under the conditions of Theorem~\ref{thm:convergenceInExpectation}, $1$ is a simple eigenvalue of $\overline{W}$ with corresponding right eigenvector $[\mathbf{1}^T \ \mathbf{0}^T]^T$, and the moduli of all other eigenvalues of $\overline{W}$ are strictly less than $1$. Let $[w_1^T \ w_2^T]^T$ be the left eigenvector of $\overline{W}$ satisfying
\begin{equation}
[w_1^T \ w_2^T] \overline{W} = [w_1^T \ w_2^T],
\end{equation}
normalized such that $[w_1^T \ w_2^T] [\mathbf{1}^T \ \mathbf{0}^T]^T = 1$. Then the expected broadcast gossip updates \eqref{eqn:expectedUpdate} converge to a limit
\begin{align}
\lim_{t \rightarrow \infty} \mathbb{E} \left( \begin{bmatrix} x(t) \\ y(t) \end{bmatrix} \Bigg| \begin{bmatrix} x(0) \\ y(0) \end{bmatrix} \right) &= \lim_{t \rightarrow \infty} \overline{W}^t \begin{bmatrix} x(0) \\ y(0) \end{bmatrix} \\
&= \begin{bmatrix} \mathbf{1} \\ \mathbf{0} \end{bmatrix} [w_1^T \ w_2^T] \begin{bmatrix} x(0) \\ y(0) \end{bmatrix} \\
&= \big(w_1^T x(0)\big) \begin{bmatrix} \mathbf{1} \\ \mathbf{0} \end{bmatrix},
\end{align}
where the last line follows since $y(0) = \mathbf{0}$. This completes the proof of Theorem~\ref{thm:convergenceInExpectation}.

\section{Convergence in the Second Moment} \label{sec:convergenceSecondMoment}

The previous section dealt with convergence in expectation. Next we present a general condition for convergence in the second moment of the broadcast gossip algorithm described in Section~\ref{sec:framework}.

\begin{theorem} \label{thm:secondMoment}
Suppose that $B$ is a (row-, column-, or doubly) stochastic matrix. Let $v \in \mathbb{R}^n$ be the vector satisfying $v^T B = v^T$ normalized such that $v^T \mathbf{1} = 1$. The sequence of vectors $\{x(t), y(t)\}_{t=1}^\infty$ generated by the broadcast gossip updates \eqref{eqn:updatexj}--\eqref{eqn:updateyi} satisfy
\begin{equation}
\lim_{t \rightarrow \infty} \mathbb{E}\left[ \left\lVert \begin{bmatrix} x(t) \\ y(t)\end{bmatrix} - \begin{bmatrix} (v^T x(0)) \mathbf{1} \\ \mathbf{0}\end{bmatrix} \right\rVert^2_2 \ \Bigg \vert \ \begin{bmatrix} x(0) \\ y(0) \end{bmatrix}\right] = 0 \label{eqn:secondMomentConvergence}
\end{equation}
if and only if 
\begin{equation}
\rho\left(\mathbb{E}[W(t) \otimes W(t)] - \left(\begin{bmatrix} \mathbf{1} \\ \mathbf{0} \end{bmatrix} \otimes \begin{bmatrix} \mathbf{1} \\ \mathbf{0} \end{bmatrix}\right) \left(\begin{bmatrix} v \\ v \end{bmatrix} \otimes \begin{bmatrix} v \\ v \end{bmatrix}\right)^T\right) < 1. \label{eqn:secondMomentCondition}
\end{equation}
\end{theorem}

\begin{remark}
Theorem~\ref{thm:secondMoment} can be viewed as generalizing the convergence conditions for linear iterations described in~\cite{XI04,BO10} to update matrices which have the form \eqref{eqn:W_k}.
\end{remark}

\begin{proof}
We first prove that \eqref{eqn:secondMomentCondition} implies \eqref{eqn:secondMomentConvergence}. Let $z(t) = [x(t)^T \ y(t)^T]^T$ and define the error vector $m(t) = z(t) - J z(0)$, where
\begin{equation}
J = \begin{bmatrix} \mathbf{1} \\ \mathbf{0}\end{bmatrix} \begin{bmatrix} v^T & v^T \end{bmatrix}.
\end{equation}
Observe that $W_k J = J$ for all $k \in \calV$, since $L_k \mathbf{1} = \mathbf{0}$, and therefore $m(t+1) = W(t) m(t)$. Let $M(t) = m(t) m(t)^T$. Then $M(t+1) = W(t) M(t) W(t)^T$. Construct a vector $\widetilde{m}(t) \in \mathbb{R}^{4n^2}$ by stacking the elements of $M(t)$ column-wise, and observe that
\begin{align}
\mathbb{E}[\widetilde{m}(t) | \widetilde{m}(0)] &= \prod_{s=0}^{t-1} \mathbb{E}[W(s) \otimes W(s)] \ \widetilde{m}(0) \\
&= \mathbb{E}[W(1) \otimes W(1)]^t \ \widetilde{m}(0),
\end{align}
since the matrices $W(t)$ are independent and identically distributed. 

Under the assumption that $v^T B = v^T$, one can verify that $[v^T \ v^T] \mathbb{E}[W(t)] = [v^T \ v^T]$. Note that such a vector exists since $B$ is primitive\footnote{To see why, recall that $\calG$ is strongly connected (Assumption~1).} and either row or column stochastic (by assumption). We have also seen that $\mathbb{E}[W(t)] [\mathbf{1}^T \ \mathbf{0}^T]^T = [\mathbf{1}^T \ \mathbf{0}^T]^T$. It follows that $([v^T \ v^T]^T \otimes [v^T \ v^T]^T)$ and $([\mathbf{1}^T \ \mathbf{0}^T]^T \otimes [\mathbf{1}^T \ \mathbf{0}^T]^T)$ are the left and right eigenvectors of $\mathbb{E}[W(t) \otimes W(t)]$ corresponding to the eigenvalue $1$. If assumption \eqref{eqn:secondMomentCondition} holds, then we have
\begin{align}
\lim_{t \rightarrow \infty} \mathbb{E}[\widetilde{m}(t) | \widetilde{m}(0)] &= \left(\begin{bmatrix} \mathbf{1} \\ \mathbf{0} \end{bmatrix} \otimes \begin{bmatrix} \mathbf{1} \\ \mathbf{0} \end{bmatrix}\right) \left(\begin{bmatrix} v \\ v \end{bmatrix} \otimes \begin{bmatrix} v \\ v \end{bmatrix}\right)^T \widetilde{m}(0) \\
&= 0, \label{eqn:equalsZero}
\end{align}
where the last equality follows because $v$ is orthogonal to $m(0)$ and hence the vector $([v^T \ v^T]^T \otimes [v^T \ v^T]^T)$ is orthogonal to $\widetilde{m}(0)$.

Observe that \eqref{eqn:equalsZero} implies that $\mathbb{E}[m_i(t)^2] \rightarrow 0$ for all $i$, and therefore
\begin{equation}
\mathbb{E}\big[m(t)^T m(t) \big| m(0)\big] = \sum_{i=1}^{2n} \mathbb{E}[m_i(t)^2] \rightarrow 0,
\end{equation}
which gives us \eqref{eqn:secondMomentConvergence}.

Next, to see that \eqref{eqn:secondMomentConvergence} implies \eqref{eqn:secondMomentCondition}, observe that if \eqref{eqn:secondMomentConvergence} holds then it must be true that $\mathbb{E}[m_i(t)^2] \rightarrow 0$ for all $i=1,\dots,2n$. By the Cauchy-Schwarz inequality, we have
\begin{equation}
\mathbb{E}[m_i(t) m_j(t)]^2 \leq \mathbb{E}[m_i(t)^2] \cdot \mathbb{E}[m_j(t)^2] \rightarrow 0.
\end{equation}
Therefore, each entry in the matrix $\mathbb{E}[m(t) m(t)^T]$ tends to $0$ as $t \rightarrow \infty$, independent of $m(0)$, which implies that \eqref{eqn:secondMomentCondition} must hold.
\end{proof}

\section{Unbiased Broadcast Gossip} \label{sec:ubga}

This section proposes a particular choice of values for the parameters $a_{j,k}$, $b_{j,k}$, and $d^{(k)}_j$, corresponding to a particular family of broadcast gossip algorithms. For the choice considered in this section, we guarantee that the broadcast gossip updates \eqref{eqn:updatexj}--\eqref{eqn:updateyi} converge to the average consensus. For this reason we refer to these as \emph{unbiased broadcast gossip algorithms} (UBGA).

Recall that $\mathcal{N}_j^-$ denotes the set of out-neighbors of node $j$ and $\mathcal{N}_j^+$ denotes the in-neighbors of node $j$. Let $|\mathcal{N}|$ denote the cardinality of the set $\mathcal{N}$. Let $\delta_j^- = |\mathcal{N}_j^-|$ denote the out-degree of node $j$, and let $\delta_j^+ = |\mathcal{N}_j^+|$ denote its in-degree.

Unbiased broadcast gossip algorithms are obtained by setting
\begin{align}
d^{(k)}_j &= \begin{cases} 1 / \delta_j^+ & \text{ if } j \in \mathcal{N}_k^- \\ 0 & \text{ otherwise } \end{cases} \label{eqn:ubga_d}\\
B_{j,k} &= \begin{cases} 1 / \delta_k^- & \text{ if } j \in \mathcal{N}_k^- \\ 0 & \text{ otherwise,} \end{cases} \label{eqn:ubga_B}
\end{align}
and taking $A_{j,k}$ to be any values which satisfy the constraints \eqref{eqn:constraints_a}. In order to implement such a protocol, each node $j \in \mathcal{N}_k^-$ that receives messages from $k$ needs to know $\delta_k^-$, the out-degree of the broadcasting node $k$. If $k$ knows its out-degree (i.e., the number of neighbors that receive its broadcasts) then this can be accomplished by having $k$ broadcast the value of $\delta_k^-$ to all nodes in $\mathcal{N}_k^-$. If $\calG$ is undirected\footnote{I.e., $(i,j) \in \calE$ if and only if $(j,i) \in \calE$}, as is assumed in~\cite{AY09}, then $\delta_k^- = \delta_k^+$, and so it is reasonable for $k$ to know its out-degree. On the other hand, in a general directed graph it may be difficult or impractical for $k$ to know its out degree since $k$ may not receive messages directly from all nodes $j \in \mathcal{N}_k^-$. In Section~\ref{sec:bbga} below we describe and analyze an alternative algorithm which does not require knowledge of $\delta_k^-$, but for which we are not guaranteed to achieve consensus on the average. First we discuss theoretical convergence guarantees for UBGA.

First, note that for $\epsilon$ sufficiently small, UBGA asymptotically converges in expectation to the average consensus in the sense that $\mathbb{E}[x(t)] \rightarrow \frac{1}{n} \mathbf{1} \mathbf{1}^T x(0)$ and $\mathbb{E}[y(t)] \rightarrow \mathbf{0}$. To see why, observe that for the choice of parameters given in \eqref{eqn:ubga_B} we have $\mathbf{1}^T B = \mathbf{1}^T$. Therefore, by Theorem~\ref{thm:convergenceInExpectation} (see also Remark~1), there exists an $\eta > 0$ such that UBGA converges in expectation to the average consensus for $\epsilon \in (0, \eta]$. It turns out that UBGA also converges to the average consensus solution in the mean-squared sense.

\begin{proposition} \label{prop:ubga}
Let the parameters $d^{(k)}_j$ and $B_{j,k}$ be chosen as in \eqref{eqn:ubga_d} and \eqref{eqn:ubga_B}, and take $v = \frac{1}{n}\mathbf{1}$. Then there exists a constant $\eta > 0$ such that if $\epsilon \in (0, \eta]$ then~\eqref{eqn:secondMomentCondition} holds, and so 
\begin{equation}
\lim_{t \rightarrow \infty} \mathbb{E}\left[ \left\lVert \begin{bmatrix} x(t) \\ y(t)\end{bmatrix} - \begin{bmatrix} \frac{1}{n}\left(\mathbf{1}^T x(0)\right) \mathbf{1} \\ \mathbf{0}\end{bmatrix} \right\rVert^2_2 \ \Bigg \vert \ \begin{bmatrix} x(0) \\ y(0) \end{bmatrix}\right] = 0.
\end{equation}
\end{proposition}

\begin{proof}
To prove the theorem, we need to show that there exists an $\eta > 0$ such that \eqref{eqn:secondMomentCondition} holds. Convergence in the second moment then follows from Theorem~\ref{thm:secondMoment}.

To show that \eqref{eqn:secondMomentCondition} holds for sufficiently small $\epsilon$, we use a perturbation argument similar to one in~\cite{CA12}. Since $(I - \mathbb{E}[L_k]) \mathbf{1} = \mathbf{1}$, the matrix $I-\mathbb{E}[L_k]$ has $1$ as an eigenvalue with right eigenvector $\mathbf{1}$. It follows that there exists a corresponding left eigenvector $w$ satisfying $w^T (I - \mathbb{E}[L_k]) = w^T$ and $w^T \mathbf{1} = 1$, and all entries of $w$ are positive. Consequently, the following four equalities hold:
\begin{align}
\mathbb{E}[(I - L_k) \otimes (I - L_k)] \ ( \mathbf{1} \otimes \mathbf{1}) &= ( \mathbf{1} \otimes \mathbf{1}) \\
(w \otimes \mathbf{1})^T \ \mathbb{E}[(I - L_k) \otimes S_k] &= (w \otimes \mathbf{1})^T \\
(\mathbf{1} \otimes w)^T \ \mathbb{E}[S_k \otimes (I - L_k)] &= (\mathbf{1} \otimes w)^T \\
(\mathbf{1} \otimes \mathbf{1})^T \ \mathbb{E}[S_k \otimes S_k] &= (\mathbf{1} \otimes \mathbf{1})^T.
\end{align}
Thus, the four matrices $\mathbb{E}[(I - L_k) \otimes (I - L_k)]$, $\mathbb{E}[(I - L_k) \otimes S_k]$, $\mathbb{E}[S_k \otimes (I - L_k)]$, and $\mathbb{E}[S_k \otimes S_k]$ have eigenvalue $1$. These matrices are all non-negative and irreducible under the constraints \eqref{eqn:constraints_a} and since $\calG$ is strongly connected. Moreover, since the corresponding eigenvectors $( \mathbf{1} \otimes \mathbf{1})$, $(w \otimes \mathbf{1})^T$, $(\mathbf{1} \otimes w)^T$, and $(\mathbf{1} \otimes \mathbf{1})^T$ are all positive, it follows from the Perron-Frobenius Theorem that $1$ is the largest eigenvalue of the four matrices above, and all other eigenvalues have moduli strictly less than $1$.

In order to apply a perturbation argument to analyze the eigenvalues of $\mathbb{E}[W(t) \otimes W(t)]$, let us write $W_k = M_k + \epsilon F_k$ with
\begin{align}
M_k &= \begin{bmatrix} I - L_k & 0 \\ L_k & S_k\end{bmatrix} \\
F_k &= \begin{bmatrix} 0 & D_k \\ 0 & -D_k \end{bmatrix}.
\end{align}
Similarly, we have $W(t) = M(t) + \epsilon F(t)$, with $M(t)$ and $F(t)$ being random matrices drawn from the collection $\{(M_k, F_k)\}$ depending on which node $k$ broadcasts at iteration $t$. In the following we omit the dependence on $t$ to simplify notation. With the above definitions, we have $\mathbb{E}[W \otimes W] = \mathbb{E}[M \otimes M] + \epsilon \mathbb{E}[M \otimes F + F \otimes M + F \otimes \epsilon F]$. There is a $2n \times 2n$ permutation matrix $P$ such that $P^T \mathbb{E}[W \otimes W] P = \widehat{M} + \epsilon \widehat{F}$ where
\begin{align}
\widehat{M} &= \mathbb{E} \begin{bmatrix} (I-L) \otimes (I-L) & 0 & 0 & 0\\
(I-L) \otimes L & (I-L) \otimes S & 0 & 0\\
L \otimes (I-L) & 0 & S \otimes (I-L) & 0\\
L \otimes L & L \otimes S & S \otimes L & S \otimes S \end{bmatrix} \\
\widehat{F} &= \mathbb{E} \begin{bmatrix} 0 & (I-L) \otimes D & D \otimes (I-L) & D \otimes \epsilon D \\
0 & -(I-L) \otimes D & D \otimes L & D \otimes (S - \epsilon D) \\
0 & L \otimes D & -D \otimes (I-L) & (S - \epsilon D) \otimes D \\
0 & -L \otimes D & -D \otimes L & -S \otimes D - D \otimes S + D \otimes \epsilon D \end{bmatrix}.
\end{align}
Since $1$ is a simple eigenvalue of each of the matrices $\mathbb{E}[(I - L) \otimes (I - L)]$, $\mathbb{E}[(I - L) \otimes S]$, $\mathbb{E}[S \otimes (I - L)]$, and $\mathbb{E}[S \otimes S]$, and all other eigenvalues of these matrices have moduli less than $1$, we find that $1$ is a semi-simple eigenvalue of $\widehat{M}$ with multiplicity $4$, and all other eigenvalues of $\widehat{M}$ are strictly less than $1$. 

We will use a perturbation argument to show that for sufficiently small $\epsilon > 0$, the largest eigenvalue of $\widehat{M} + \epsilon \widehat{F}$ is $1$ and $1$ is a simple eigenvalue. Our argument is based on a generalization of Lemma~\ref{property2} that addresses bifurcation of a quadruple semi-simple eigenvalues rather than double semi-simple eigenvalues~\cite{SE04}.

Let $\widehat{\lambda}_i(\epsilon)$, $i=1,\dots,4$ denote the four bifurcating eigenvalues of $\widehat{M} + \epsilon \widehat{F}$. Similar to Lemma~\ref{property2}, we have $\widehat{\lambda}_i(\epsilon) = 1 + \epsilon \xi_i + o(\epsilon)$, where $\xi_i$, $i=1,\dots,4$ are four eigenvalues of the matrix of similar structure to \eqref{eqn:property2}. Solving for $\xi_1,\dots,\xi_4$, we find that the derivatives of the eigenvalues $\widehat{\lambda}_i(\epsilon)$ with respect to $\epsilon$ are given by
\begin{align}
&\frac{d \widehat{\lambda}_1(\epsilon)}{d \epsilon} = 0 \\
&\frac{d \widehat{\lambda}_2(\epsilon)}{d \epsilon} = \frac{d \widehat{\lambda}_3(\epsilon)}{d \epsilon} = -N v_1^T \mathbb{E}[D] v_2 < 0 \\
&\frac{d \widehat{\lambda}_4(\epsilon)}{d \epsilon} = - 2N v_1^T \mathbb{E}[D] v_2 < 0,
\end{align}
where $v_1$ is the positive left eigenvector of $I - L$ normalized such that $v^T \mathbf{1} = 1$, and $v_2$ is the positive right eigenvector of $S$ normalized so that $\mathbf{1}^T v_2 = 1$. Similar as in Lemma~\ref{property2} (see~\cite{SE04}), it follows that there exists a real number $\eta > 0$ such that the matrix $\mathbb{E}[W \otimes W]$ has only one simple eigenvalue $1$, and the moduli of all other eigenvalues are smaller than $1$ when $\epsilon \in (0, \eta]$. Thus, \eqref{eqn:secondMomentCondition} holds, and convergence in the second moment follows from Theorem~\ref{thm:secondMoment}.
\end{proof}

\section{Biased Broadcast Gossip} \label{sec:bbga}

The previous section proposed UBGA, a broadcast gossip algorithm which provably converges to the average consensus in both expectation and in the mean-squared sense. UBGA is practical in situations when the network can be guaranteed to be undirected, or when nodes otherwise know their out-degree. For instance, one could enforce that only symmetric links are used by having each node broadcast its set of in-neighbors and then only updating using messages from neighbors for which the neighborhood relationship is symmetric. However, this may be undesirable in some applications, and so in this section we consider an alternative family of broadcast gossip algorithms. These algorithms are no longer guaranteed to converge to an average consensus, and so we refer to them as \emph{biased broadcast gossip algorithms} (BBGAs). However, we still guarantee convergence in expectation and in the mean-squared sense to a characterizable value which depends on the initial state at each node and the structure of the network.

Biased broadcast gossip algorithms are obtained by setting 
\begin{align}
d^{(k)}_j &= \begin{cases} 1 / \delta_j^+ & \text{ if } j \in \mathcal{N}_k^- \\ 0 & \text{ otherwise } \end{cases} \label{eqn:bbga_d}\\
B_{j,k} &= \begin{cases} 1 / \delta_j^+ & \text{ if } j \in \mathcal{N}_k^- \\ 0 & \text{ otherwise,} \end{cases} \label{eqn:bbga_B}
\end{align}
and taking $A_{j,k}$ to be any values which satisfy the constraints \eqref{eqn:constraints_a}. To implement such a scheme we only require that each node has knowledge of its in-neighbors, which is reasonable in the broadcast setting. 

Observe that, for the choice of parameters just specified, both $I - L_k$ and $\overline{S}$ are row-stochastic matrices. Let $v$ be such that $v^T B = v^T$ and $v^T \mathbf{1} = 1$. Thus, the entries of $v$ satisfy $v_k = \sum_{j \in \calN_k^-} v_j / \delta_j^+$, and all entries of $v$ are positive. Such an eigenvector exists since $B$ is also row-stochastic. One can verify that $v^T S_k = v^T$ also holds, and so $v^T \overline{S} = v^T$. 
In general, we do not have $v = \frac{1}{n} \mathbf{1}$ unless $\delta_j^+ = \delta_j^-$ for all $j$ and $\delta_j^+ = \delta_i^+$ for all $i \ne j$. Therefore convergence to the average consensus can no longer be guaranteed in general. However, we still obtain convergence in expectation to a (non-average) consensus, via Theorem~\ref{thm:convergenceInExpectation}, and we can also show that BBGA converges in the second moment.

\begin{proposition} \label{prop:bbga}
Let the parameters $d^{(k)}_j$ and $B_{j,k}$ be chosen as in~\eqref{eqn:bbga_d} and \eqref{eqn:bbga_B}. There exists $\eta > 0$ such that if $\epsilon \in (0, \eta]$ then \eqref{eqn:secondMomentCondition} holds and so \eqref{eqn:secondMomentConvergence} also holds with $v$ being the vector such that $v^T B = v^T$ and $v^T \mathbf{1} = 1$.
\end{proposition}

\begin{proof}
To prove the claim we show that \eqref{eqn:secondMomentCondition} holds and then invoke Theorem~\ref{thm:secondMoment}. We use an argument similar to that used in the proof of Proposition~\ref{prop:ubga}. Since $(I - L_k) \mathbf{1}$, $S_k \mathbf{1} = \mathbf{1}$, and $v^T S_k = v^T$, the following four equalities hold:
\begin{align}
\mathbb{E}[(I- L_k) \otimes (I - L_k)] (\mathbf{1} \otimes \mathbf{1}) &= (\mathbf{1} \otimes \mathbf{1}) \\
\mathbb{E}[(I- L_k) \otimes S_k] (\mathbf{1} \otimes \mathbf{1}) &= (\mathbf{1} \otimes \mathbf{1}) \\
\mathbb{E}[S_k \otimes (I - L_k)] (\mathbf{1} \otimes \mathbf{1}) &= (\mathbf{1} \otimes \mathbf{1})\\
(v^T \otimes v^T) \mathbb{E}[S_k \otimes S_k] &= (v^T \otimes v^T).
\end{align}
Since the eigenvectors above are all positive, it follows that $1$ is the largest eigenvalue of each of the four matrices $\mathbb{E}[(I- L_k) \otimes (I - L_k)]$, $\mathbb{E}[(I- L_k) \otimes S_k]$, $\mathbb{E}[S_k \otimes (I - L_k)]$, and $\mathbb{E}[S_k \otimes S_k]$.

Similar to the proof of Proposition~\ref{prop:ubga}, we find that $\widehat{M}$ has largest eigenvalue $1$ with multiplicity $4$, and all the moduli of all other eigenvalues are strictly smaller than $1$. Let $\lambda_1(\epsilon), \dots, \lambda_4(\epsilon)$ denote the four corresponding eigenvalues of $\mathbb{E}[W \otimes W]$. In this case, we can solve for the eigenvalues and again take their derivatives to find that
\begin{align}
&\frac{d \widehat{\lambda}_1(\epsilon)}{d \epsilon} = 0 \\
&\frac{d \widehat{\lambda}_2(\epsilon)}{d \epsilon} = \frac{d \widehat{\lambda}_3(\epsilon)}{d \epsilon} = - v_1^T \mathbb{E}[D] \mathbf{1} < 0 \\
&\frac{d \widehat{\lambda}_4(\epsilon)}{d \epsilon} = - 2 v_1^T \mathbb{E}[D] \mathbf{1} < 0,
\end{align}
where $v_1$ satisfies $v_1^T \mathbb{E}[I - L_k] = v_1^T$ and $v_1^T \mathbf{1} = 1$. Thus, there exists a positive scalar $\eta > 0$ such that $1$ is a simple eigenvalue of $\mathbb{E}[W \otimes W]$ and all other eigenvalues are strictly less than $1$ when $\epsilon \in (0, \eta]$. Subsequently, \eqref{eqn:secondMomentCondition} holds, and convergence in the second moment follows from Theorem~\ref{thm:secondMoment}.
\end{proof}

\section{Upper Bound on $\eta$} \label{sec:supremum}

So far we have demonstrated that there exist broadcast gossip algorithms of the form described in Section~\ref{sec:framework} which are guaranteed to converge when the parameter $\epsilon$ is chosen to be sufficiently small. In this section we derive bounds on $\eta$ which can be used as practical guidelines for setting this parameter. Previous results suggest that, in general, one must take $\eta = \Theta(n^{-n})$, which is extremely conservative~\cite{CA12,BH96,ST90}. The bounds in this section make use of the specific structure of $W(t)$ to obtain tighter, more useful bounds.

We begin with a simple observation related to the expected BBGA update matrix.

\begin{lemma}
For updates using the BBGA parameters and for sufficiently small $\epsilon > 0$, the second largest eigenvalue of $\overline{W}$ is $1 - \epsilon / n$.
\end{lemma}

\begin{proof}
One can verify that, for BBGA, $\overline{W} [\mathbf{1}^T \ \mathbf{0}^T]^T = [\mathbf{1}^T \ \mathbf{0}^T]^T$ and $\overline{W} [\mathbf{1}^T \ -\mathbf{1}^T]^T = (1 - \epsilon / n) [\mathbf{1}^T \ -\mathbf{1}^T]^T$. Thus, $1$ and $1 - \epsilon / n$ are eigenvalues of $\overline{W}$. According to Lemma~\ref{property1}, the eigenvalues of $\overline{W}$ are continuous functions of $\epsilon$. If $1$ or $1 - \epsilon / n$ are not eigenvalues stemming from the semi-simple double eigenvalue $1$ of $W_0$, then we obtain a contradiction as $\epsilon \rightarrow 0$. Therefore, $1 - \epsilon / n$ must be the second largest eigenvalue of $\overline{W}$ for sufficiently small $\epsilon$.
\end{proof}

In order to provide a tight characterization of the upper bound $\eta$ on the perturbation parameter, we need to make more specific assumptions about the values of the weights $a_{j,k}$. Previously, we only assumed that they satisfy the constraints \eqref{eqn:constraints_a}. For the remainder of the article, unless otherwise stated, we assume that
\begin{equation}
a_{j,k} = \begin{cases} 1 / \delta_j^+ & \text{ if } j \in \mathcal{N}_k^- \\ 0 & \text{ otherwise}.\end{cases} \label{eqn:bbga_a}
\end{equation}
In this case, for BBGA, we have $\overline{S} = I - \overline{L}$.

Let $\xi_i$ denote the eigenvalues of the graph Laplacian $L = \diag(A \mathbf{1}) - A$ sorted by increasing real part~\cite{BE00}, where $n$ is counted with multiplicity; i.e., $0 = \Re(\xi_1) \le \Re(\xi_2) \le \dots \le \Re(\xi_n) \le 2$. Note that this last inequality holds because $\rho(L) \le \|L \|_\infty = 2$, and the real part of $\xi_k$ for any $k$ is nonnegative because $L$ is diagonally dominant. We have the following lemma characterizing the eigenvalues of $\overline{W}$ for BBGA.

\begin{lemma} \label{lemma10}
For BBGA, the $2n$ eigenvalues of $\overline{W}$ are
\begin{align}
\lambda_{k,1} &= 1 - \frac{1}{n} \xi_k - \frac{\epsilon}{2n} - \frac{1}{n} \sqrt{\epsilon \xi_k + \frac{\epsilon^2}{4}} \quad \text{ for } k=1,\dots,n, \label{eqn:lambdak1} \\
\lambda_{k,2} &= 1 - \frac{1}{n} \xi_k - \frac{\epsilon}{2n} + \frac{1}{n} \sqrt{\epsilon \xi_k + \frac{\epsilon^2}{4}} \quad \text{ for } k=1,\dots,n. \label{eqn:lambdak2}
\end{align}
\end{lemma}

\begin{proof}
Observe that \eqref{eqn:bbga_d} implies $\overline{D} = \frac{1}{n} I$ for BBGA. The characteristic polynomial of $\overline{W}$ is
\begin{align}
\det(\lambda I - \overline{W}) 
&= \det\left(\begin{bmatrix} (\lambda - 1)I + \frac{1}{n} L & - \frac{\epsilon}{n} I \\
- \frac{1}{n} L & (\lambda - 1 + \frac{\epsilon}{n})I + \frac{1}{n} L\end{bmatrix}\right) \\
&= \det\left(\begin{bmatrix} (\lambda - 1) \left(\lambda - 1 + \frac{\epsilon}{n}\right) I + \frac{2(\lambda - 1)}{n} L + \frac{1}{n^2} L^2\end{bmatrix}\right), \label{eqn:charPoly}
\end{align}
where the last inequality follows since the matrices $I$ and $L$, as well as any linear combination of them, are commutative. The zeros of the characteristic polynomial of $\overline{W}$ thus correspond to zero eigenvalues of the matrix
\begin{equation}
U = (\lambda - 1) \left(\lambda - 1 + \frac{\epsilon}{n}\right) I + \frac{2(\lambda - 1)}{n} L + \frac{1}{n^2} L^2.
\end{equation}
According to the spectral mapping theorem\cite{ME01}, for any eigenvalue $\xi_i$ of the matrix $L$, the matrix $U$ has a corresponding eigenvalue $(\lambda - 1)(\lambda - 1 - \frac{\epsilon}{n}) + \frac{2(\lambda - 1)}{n}\xi_i + \frac{1}{n^2}\xi_i^2$. Therefore, given an eigenvalue $\lambda_k$ of $\overline{W}$, there must exist an eigenvalue $\xi_j$ of $L$ such that
\begin{equation}
(\lambda_k - 1) \left(\lambda_k - 1 + \frac{\epsilon}{n}\right) + \frac{2(\lambda_k - 1)}{n} \xi_j + \frac{1}{n^2} \xi_j^2 = 0. \label{eqn:40}
\end{equation}
According to Lemma~\ref{property1}, $\lambda_k$ is a continuous function of $\epsilon$, and so as $\epsilon \rightarrow 0$, $\lambda_k$ is an eigenvalue of $W_0$. For the particular choice of parameters made above, we have that $I - \overline{L} = \overline{S}$. Since $\xi_k$ is an eigenvalue of $L$, we have that $1 - \frac{1}{n} \xi_k$ is an eigenvalue of $I - \overline{L} = \overline{S} = I - \frac{1}{n} L$. Thus, in the continuous limit as $\epsilon \rightarrow 0$, we have that $\lambda_k = 1 - \frac{1}{n} \xi_k$ is an eigenvalue of $\overline{W}$. Taking $\epsilon = 0$ and substituting for $\lambda_k$ in \eqref{eqn:40}, we find that $(\xi_j - \xi_k)^2 = 0$, which implies that $j = k$. In this case, \eqref{eqn:40} can be simplified to
\begin{equation}
(\lambda_k - 1) \left(\lambda_k - 1 + \frac{\epsilon}{n}\right) + \frac{2(\lambda_k - 1)}{n} \xi_k + \frac{1}{n^2} \xi_k^2 = 0. 
\end{equation}
Solving for the roots of these quadratic polynomials in $\lambda_k$ proves the claim.
\end{proof}

We have already seen that $L$ has an eigenvalue $\lambda_1 = 0$. From Lemma~\ref{lemma10}, we again find that $\lambda_{1,1} = 1 - \epsilon / n$ and $\lambda_{1,2} = 1$ are eigenvalues of $\overline{W}$. If all eigenvalues of $L$ are real, then all eigenvalues of $\overline{W}$ are also real. In this case, we can use the monotonic ordering of eigenvalues to determine an upper bound $\eta$ on $\epsilon$ to ensure that BBGA converges in expectation. We restrict to the case where all eigenvalues of $L$ are real. When $L$ has some complex eigenvalues monotonicity is no longer preserved making it difficult to determine a reasonable bound. The main result of this section is as follows.

\begin{proposition} \label{theorem4}
Consider BBGA updates and suppose that all eigenvalues of $L$ are real. Then BBGA converges in expectation when $0 \le \epsilon \le 2n + \xi_n^2/(2n) - 2\xi_n$, where $\xi_n$ is the largest eigenvalue of $L$.
\end{proposition}

The proof of Proposition~\ref{theorem4} relies on two intermediate lemmas.

\begin{lemma} \label{lemma11}
For BBGA updates, if all eigenvalues of $L$ are real, then there is one stable eigenvalue $\lambda_{1,2} = 1$. All other eigenvalues $\lambda_{k,1}$ are strictly decreasing functions of $\epsilon$, and all eigenvalues $\lambda_{k,2}$ are strictly increasing functions of $\epsilon$.
\end{lemma}

\begin{proof}
To prove the claim, differentiate \eqref{eqn:lambdak1} and \eqref{eqn:lambdak2} with respect to $\epsilon$:
\begin{align}
\frac{d \lambda_{k,1}}{d \epsilon} &= - \frac{1}{2n}\left(1 + \frac{2 \xi_k + \epsilon}{\sqrt{4 \epsilon \xi_k + \epsilon^2}}\right) \\
&\le - \frac{1}{2n}\left(1 + \frac{2 \xi_k + \epsilon}{\sqrt{(2 \xi_k + \epsilon)^2}}\right) \\
&= - \frac{1}{n} \\
&< 0,
\end{align}
and
\begin{align}
\frac{d \lambda_{k,2}}{d \epsilon} &= - \frac{1}{2n}\left(1 - \frac{2 \xi_k + \epsilon}{\sqrt{4 \epsilon \xi_k + \epsilon^2}}\right)  \\
&\ge - \frac{1}{2n}\left(1 - \frac{2 \xi_k - \epsilon}{\sqrt{(2 \xi_k + \epsilon)^2}}\right) \label{eqn:43} \\ 
&= 0. 
\end{align}
Note that equality holds in \eqref{eqn:43} if and only if $\xi_k = 0$.
\end{proof}

\begin{lemma} \label{lemma12}
For BBGA updates, If all eigenvalues of $L$ are real, then the eigenvalues $\lambda_{k,1}$ and $\lambda_{k,2}$ are monotonic decreasing functions of $\xi_k$.
\end{lemma}

\begin{proof}
It is clear from \eqref{eqn:lambdak1} that $\lambda_{k,1}$ is a monotonic decreasing function of $\xi_k$. For $\lambda_{k,2}$ and for fixed $\epsilon > 0$, observe that
\begin{equation}
\frac{d \lambda_{k,2}}{d \xi_k} = - \frac{1}{n} + \frac{1}{n \sqrt{4 \xi_k / \epsilon + 1}} \le 0,
\end{equation}
where the inequality follows since $\xi_k \ge 0$. Therefore, $\lambda_{k,2}$ is also a monotonic decreasing function of $\xi_k$.
\end{proof}

\begin{proof}[Proof of Proposition~\ref{theorem4}]
Under the assumption that the eigenvalues of $L$ are nonnegative real numbers, we have
\begin{equation}
\frac{\epsilon}{2} \le \sqrt{\epsilon \xi_k + \frac{\epsilon^2}{4}} \le \sqrt{\left(\xi_k + \frac{\epsilon}{2}\right)^2} = \frac{\epsilon}{2} + \xi_k. \label{eqn:45}
\end{equation}
Substituting \eqref{eqn:45} into \eqref{eqn:lambdak1} and \eqref{eqn:lambdak2}, we get
\begin{align}
1 - \frac{\epsilon}{n} - \frac{2\xi_k}{n} \le &\ \lambda_{k,1} \le 1 - \frac{\epsilon}{n} - \frac{\xi_k}{n},\\
1 - \frac{\xi_k}{n} \le &\ \lambda_{k,2} \le 1,
\end{align}
where the equalities hold if and only if $k=1$ with corresponding $\xi_k=0$. From these expressions, it is clear that $\lambda_{1,2} =1$ is a simple eigenvalue of $\overline{W}$ and all other eigenvalues are strictly smaller than $1$ since $0 \le \xi_k \le 2$ and $\epsilon > 0$. Furthermore, convergence is guaranteed when all eigenvalues are strictly larger than $-1$. For a given $\xi_k$, observe that $\lambda_{k,1} \le \lambda_{k,2}$. In addition, from Lemma~\ref{lemma12}, we have that $\lambda_{k,1}$ is a monotonic decreasing function of $\xi_k$. Therefore, $\lambda_{k,2} \ge \lambda_{k,1} \ge \lambda_{n,1}$ for all $k=1,\dots,n$. Thus, we focus on determining conditions under which $\lambda_{n,1} > -1$. According to Lemma~\ref{lemma11}, $\lambda_{n,1}$ is a strictly decreasing function of $\epsilon$. If there exists $\eta$ such that $\lambda_{n,1} = -1$ when $\epsilon = \eta$, then $\lambda_{n,1} > -1$ when $\epsilon < \eta$. Solving \eqref{eqn:lambdak1} for $\lambda_{n,1} = -1$ we obtain
\begin{equation}
\eta = 2n + \frac{\xi_n^2}{2n} - 2 \xi_n, \label{eqn:48}
\end{equation}
which completes the proof.
\end{proof}

\begin{remark} \label{practicalBound}
In general, the value of $\xi_n$ depends on the network topology, and it may not be easy to determine a precise value of $\xi_n$. A more practical guideline is to take $\epsilon \in \big(0, \frac{2}{n} (n - 1)^2\big)$. To see why this is reasonable, differentiate \eqref{eqn:48} with respect to $\xi_n$:
\begin{equation}
\frac{d \eta}{d \xi_n} = \frac{\xi_n}{n} - 2 \le \frac{2}{n} - 2 \le 0.
\end{equation}
Therefore, $\eta$ is a monotonic decreasing function of $\xi_n$, and $\eta$ thus satisfies $\frac{2}{n}(n - 1)^2 \le \eta \le 2n$ since $0 \le \xi_n \le 2$. If the perturbation parameter $\epsilon$ is not larger than $\frac{2}{n}(n-1)^2$ then BBGA is guaranteed to converge in expectation.
\end{remark}

Note that, from Remark~\ref{practicalBound}, the upper bound $\eta$ is not smaller than $1$. In the following section we investigate what value of $\epsilon$ leads to the fastest convergence. We find that we typically seek values of $\epsilon$ less than $1$, and so this upper bound will suffice.

Although the guidelines derived above are for BBGA, in extensive simulations we have observed that the maximal value of $\epsilon$ under which UBGA still converges is typically no different than that for BBGA for a given graph. Therefore, the guidelines derived above can be also used as approximate guidelines for setting the parameters of UBGA.

\section{Optimal Perturbation Parameter} \label{sec:optimal}

In the previous section we determined an upper bound on the perturbation parameter $\epsilon$ under which convergence in expectation is guaranteed. This can be viewed as a sort of stability result. In this section we investigate what value of the perturbation parameter leads to the fastest rate of convergence. It is well known that the worst-case rate of convergence of systems of the form~\eqref{eqn:expectedUpdate} is governed by the second largest eigenvalue of $\overline{W}$. In the previous section we saw that, for BBGA, this second largest eigenvalue is $1-\epsilon/n$ if the perturbation parameter $\epsilon > 0$ is sufficiently small, and this eigenvalue is a monotonic decreasing function of $\epsilon$. At the same time, other eigenvalues of $\overline{W}$ are monotonic increasing, and so it follows that the optimal value of $\epsilon$ is the one where the modulus of $1 - \epsilon/n$ first coincides with the modulus of another eigenvalue of $\overline{W}$.

\begin{theorem} \label{theorem5}
Consider the expected update matrix $\overline{W}$ corresponding to BBGA, and suppose that all eigenvalues of the Laplacian $L = \diag(A \mathbf{1}) - A$ are real. For networks with at least three nodes, the modulus of the second largest eigenvalue of $\overline{W}$ is minimized when the perturbation parameter is equal to $\epsilon^* = \xi_2 / 2$, where $\xi_2$ is the second largest eigenvalue of the graph Laplacian $L$. In this case, the second largest eigenvalue of $\overline{W}$ is $1 - \xi_2 / (2n)$. When $n=2$, the second largest eigenvalue of $\overline{W}$ is minimized by $\epsilon^* = 2 - \sqrt{2}$.
\end{theorem}

\begin{proof}
According to Lemma~\ref{lemma11}, for a fixed $\epsilon > 0$, the eigenvalues of $\overline{W}$ satisfy
\begin{equation}
1 - \epsilon / n = \lambda_{1,1} \ge \lambda_{2,1} \ge \dots \ge \lambda_{n,1},
\end{equation}
and
\begin{equation}
1 = \lambda_{1,2} \ge \lambda_{2,2} \ge \dots \ge \lambda_{n,2}.
\end{equation}
We also have monotonicity of the respective eigenvalues as a function of $\epsilon$ from Lemma~\ref{lemma12}. Because the eigenvalues are continuous functions of $\epsilon$, it follows that there are two points of interest where the second largest eigenvalue (in modulus) may switch from being $\lambda_{1,1} = 1 - \epsilon/n$. These are the points $\epsilon_1$ where $\lambda_{1,1}(\epsilon_1) = \lambda_{2,2}(\epsilon_1)$ and $\epsilon_2$ where $\lambda_{1,1}(\epsilon_2) = -\lambda_{n,1}(\epsilon_2)$. To complete the proof we can solve for these two values of and then determine that $\epsilon^* = \min\{\epsilon_1, \epsilon_2\}$.

Solving $\lambda_{1,1}(\epsilon_1) = \lambda_{2,2}(\epsilon_2)$, we have $\epsilon_1 = \xi_2/2$ and the corresponding eigenvalue of $\overline{W}$ is $1 - \xi_2 / (2n)$. 

To solve $\lambda_{1,1}(\epsilon_2) = - \lambda_{n,1}(\epsilon_2)$, observe that
\begin{align}
\epsilon_2 &= \frac{3n - \xi_n - \sqrt{n^2 + 2n \xi_n - \xi_n^2}}{2} \\
&\ge \frac{3n - \xi_n - \sqrt{(n + \xi_n)^2}}{2} \\
&= n - \xi_n.
\end{align}
Since $0 \le \xi_n \le 2$, it must be that $\epsilon_2 > 1$ if $n \ge 3$, from which we find that $\epsilon_2 > \epsilon_1$ since $\xi_2 / 2 < 1$. Therefore, the optimal perturbation parameter is $\epsilon^* = \epsilon_1 = \xi_2 / 2$ when $n \ge 3$. If $n=2$, then there is only one non-zero eigenvalue of the weighted Laplacian matrix $L$ and it is equal to $2$. In this case, $\epsilon_2 = 2 - \sqrt{2} < 1 = \epsilon_1$, so the optimal perturbation parameter is $\epsilon^* = 2 - \sqrt{2}$.
\end{proof}

\begin{remark}
Note that since the modulus of the second largest eigenvalue of $\overline{W}$ satisfies $|\lambda_{1,1}(\epsilon^*)| = \lambda_{1,1}(\epsilon^*) < \lambda_{1,1}(0) < 1$, we see that BBGA is guaranteed to converge in expectation for this setting.
\end{remark}

The above analysis focused on the case where the eigenvalues of $L$ are assumed to be real. In extensive simulations, we have observed that this is the case whenever $\calG$ is undirected, regardless of whether the edge weights are symmetric. For digraphs, the eigenvalues $\xi_i$ of $L$ are generally complex numbers, so a monotonicity property such as that obtained in Lemma~\ref{lemma12} is no longer readily available. Below we analyze the optimal value of the perturbation parameter on random digraphs via simulation. We find that $\tilde{\epsilon} = \Re(\xi_2) / 2$ is a good guideline for directed graphs.

\subsection{Undirected Graphs}

Consider an undirected graph as illustrated in Figure~\ref{epsilon_undirected} with 16 nodes distributed uniformly in the unit square. Nodes are connected if the Euclidean distance between them is no more than $\sqrt{2\log{n}/n}$ so that the graph $\calG$ is connected with probability at least $1-1/n^2$ \cite{GU98,PE03}; this is the standard random geometric graph model. For the graph shown in Fig.~\ref{epsilon_undirected}, the second smallest eigenvalue of weighted Laplacian matrix $L$ for BBGA is $\xi_{2}=0.5335$ so the optimal perturbation parameter for BBGA is $\epsilon^*=\xi_{2}/2=0.2668$.

\begin{figure}[!t]
\centering
\includegraphics[width=2.5in]{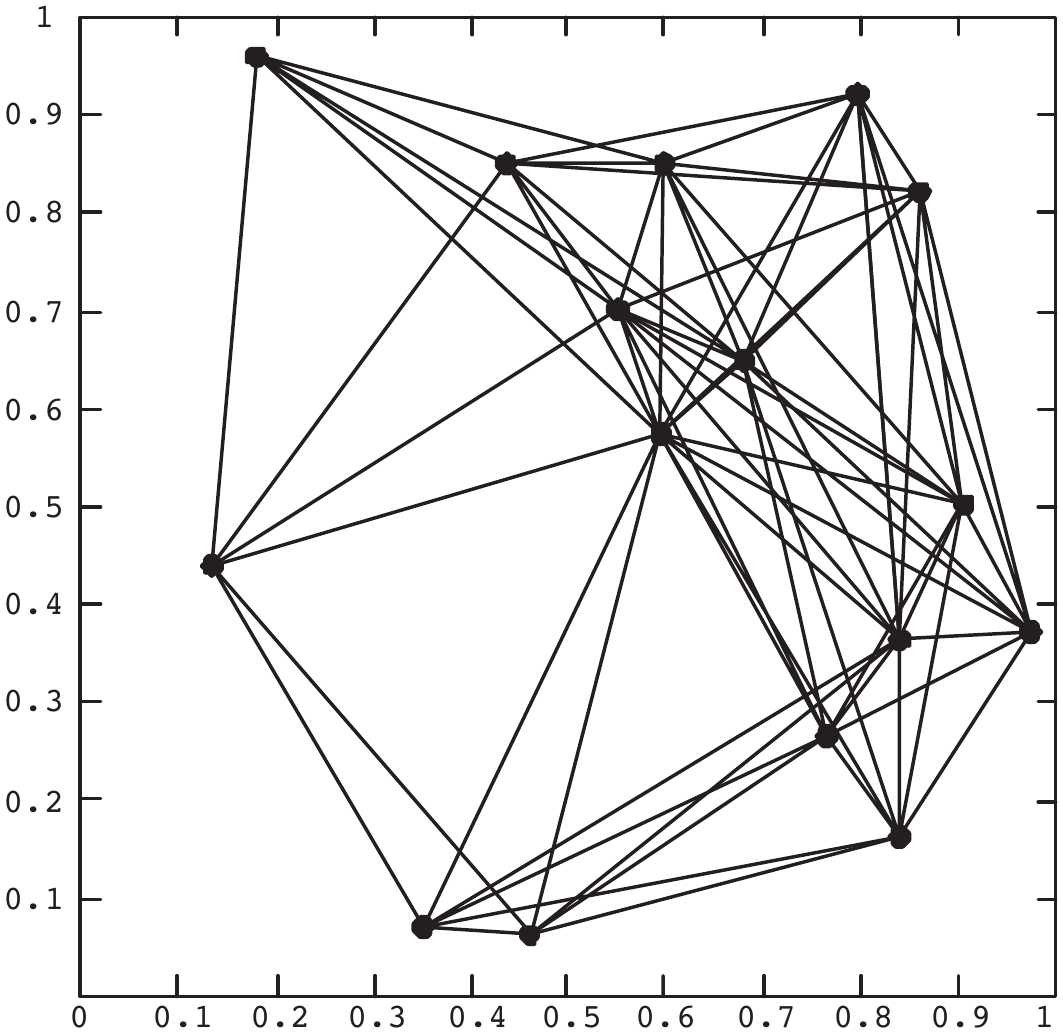}
\caption{An example for undirected graphs with 16 nodes}
\label{epsilon_undirected}
\end{figure}

\begin{figure}[!t]
\centering
\includegraphics[width=2.5in]{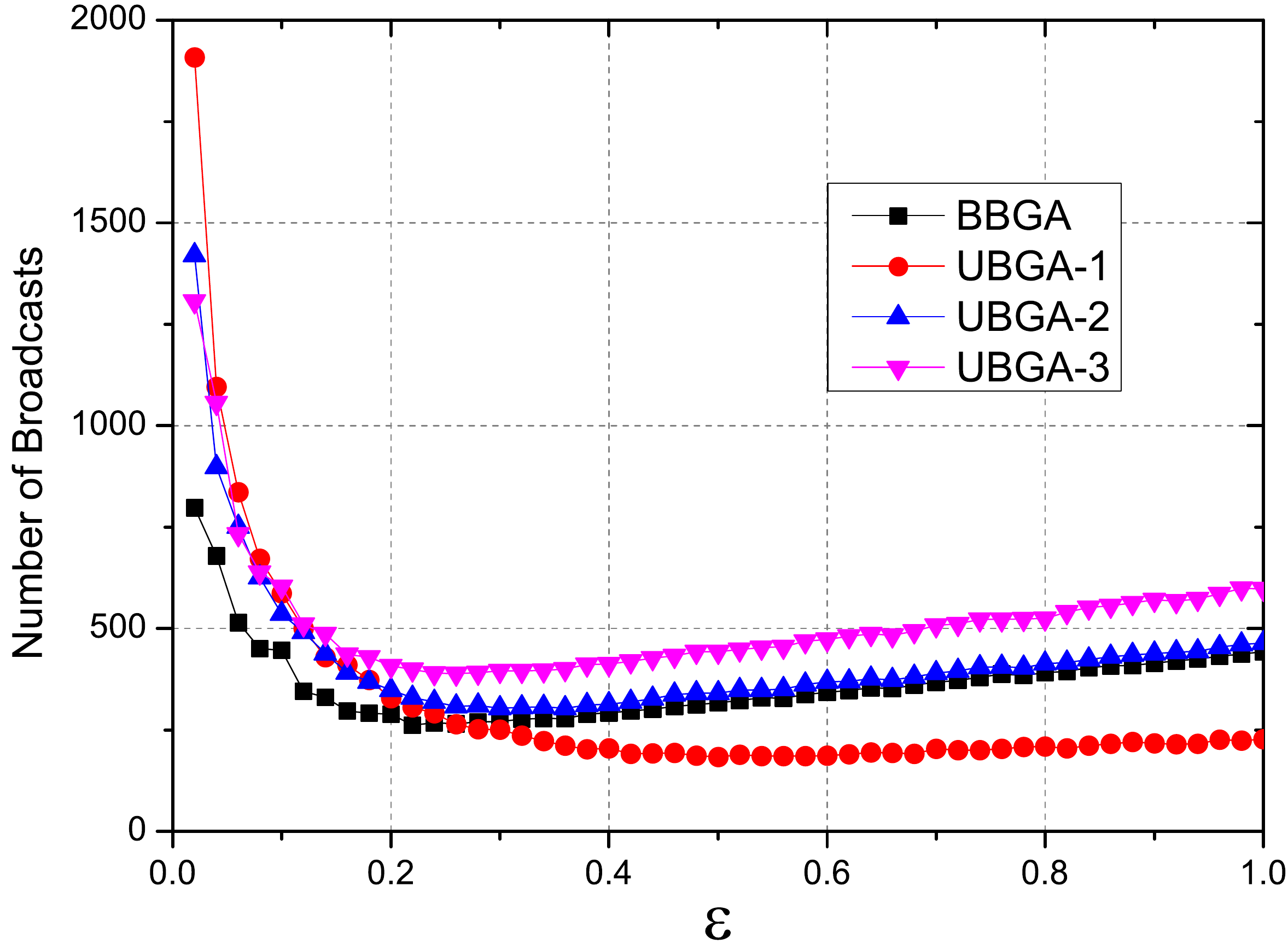}
\caption{Number of broadcasts to converge with respect to $\epsilon$ for simulations on the graph shown in Fig.~\ref{epsilon_undirected}.}
\label{epsilon_undirected_broadcasts}
\end{figure}

Figure~\ref{epsilon_undirected_broadcasts} shows the number of broadcasts to achieve consensus as a function of $\epsilon$. Each point is an average over 100 trials, and we sweep over values of $\epsilon$ from $0.02$ to $1$ in increments of $0.02$. The initial values of all nodes are independent and uniformly distributed between $0$ and $1$. Recall the error vector $m(t)$ defined in the proof of Theorem~\ref{thm:secondMoment}. We declare that consensus is achieved at the first iteration $t$ where $\|m(t) - m(t-1)\|_2 \le 10^{-5}$. Here we compare four broadcasts gossip algorithms. For BBGA we use the weights $a_{j,k}$ as defined in \eqref{eqn:bbga_a}. The three versions of UBGA have different choices of weights $a_{j,k}$; they are
\begin{equation}
a_{j,k} = \begin{cases} 0.5 & \text{ if } j \in \calN_k^- \text{ for UBGA-1},\\
1/\delta_j^+ & \text{ if } j \in \calN_k^- \text{ for UBGA-2},\\
1/\delta_j^- & \text{ if } j \in \calN_k^- \text{ for UBGA-3},\\
0 & \text{ if } j \notin \calN_k^-.
\end{cases} \label{eqn:ubga_as}
\end{equation}
Observe that the fastest convergence for BBGA occurs near $\epsilon=0.26$, which matches the value predicted by Proposition~\ref{theorem5}. Also observe that all three versions of UBGA have larger optimal perturbation parameter than BBGA. Using any version of UBGA at the optimal value $\epsilon^*$ for BBGA results in suboptimal performance. UBGA-1 exhibits a number advantages over the other algorithms: it converges in fewer broadcasts than the other algorithms for suitably chosen $\epsilon$, and the curve for UBGA-1 in Fig.~\ref{epsilon_undirected_broadcasts} is extremely flat near the optimal value, so its performance is very robust to the choice of $\epsilon$ in this region. From a practical perspective, UBGA-1 is also easy to implement in undirected networks since all weights $a_{j,k}$ are constants only depending on the network connectivity. For the graph shown in Fig.~\ref{epsilon_undirected}, the largest eigenvalue of $L$ for BBGA is $1.3796$, and the corresponding upper bound for $\epsilon$ is $29.30$, which can also be verified by simulation.

%

\subsection{Strongly Connected Digraphs}

In practical wireless settings, not all links may be symmetric due to differing transmit powers (e.g., if the batteries at different nodes have experienced different usage), multipath effects, or interference. To simulate directed networks, we begin with a (undirected) random geometric graph and then add and delete directed edges by random coin flips (while ensuring that the directed graph remains strongly connected. Figure~\ref{topo_directed} illustrates an example of a strongly connected directed graph. In this example, since the directed graph has fewer edges than the corresponding undirected graph shown in Fig.~\ref{epsilon_undirected}, one would expect that more broadcasts are needed to achieve consensus. The real part of the second smallest eigenvalue of $L$ for the directed graph in Fig.~\ref{topo_directed} is $0.3930$, so the approximately optimal perturbation parameter is $\epsilon^* = 0.1965$.

\begin{figure}[!t]
\centering
\includegraphics[width=2.5in]{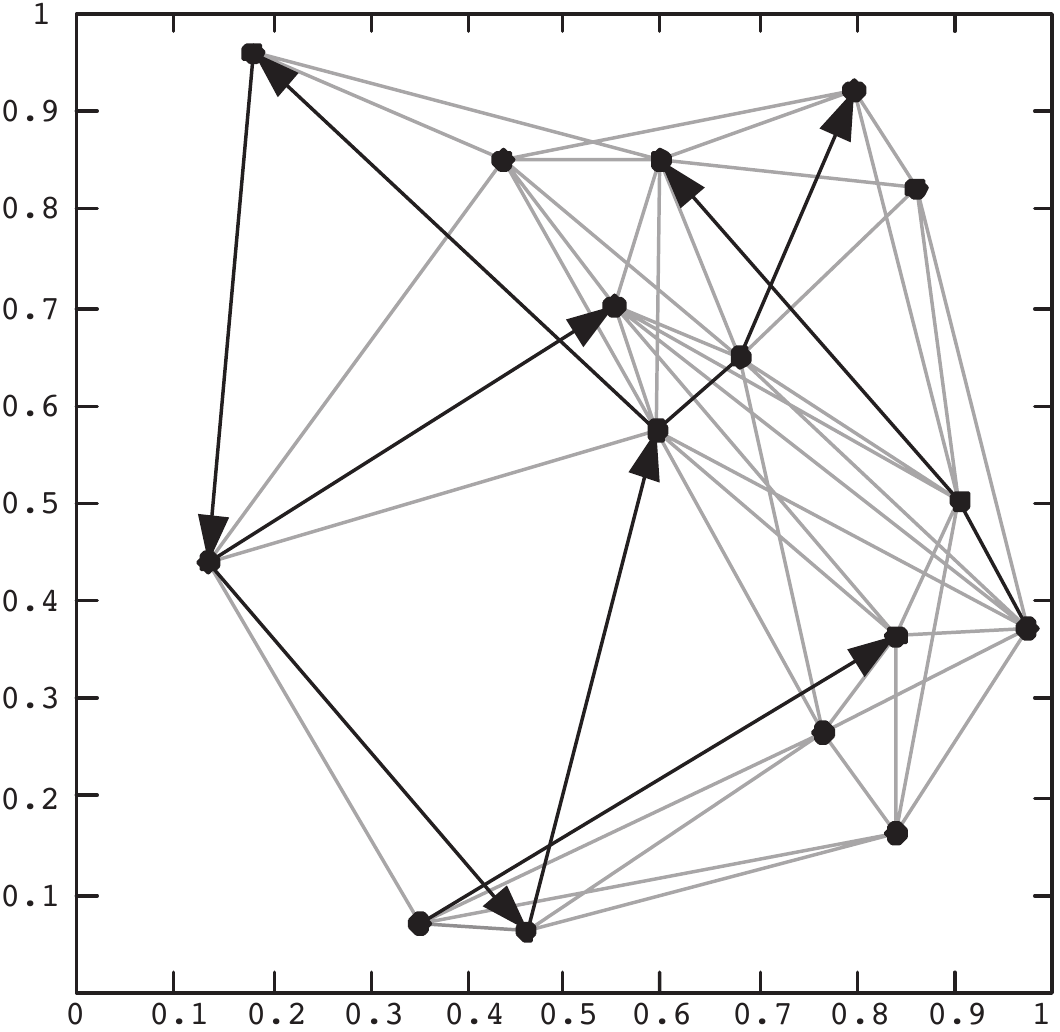}
\caption{An example for digraphs with 16 nodes. The gray lines denote undirected edges and lines with one arrow denote directed edges.}
\label{topo_directed}
\end{figure}

\begin{figure}[!t]
\centering
\includegraphics[width=2.5in]{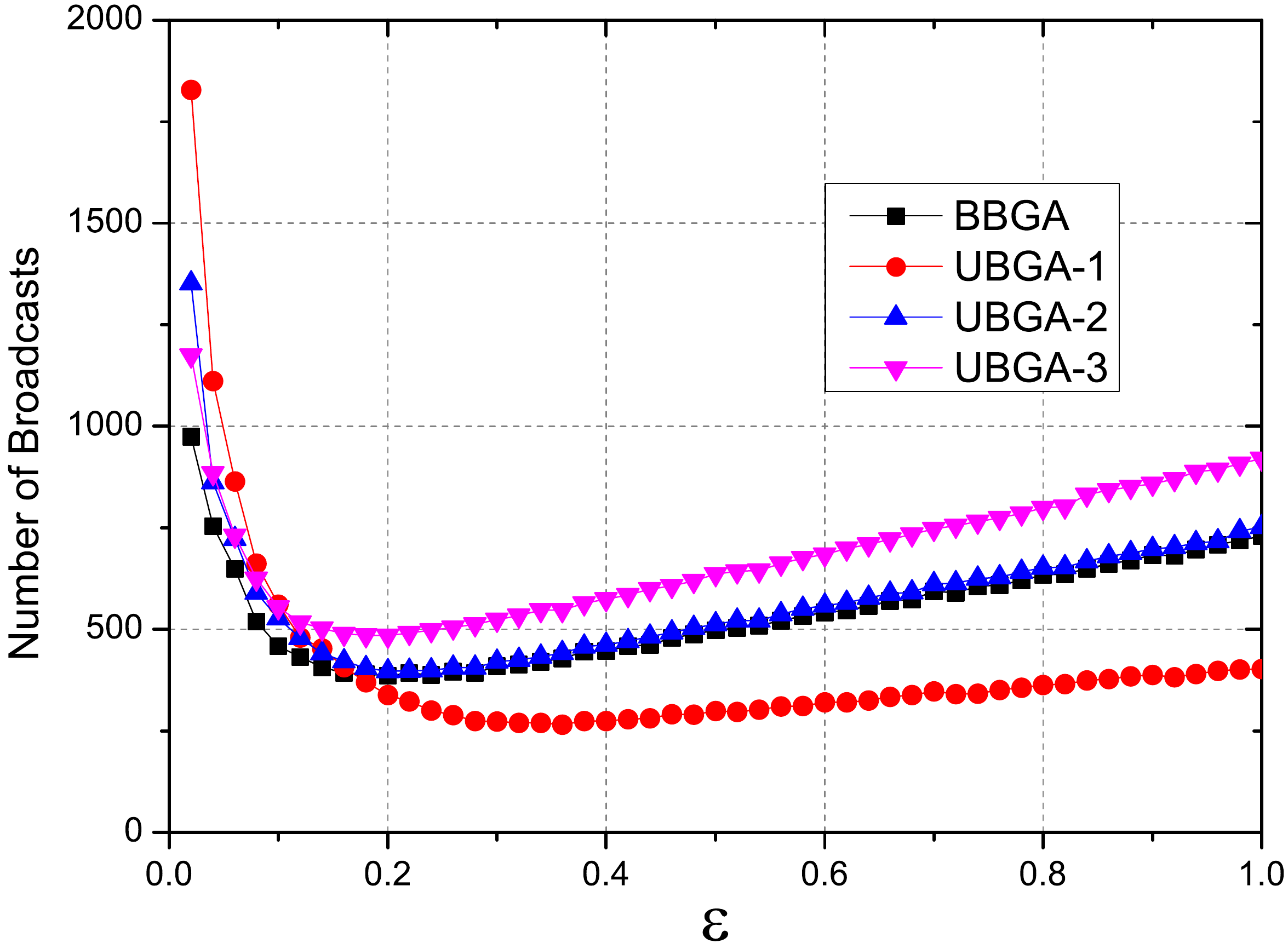}
\caption{Number of broadcasts to converge with respect to $\epsilon$ for simulations on the directed graph shown in Fig.~\ref{topo_directed}.}
\label{number-broadcasts}
\end{figure}

Fig.~\ref{number-broadcasts} illustrates the number of broadcasts with respect to perturbation parameter $\epsilon$ required to obtain $\|m(t) - m(t-1)\|_2 \le 10^{-5}$. The initial values are independent and uniform over $[0,1]$. The optimal perturbation parameter occurs at $0.20$, which matches the predicted value well. From this figure, we see that UBGA-1 still gives the best performance

\subsection{Scaling Behavior}

In the previous two subsections, we illustrate the performance of UBGA and BBGA on particular directed and undirected graphs. In this subsection, we demonstrate the scaling behavior of these two algorithms as the size of the network increases. We compare the performance of three varieties of UBGA and two varieties of BBGA. The three varieties of UBGA are those as defined above, with weights $a_{j,k}$ given in \eqref{eqn:ubga_as}, and with $\epsilon=0.5$. For BBGA, we use the same weights $a_{j,k}$ as given in \eqref{eqn:bbga_a} and set $\epsilon$ either to $0.5$ (BBGA-0.5) or to $\epsilon^*$ (BBGA-opt).

Following~\cite{AY09}, we investigate two metrics for error. The UBGA algorithms are guaranteed to converge to the average consensus solution, and we measure the mean squared error after $t$ iterations,
\begin{equation}
\textstyle r(t) = \frac{1}{n} \|x(t) - \frac{1}{n}\mathbf{1}\mathbf{1}^T x(0)\|_2^2.
\end{equation}
Since BBGA and other biased broadcast gossip algorithms do not converge to the average consensus, we measure their rate of convergence via the deviation,
\begin{equation}
\textstyle q(t) = \frac{1}{n} \|x(t) - \frac{1}{n}\mathbf{1} \mathbf{1}^T x(t)\|_2^2,
\end{equation}
which is guaranteed to go to zero. Note that we do not include the companion variables $y(t)$ in these calculations, since ultimately the aim is to reach consensus only on the states $x(t)$.

As above, we declare that consensus is achieved if $\|m(t) - m(t-1)\|_2 \le 10^{-5}$. In the sequel, we report the results of numerical simulations on random geometric graphs and random strongly connected digraphs simulated using the procedure described above. 

\subsubsection{Convergence rate}

\begin{figure}[!t]
\centering
\includegraphics[width=2.5in]{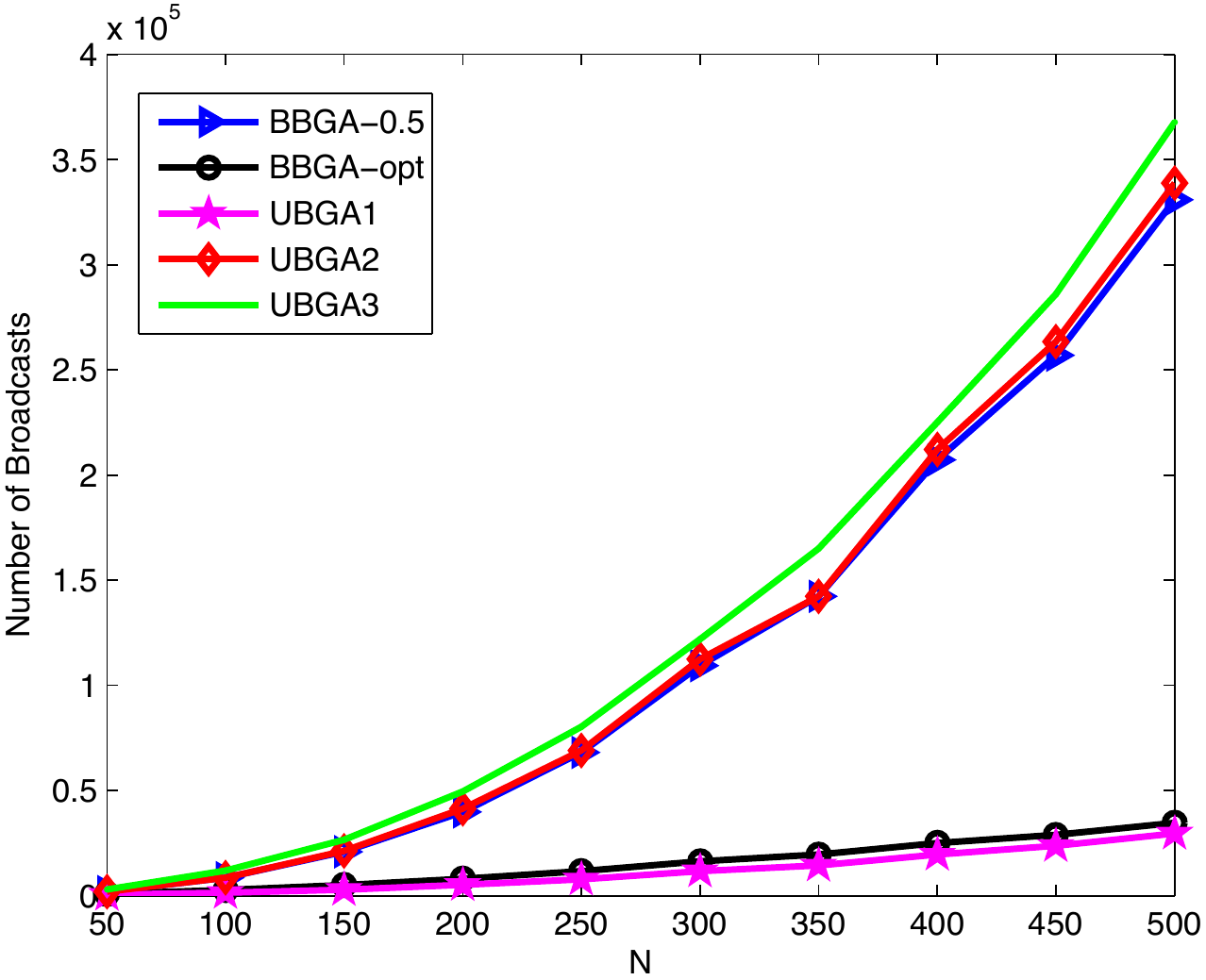}
\caption{Number of broadcasts to converge with respect to $N$ on undirected graphs.}
\label{undi:number-broadcasts}
\end{figure}
\begin{figure}[!t]
\centering
\includegraphics[width=2.5in]{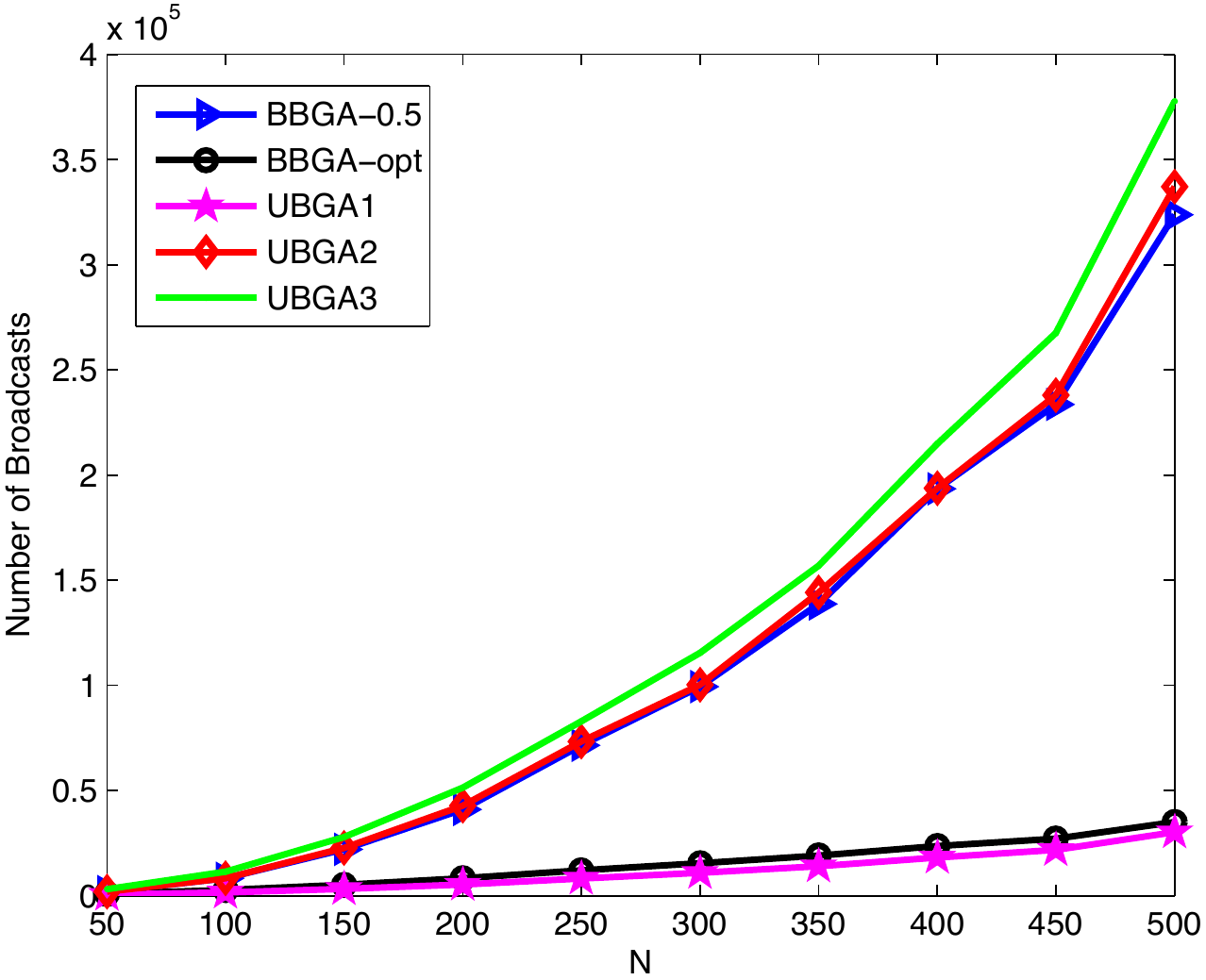}
\caption{Number of broadcasts to converge with respect to $N$ on digraphs.}
\label{dig:number-broadcasts}
\end{figure}

The number of transmissions required to achieve consensus for undirected graphs and digraphs are shown in Figs.~\ref{undi:number-broadcasts} and \ref{dig:number-broadcasts} respectively. Unsurprisingly, BBGA-opt converges significantly faster than BBGA-0.5. However UBGA-1 is still the best one in terms of converge rate performance of scaling behavior. Since UBGA-1 can achieve the average consensus, we prefer UBGA-1 if out-degree information is available; otherwise, BBGA-opt is the winner.

\subsubsection{Deviation}

\begin{figure}[!t]
\centering
\includegraphics[width=2.5in]{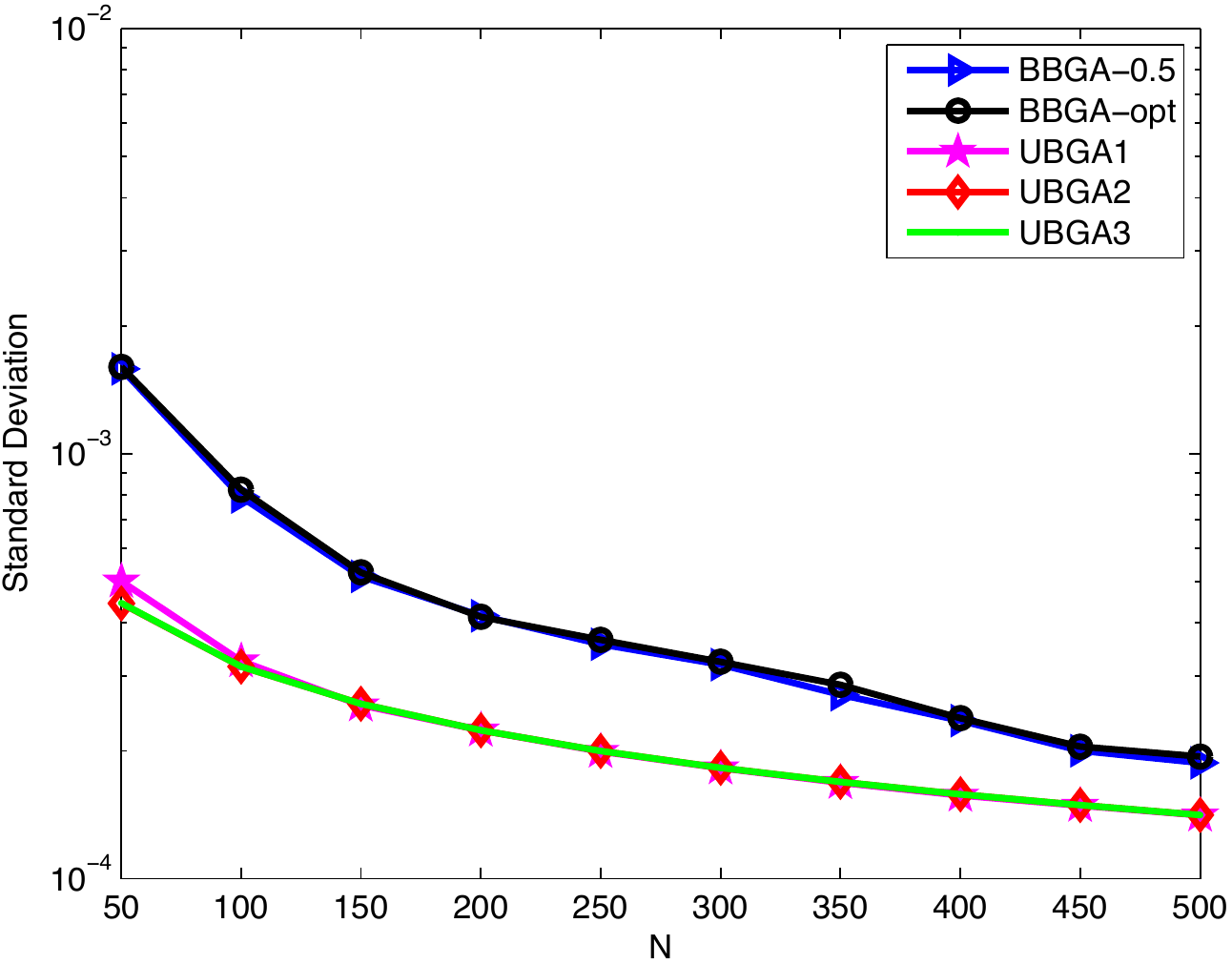}
\caption{The standard deviation performance with respect to $N$ on undirected graphs.}
\label{undi:std_dev}
\end{figure}

\begin{figure}[!t]
\centering
\includegraphics[width=2.5in]{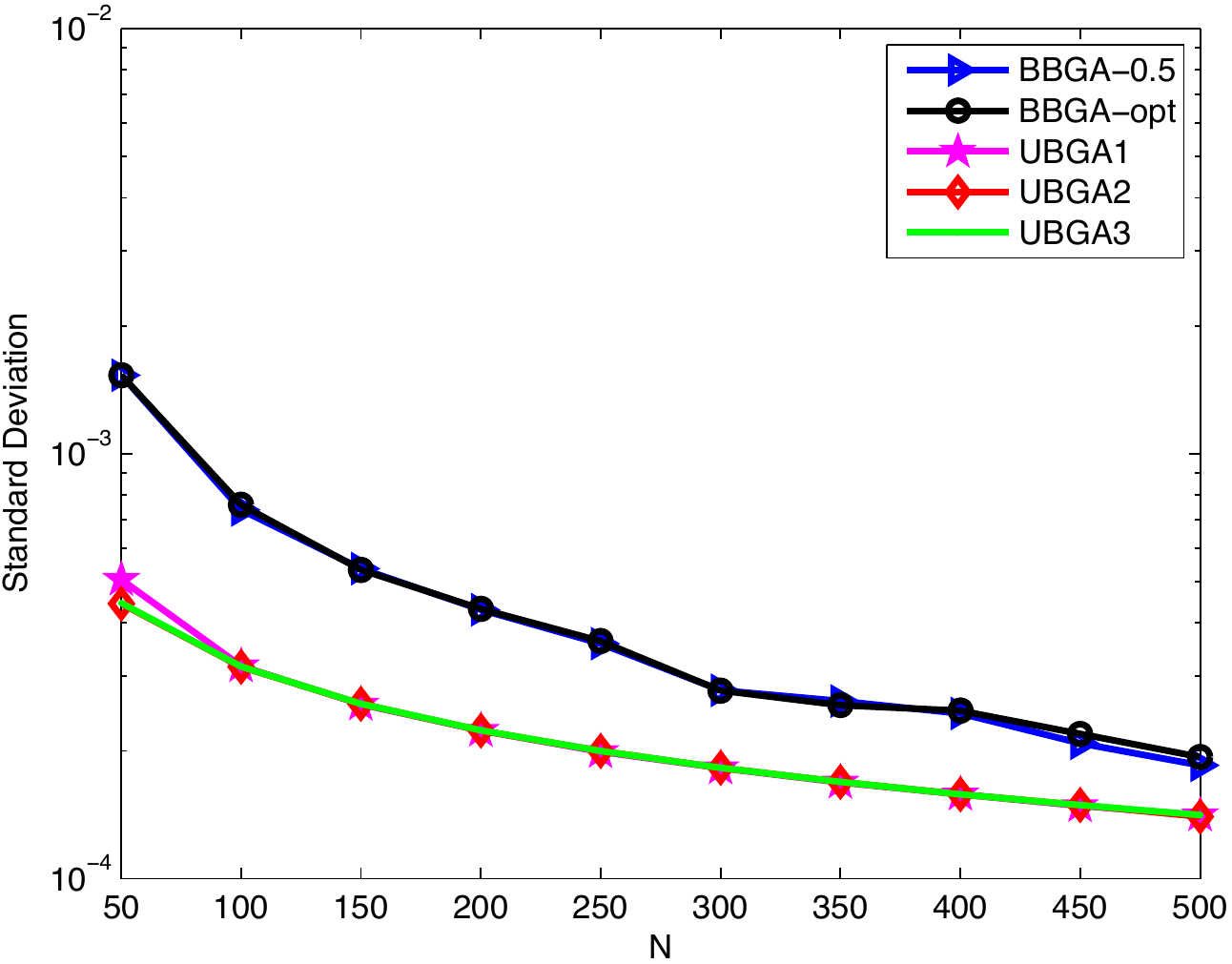}
\caption{The standard deviation performance with respect to $N$ on digraphs.}
\label{dig:std_dev}
\end{figure}

Figures~\ref{undi:std_dev} and \ref{dig:std_dev} show the deviation $q(t)$ at the time the algorithm is declared to have converged. For this particular initialization scheme (i.i.d.~uniform), we note that the BBGA algorithms are roughly half an order of magnitude worse than the UBGA schemes. We investigate the effects of initialization on deviation further in the next section. For now, we note that both versions of BBGA achieve comparable performance in terms of deviation, and likewise, all three versions of UBGA achieve effectively the same deviation at the time they converge.

In summary, from the experiments reported in this section we conclude that UBGA-1 is the most desirable solution if the out-degree information is available (including, when $\calG$ is undirected); otherwise BBGA-opt is the next most preferable since it gives the fastest rate of convergence.

\section{Performance Analysis} \label{sec:perf}

In this section, we compare the broadcast gossip algorithms proposed in this paper with the previous broadcast gossip algorithms of~\cite{AY09} and \cite{FR09}. In the figures and discussion below, BGA-1 refers to the algorithm in \cite{AY09}, BGA-2 refers to the one in \cite{FR09}, and BBGA and UBGA are the algorithms proposed in this paper. The previous section illustrated that UBGA-1 exhibits many advantages, both in terms of the choice of perturbation parameter and the rate of convergence, compared to the other UBGA algorithms. For this reason, in this section we use UBGA-1 as the representative of the UBGAs. For both UBGA and BBGA, we will investigate two settings for the perturbation parameter: $\epsilon = 0.5$ and $\epsilon^*$. Note that $\epsilon^*$ is only optimal for BBGA, and it may be suboptimal for UBGA. We also remark that the comparison of BGA-1 and BGA-2 with BBGA-opt and UBGA-opt (i.e., those using $\epsilon^*$) is unfair, since the information used to determine $\epsilon^*$ is not made available to either BGA-1 or BGA-2; in particular, neither of those algorithms uses global topology information such as $\xi_2$. This is our primary motivation for also considering the performance of UBGA and BBGA with $\epsilon = 0.5$.


All simulations in this section use (undirected) random geometric graph topologies with the same connectivity radius as in the previous section. When $\calG$ is directed, BGA-1 is no longer guaranteed to converge to the average consensus in expectation. Unless otherwise noted, each result corresponds to the average over $100$ Monte Carlo trials.

Since the initial values effect the performance of various broadcast gossip algorithms, we consider four approaches to initializing the values $x_i(0)$: 1) independent and \emph{uniform} over $[0,1]$; 2) independent and \emph{Gaussian} with zero mean, unit variance; 3) the \emph{spike} initialization, where one random node has an initial value of $1$ and all other nodes have initial values $0$; and 4) the \emph{slope} initialization, the initial value at node $i$ is the sum of its x- and y-coordinates in the unit square (note all graphs are drawn from the ensemble of random geometric graphs in the plane.

\subsection{Deviation}

Figures~\ref{uniform:sim_var_und}--\ref{spike:sim_var_und} show the deviation $q(t) = \frac{1}{n}\|x(t) - \frac{1}{n}\mathbf{1}\mathbf{1}^T x(t)\|_2^2$ as a function of time for different initializations. Note that this indicates how quickly the algorithms converge to a consensus, regardless of the value on which consensus is achieved.

It is clear from these figures that BGA-1 converges to its final value faster than the other algorithms. As we will see below, this is because BGA-1 generally achieves a lower accuracy (in terms of mean squared error, $r(t)$) than the other methods.

The algorithms BGA-2, BBGA, and UBGA all maintain companion variables. Among these algorithms we observe that BGA-2 converges slower, in general, than BBGA. Also note that BBGA-opt converges significantly faster than BBGA-0.5 when $n=50$ or $100$, but the performance of the two is much closer for larger networks. Somewhat surprisingly, the deviation of UBGA-opt and UBGA-0.5 are typically better or comparable to BBGA-opt. This is surprising because it indicates that UBGA is converging faster, despite the fact that it is converging to the average consensus. On the other hand, BGA-1 converges quickly to a consensus which is not on the average, and BGA-2 typically converges to the average consensus but more slowly than UBGA. We conclude that UBGA strikes a desirable balance between converging quickly while achieving consensus on the average.

\begin{figure*}[!t]
\centering{\subfigure[$n=50$]{\includegraphics[width=2in]{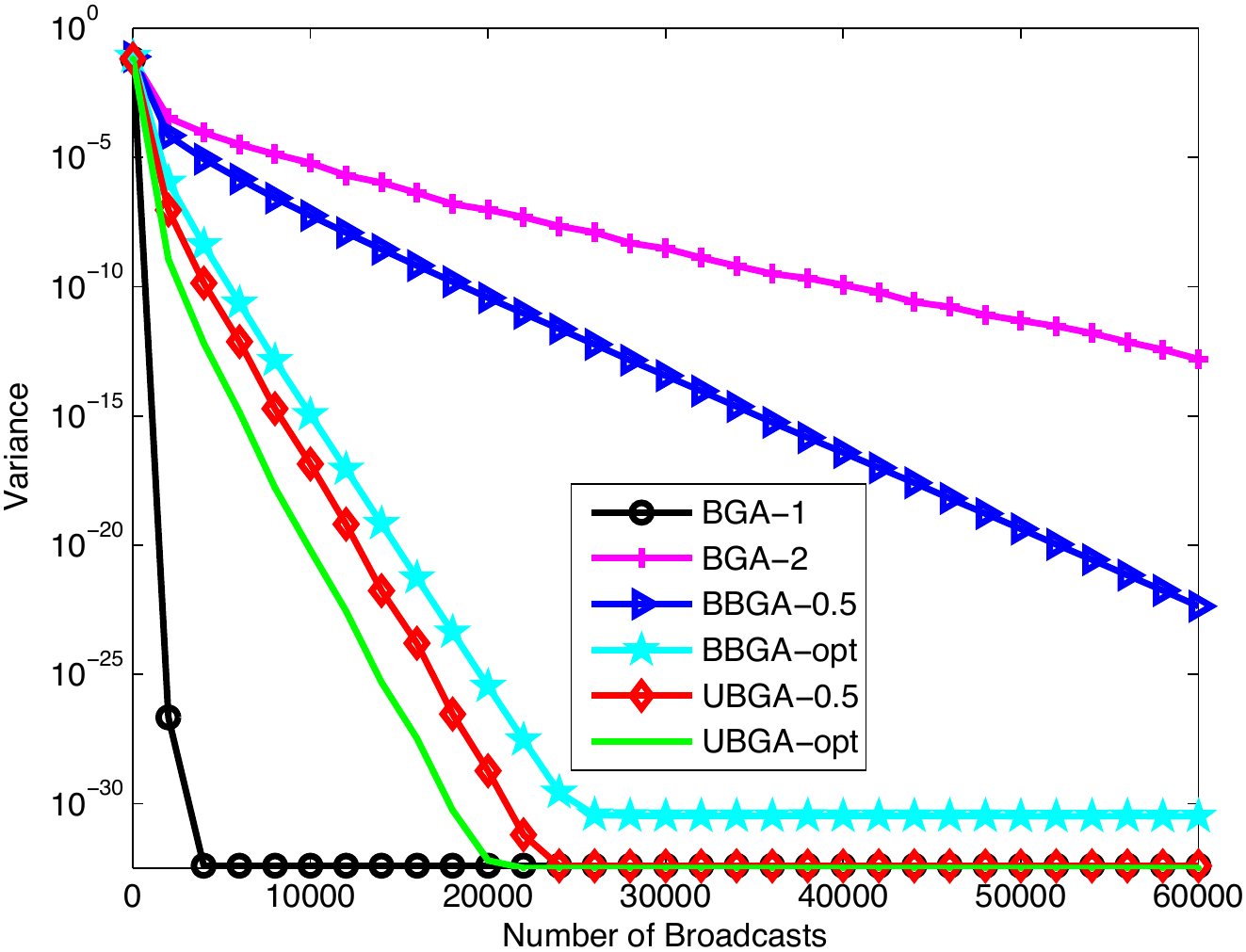}
\label{uniform:variance for undirected graphs with 50 nodes}}
\hfil
\subfigure[$n=100$]{\includegraphics[width=2in]{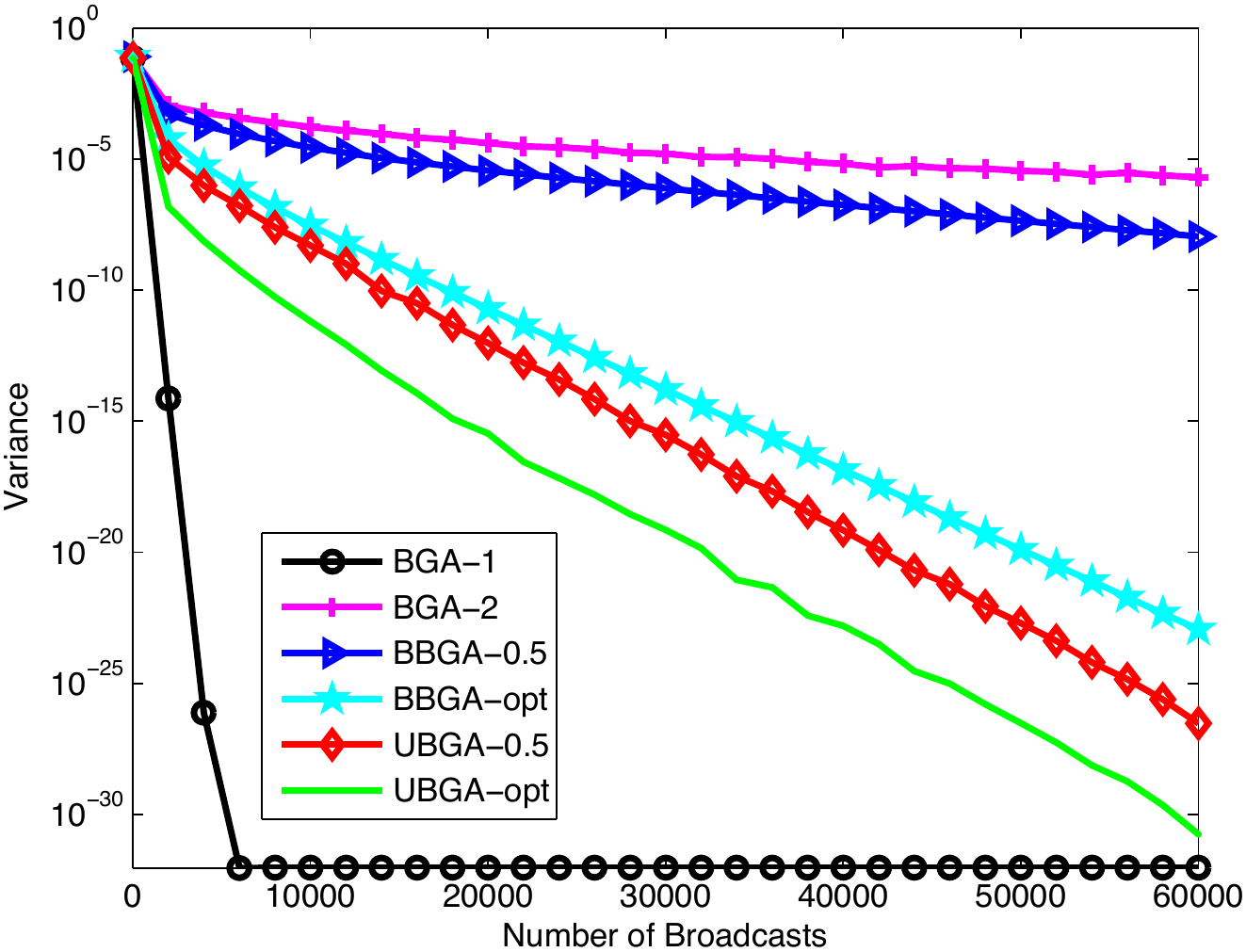}
\label{uniform:variance for undirected graphs with 100 nodes}}
\hfil
\subfigure[$n=500$]{\includegraphics[width=2in]{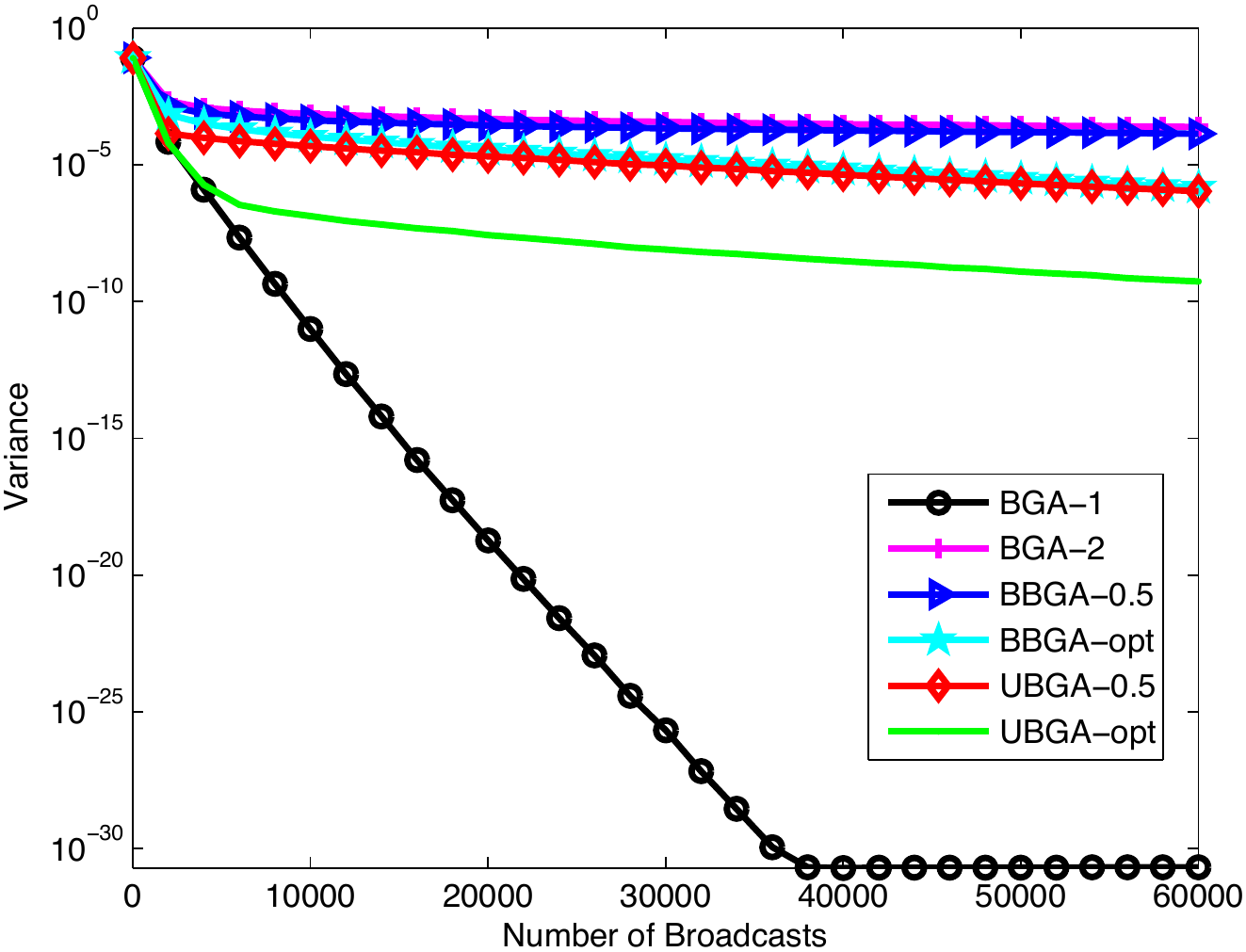}
\label{uniform:ariance for undirected graphs with 500 nodes}}}
\caption{The deviation of BGA-1, BGA-2, BBGA and UBGA with respect to the number of broadcasts on undirected random geometric graphs with uniform distribution for initial values.}
\label{uniform:sim_var_und}
\end{figure*}

\begin{figure*}[!t]
\centering{\subfigure[$n=50$]{\includegraphics[width=2in]{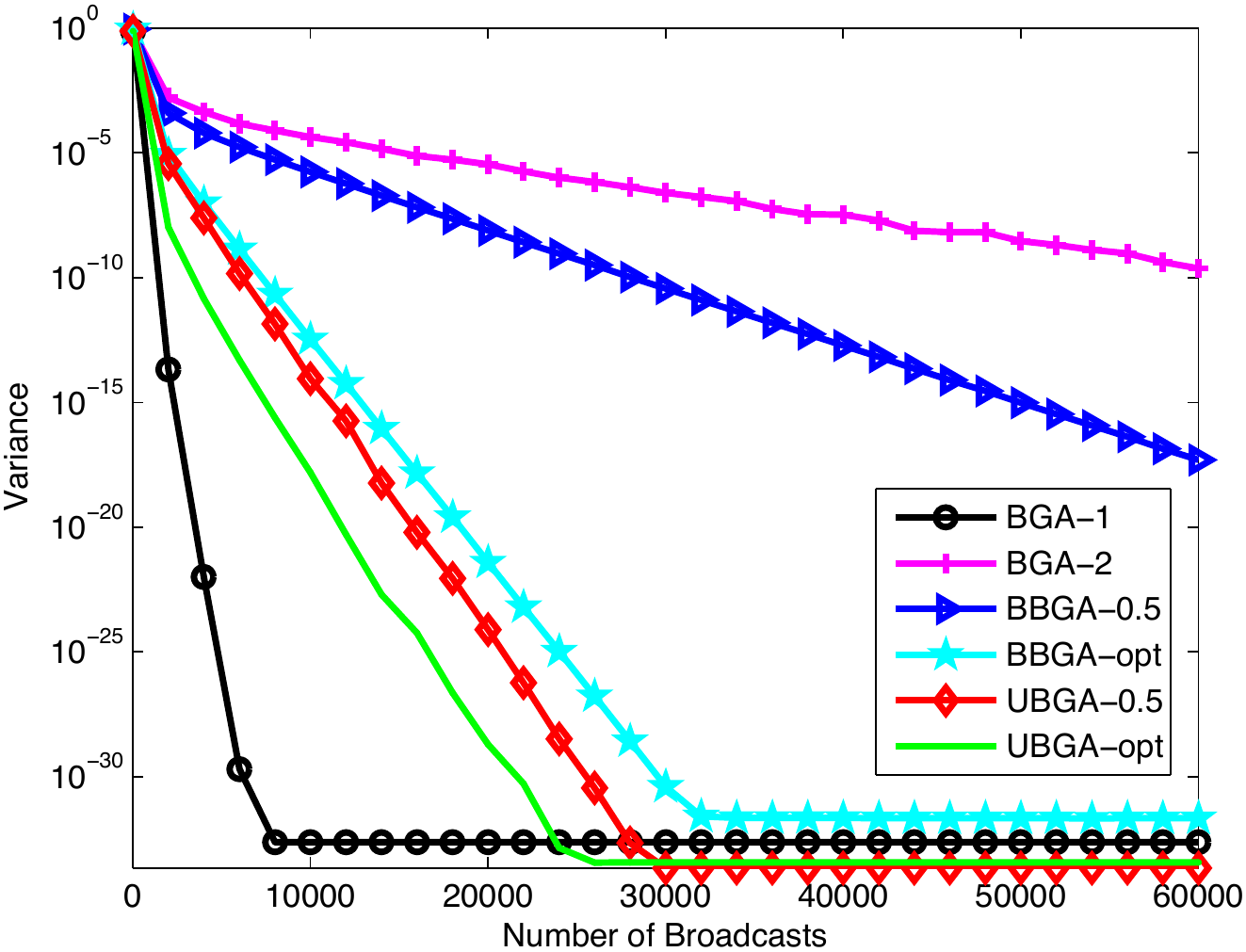}
\label{gaussian:variance for undirected graphs with 50 nodes}}
\hfil
\subfigure[$n=100$]{\includegraphics[width=2in]{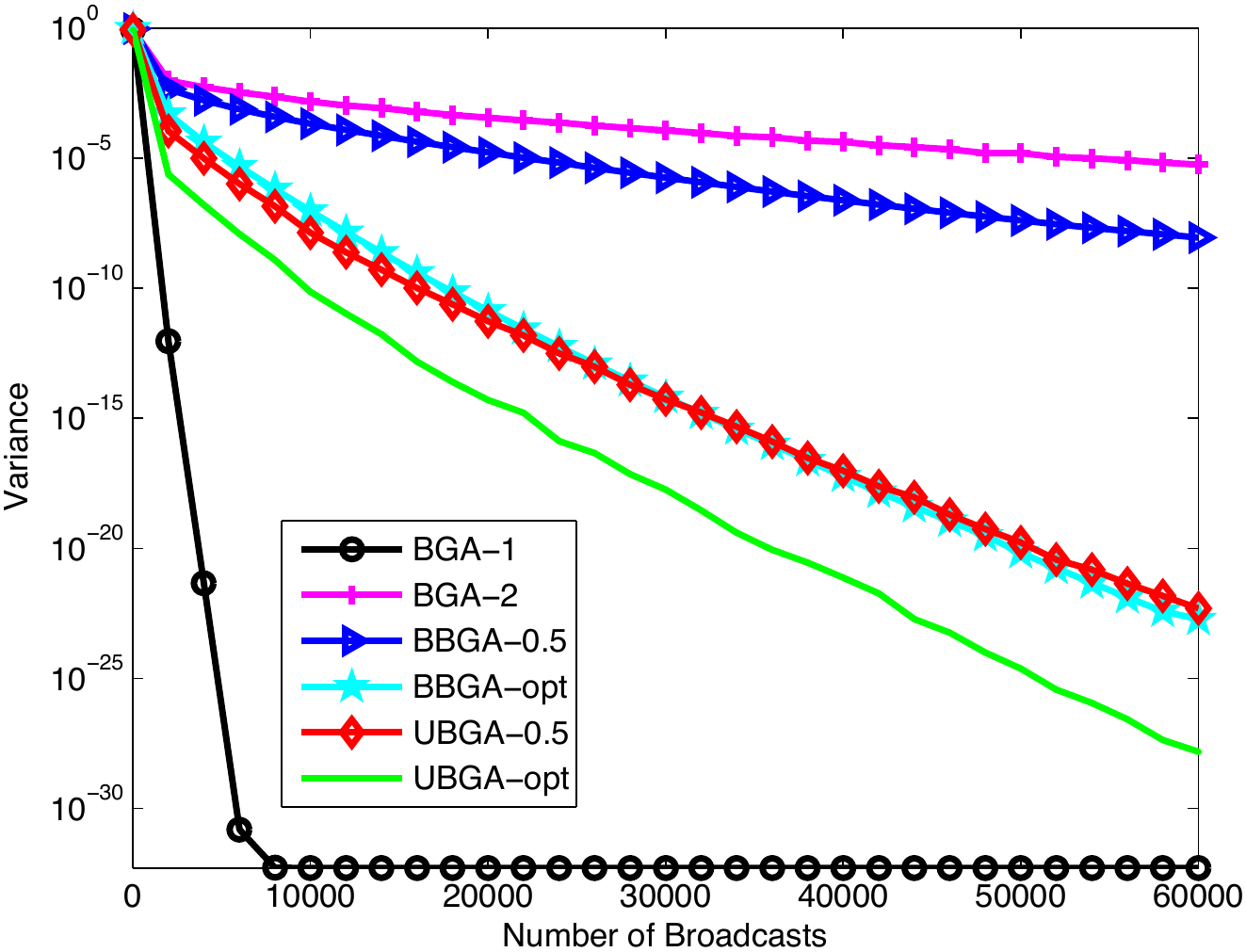}
\label{gaussian:variance for undirected graphs with 100 nodes}}
\hfil
\subfigure[$n=500$]{\includegraphics[width=2in]{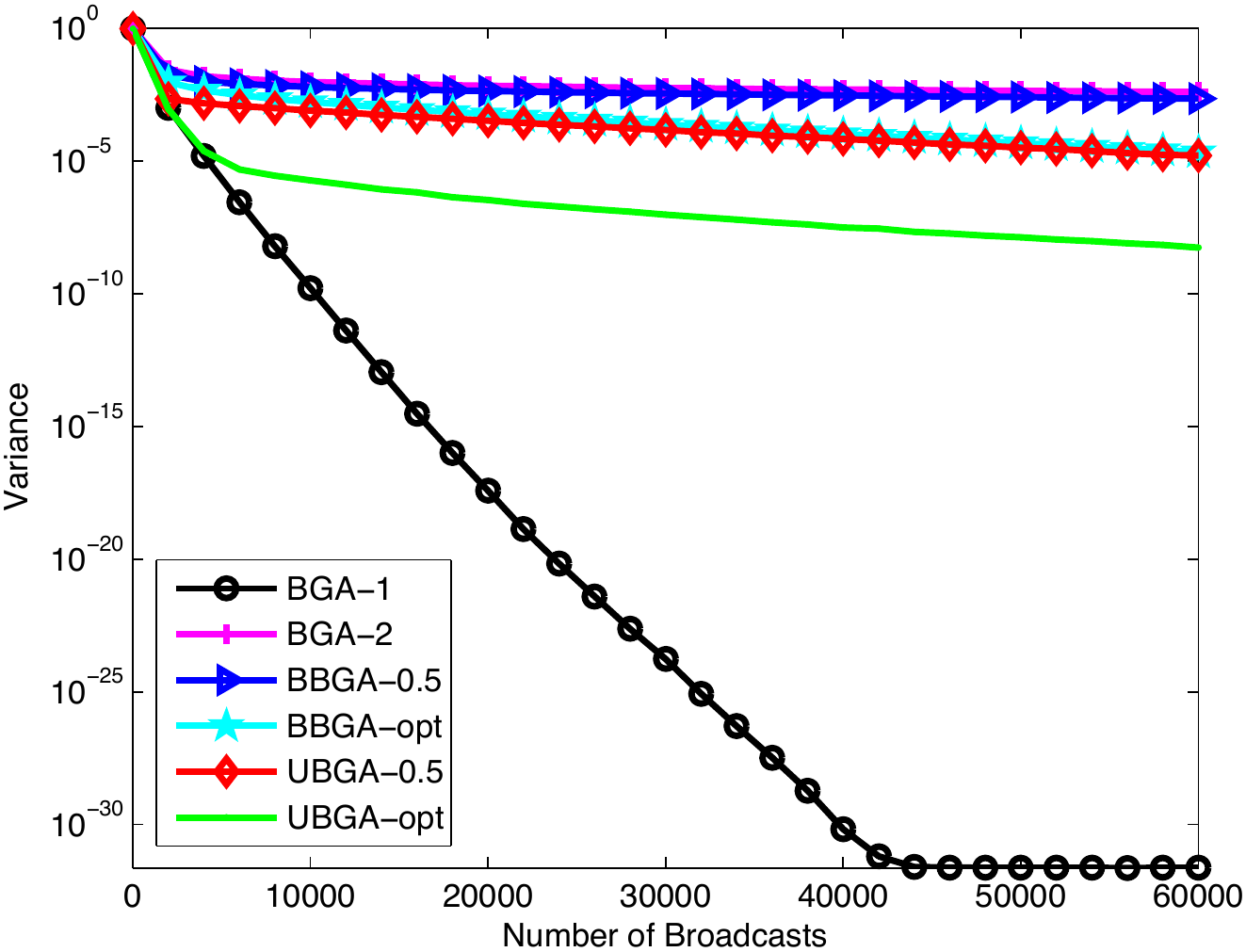}
\label{gaussian:ariance for undirected graphs with 500 nodes}}}
\caption{The deviation of BGA-1, BGA-2, BBGA and UBGA with respect to the number of broadcasts on undirected random geometric graphs with Gaussian initial values.}
\label{gaussian:sim_var_und}
\end{figure*}

\begin{figure*}[!t]
\centering{\subfigure[$n=50$]{\includegraphics[width=2in]{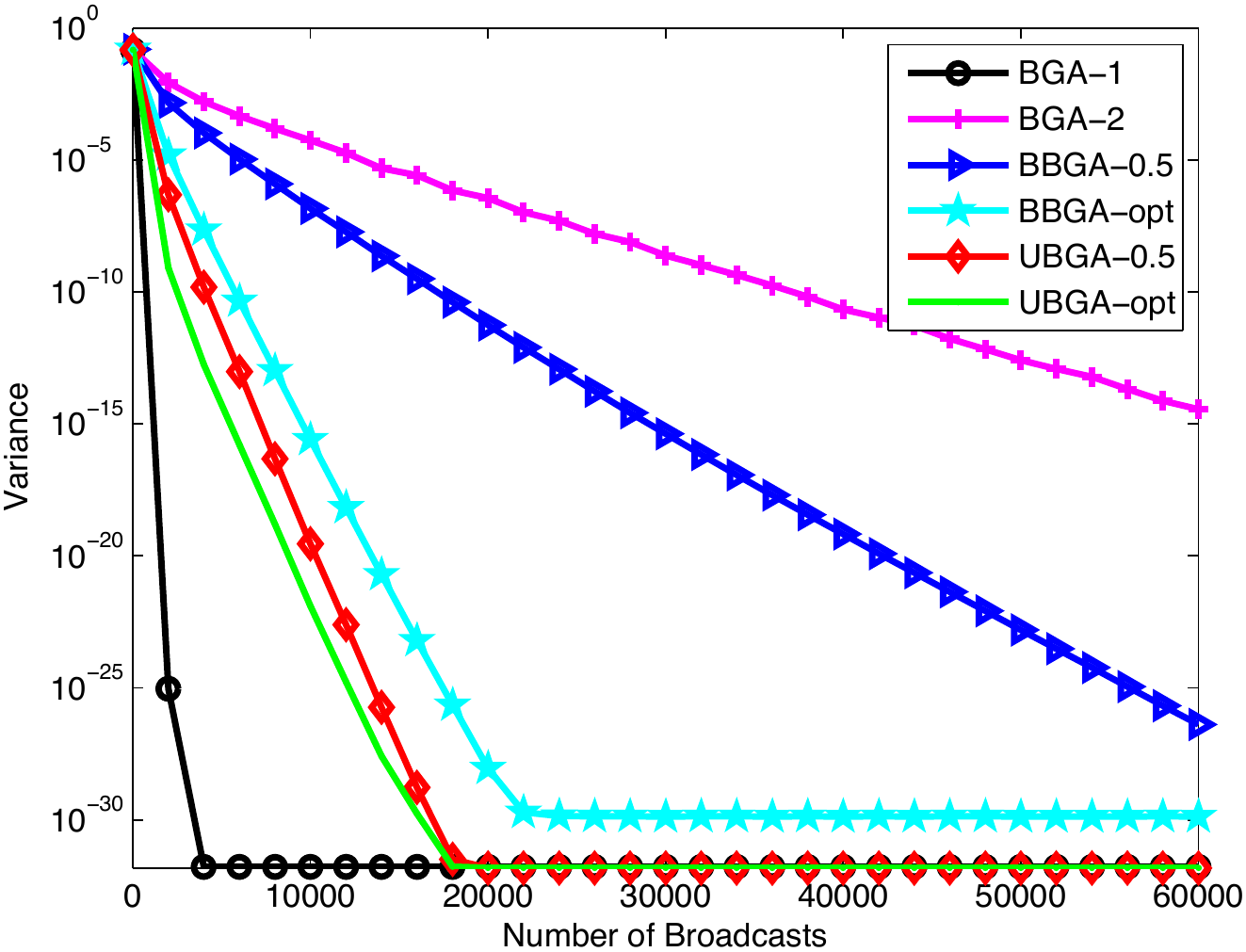}
\label{slope:variance for undirected graphs with 50 nodes}}
\hfil
\subfigure[$n=100$]{\includegraphics[width=2in]{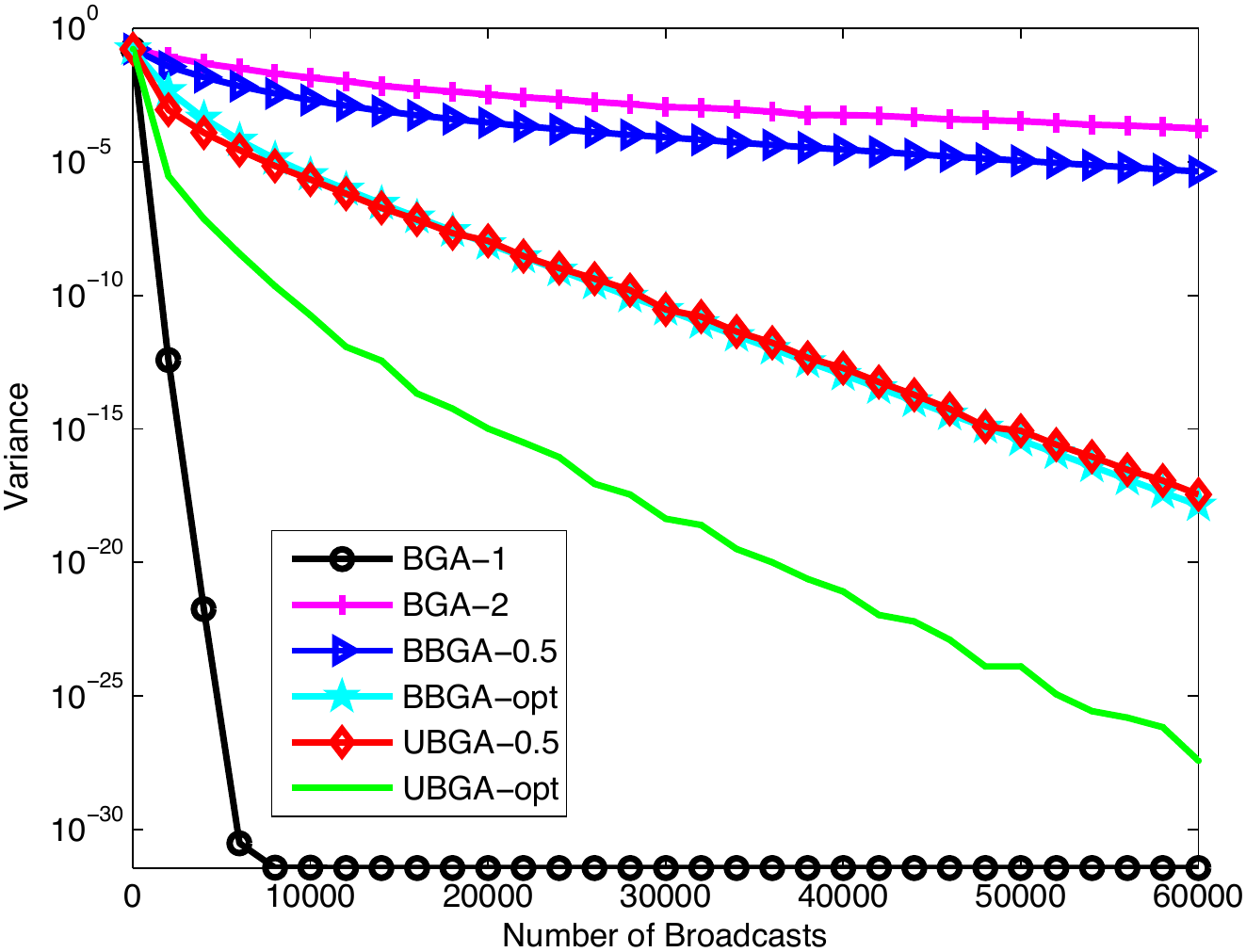}
\label{slope:variance for undirected graphs with 100 nodes}}
\hfil
\subfigure[$n=500$]{\includegraphics[width=2in]{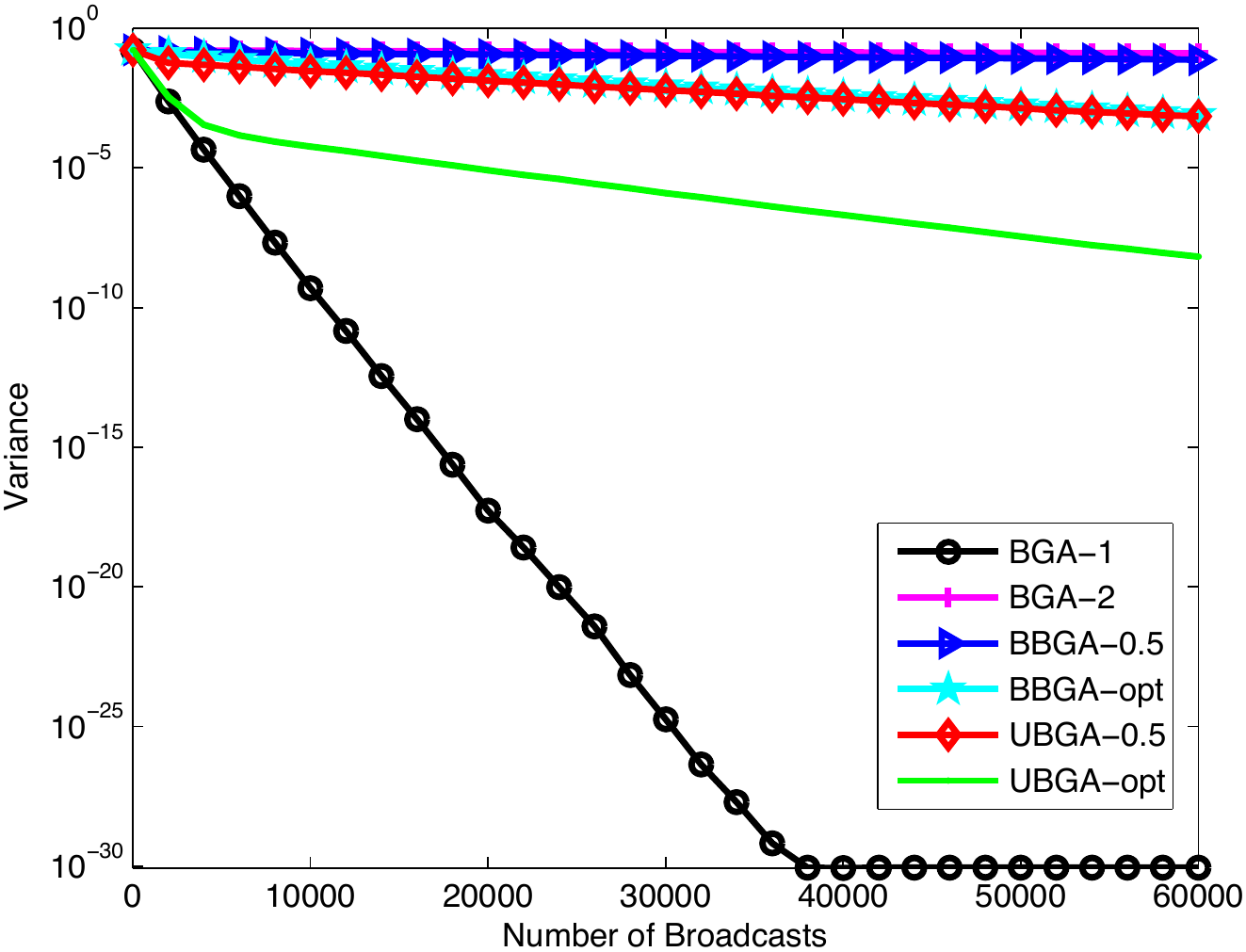}
\label{slope:ariance for undirected graphs with 500 nodes}}}
\caption{The deviation of BGA-1, BGA-2, BBGA and UBGA with respect to the number of broadcasts on undirected random geometric graphs with slope initial values.}
\label{slope:sim_var_und}
\end{figure*}

\begin{figure*}[!t]
\centering{\subfigure[$n=50$]{\includegraphics[width=2in]{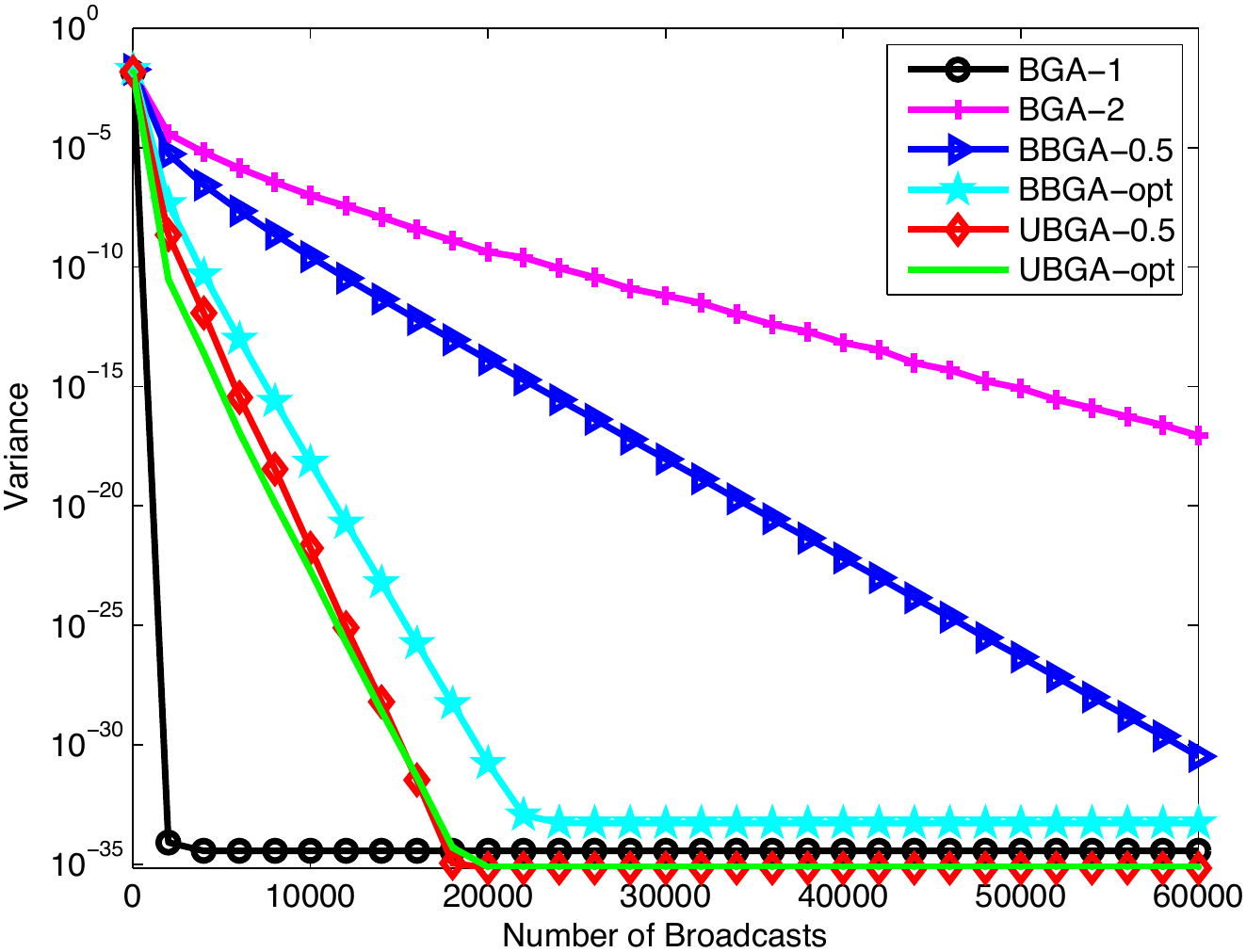}
\label{spike:variance for undirected graphs with 50 nodes}}
\hfil
\subfigure[$n=100$]{\includegraphics[width=2in]{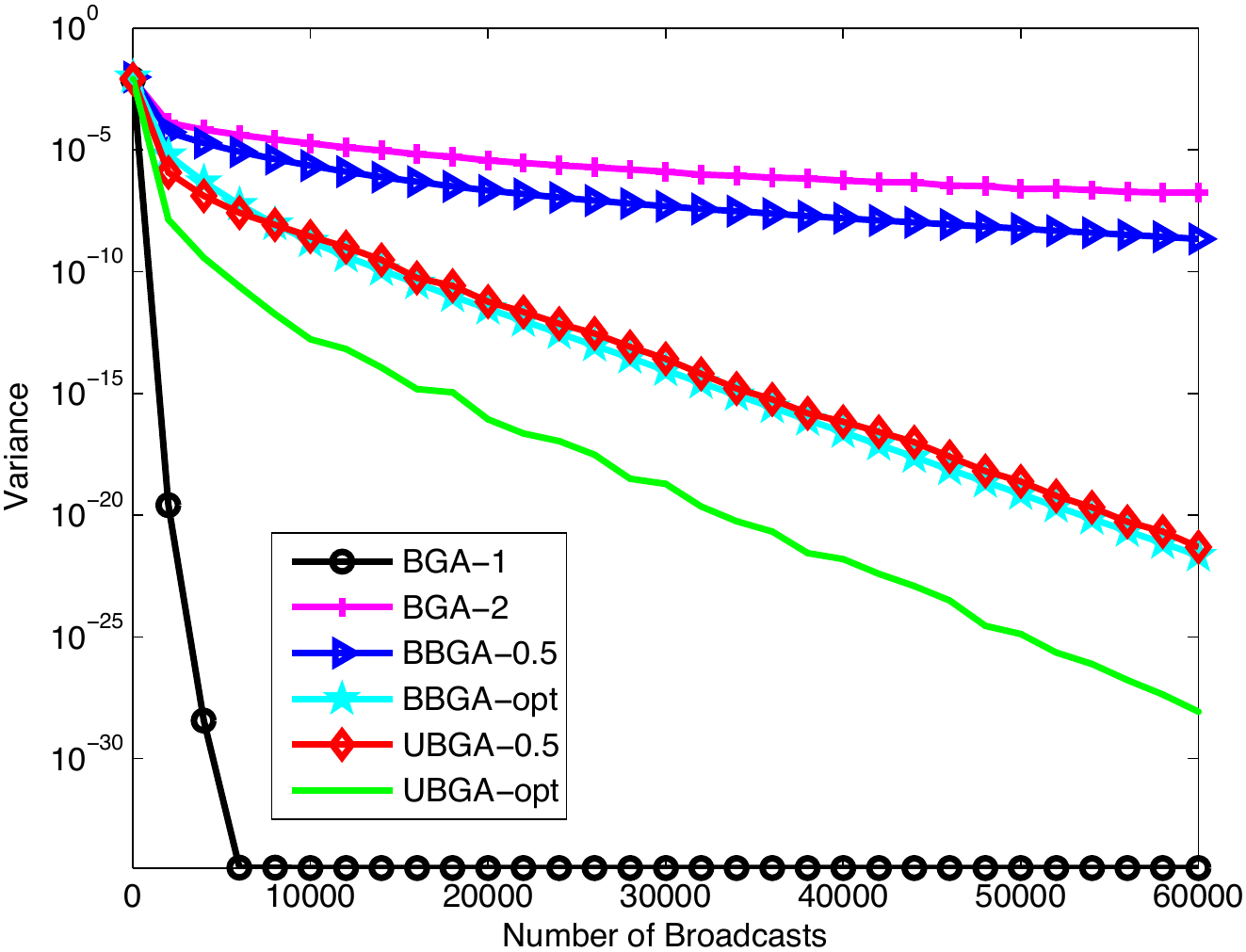}
\label{spike:variance for undirected graphs with 100 nodes}}
\hfil
\subfigure[$n=500$]{\includegraphics[width=2in]{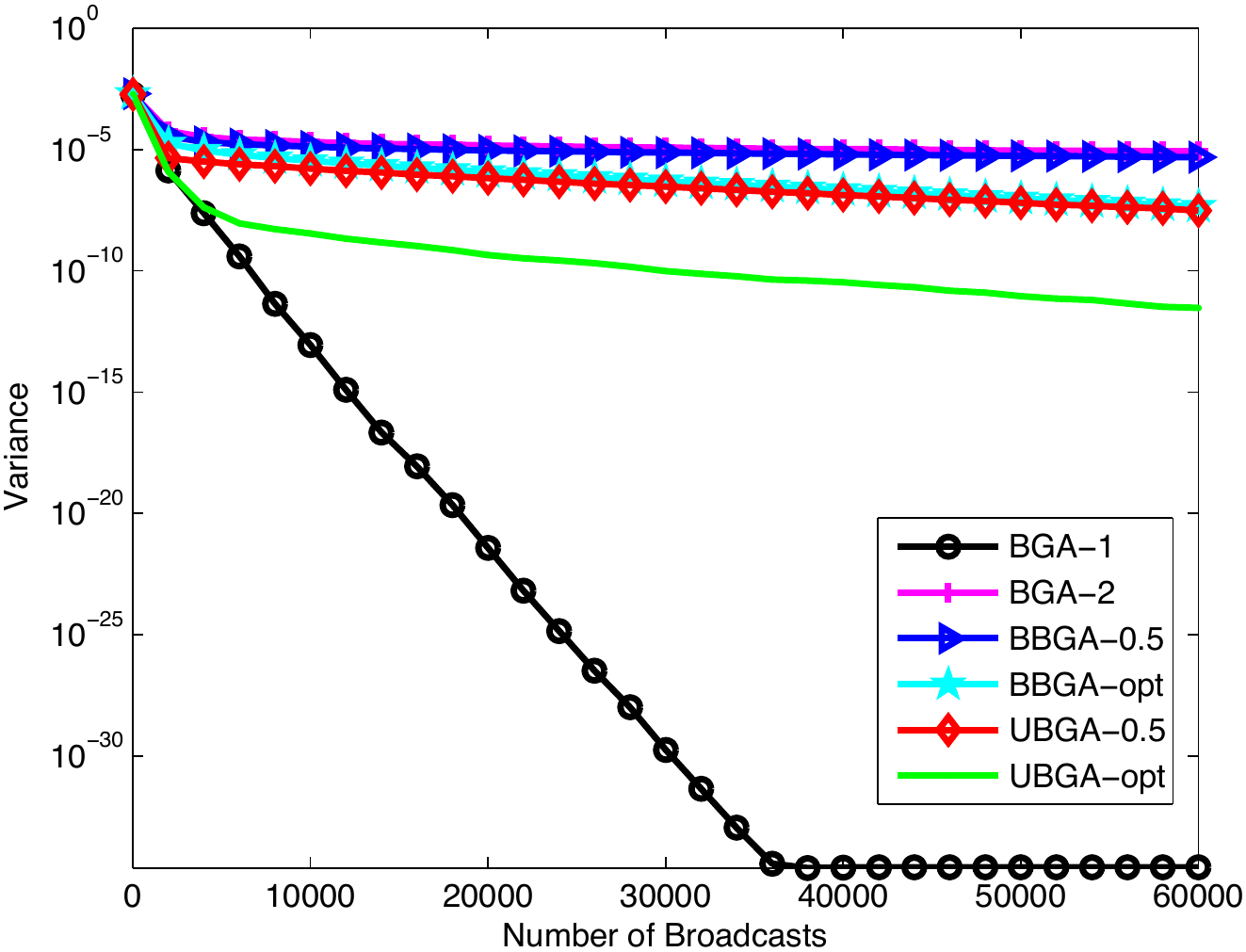}
\label{spike:ariance for undirected graphs with 500 nodes}}}
\caption{The deviation of BGA-1, BGA-2, BBGA and UBGA with respect to the number of broadcasts on undirected random geometric graphs with spike initial values.}
\label{spike:sim_var_und}
\end{figure*}

\subsection{Mean Squared Error}

Figs.~\ref{uniform:mse-und}--\ref{spike:mse-und} show the mean squared error $r(t) = \frac{1}{n} \|x(t) - \frac{1}{n}\mathbf{1}\mathbf{1}^T x(0)\|_2^2$ as a function of $t$ for all four algorithms and the four different initializations on networks of $n=50$, $100$, and $500$ nodes. UBGA generally has the best performance among all algorithms, in the sense that a small deviation is achieved with relatively few broadcasts. BGA-1 has a high deviation; it is well-known that it converges quickly but that it does not converge to the average consensus. When out-degree information is available UBGA is preferable. For networks with $n=50$ or $100$ nodes, using $\epsilon = 0.5$ is close enough to optimal that the performance is extremely good for UBGA-0.5. For larger graphs, the performance of UBGA-opt dominates that of UBGA-0.5. An interesting open problem is to come up with a better practical guideline for setting $\epsilon$ as a function of network size and structure, e.g., for random geometric graphs.

It is interesting to note that BBGA has better performance than BGA-2 for a smaller number of broadcasts. Since BGA-2 converges to the average consensus in most examples, but BBGA does not, this indicates that BGA-2 converges slower than BBGA. For larger networks BBGA may be preferable as an alternative which quickly reaches a reasonably accurate solution.

\begin{figure*}[!t]
\centering{\subfigure[$n=50$]{\includegraphics[width=2in]{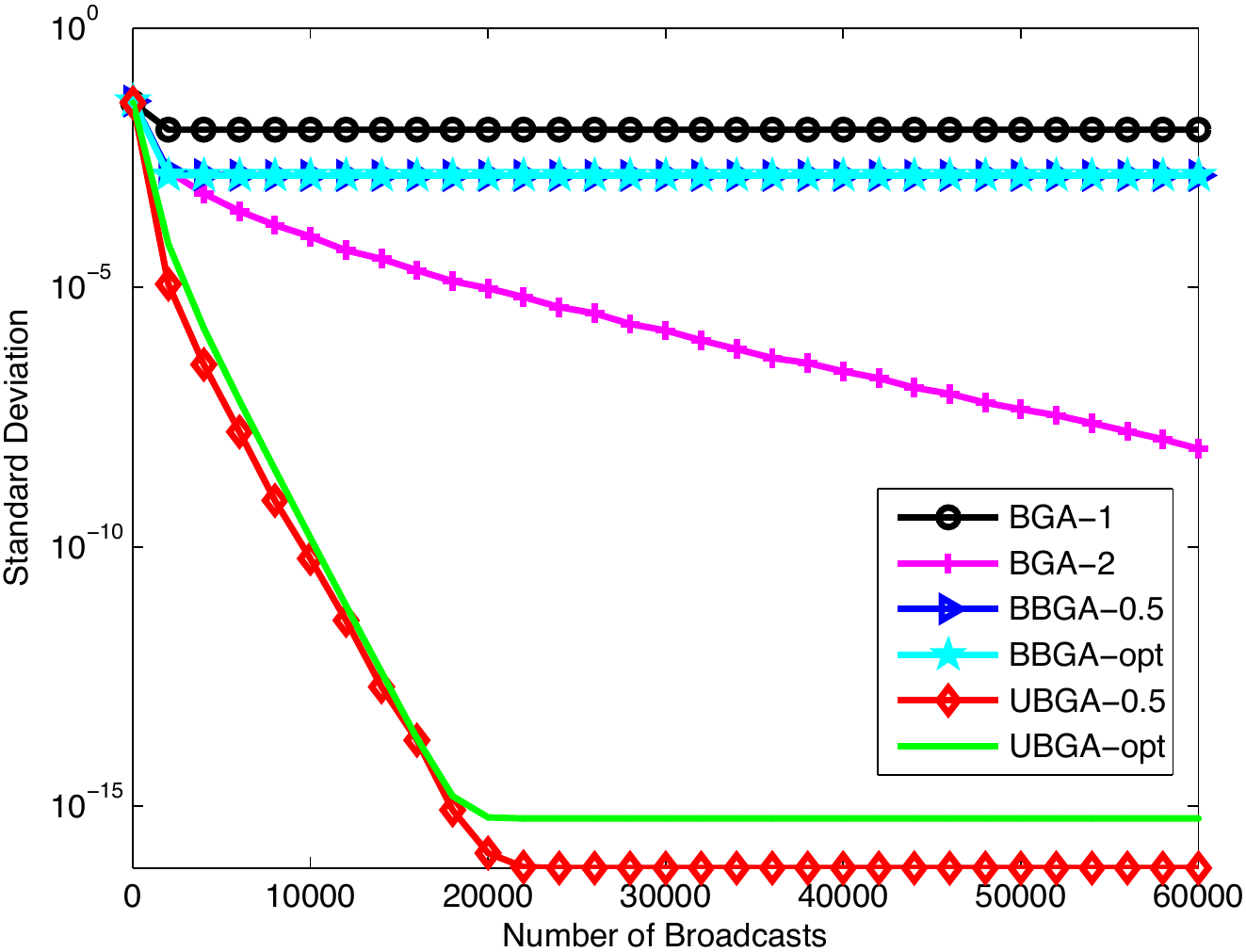}
\label{uniform:mse for undirected graphs with 50 nodes}}
\hfil
\subfigure[$n=100$]{\includegraphics[width=2in]{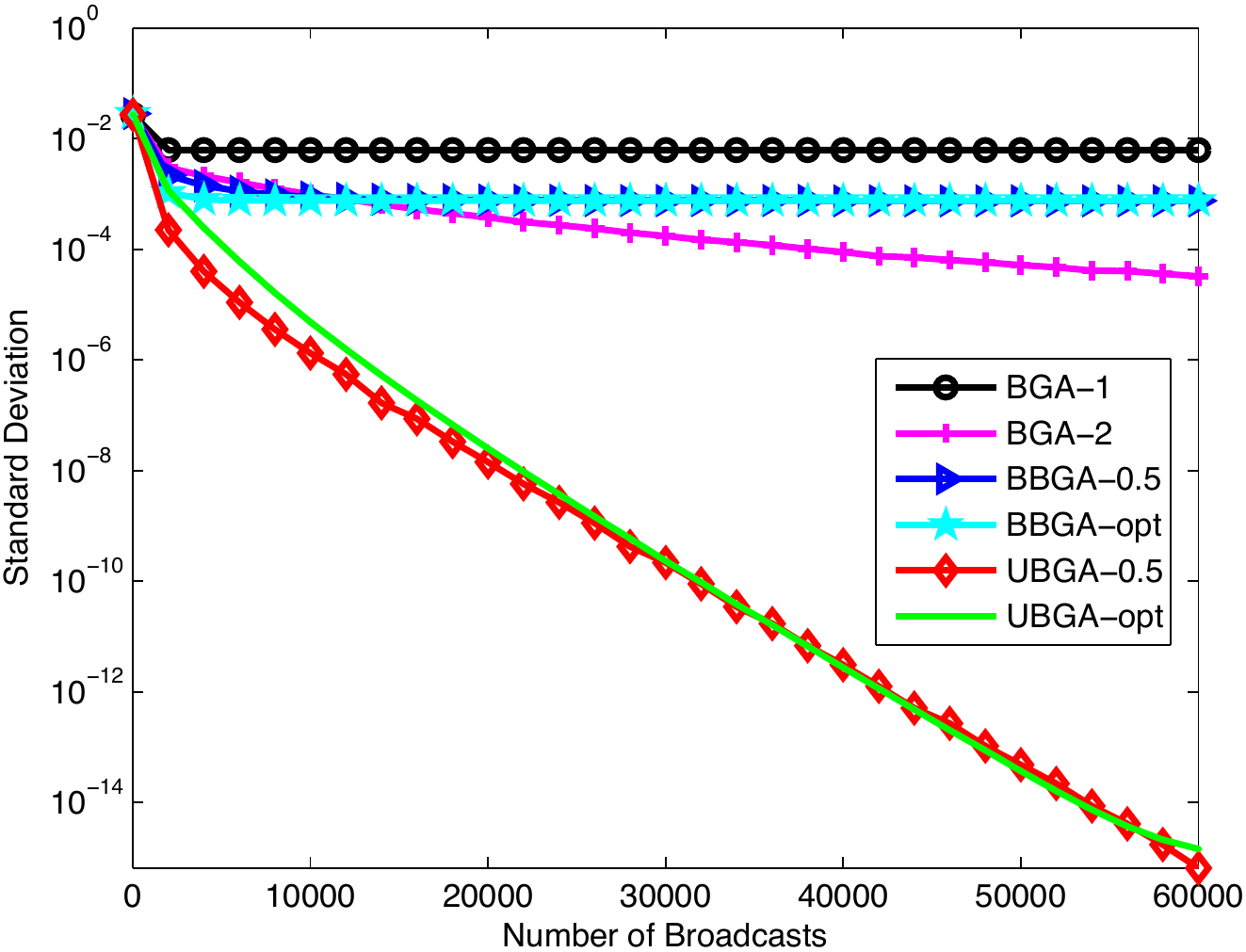}
\label{uniform:mse for undirected graphs with 100 nodes}}
\hfil
\subfigure[$n=500$]{\includegraphics[width=2in]{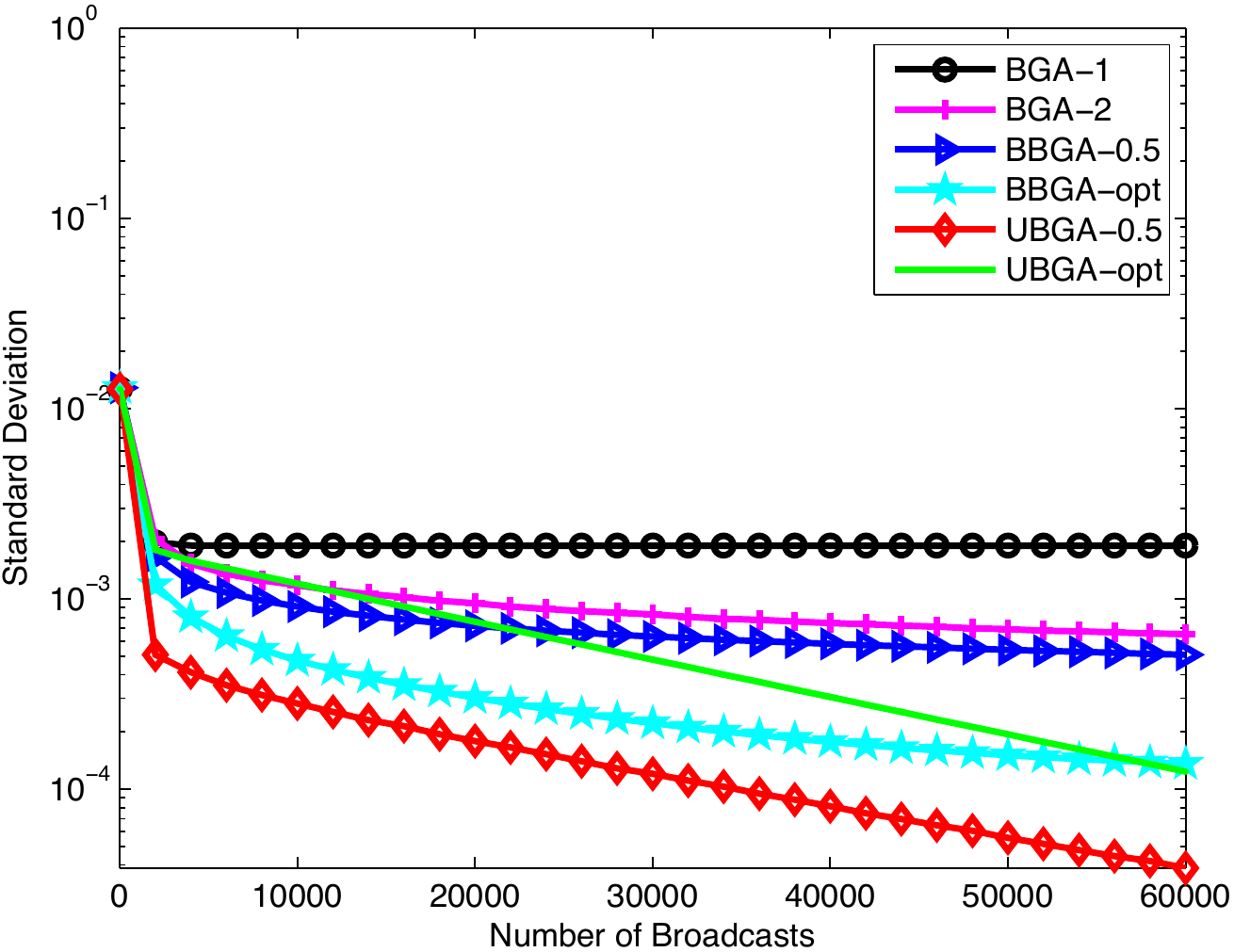}
\label{uniform:mse for undirected graphs with 500 nodes}}}
\caption{The mean squared error of BGA-1, BGA-2, BBGA and UBGA with respect to the number of broadcasts on undirected random geometric graphs with uniform initial values.}
\label{uniform:mse-und}
\end{figure*}

\begin{figure*}[!t]
\centering{\subfigure[$n=50$]{\includegraphics[width=2in]{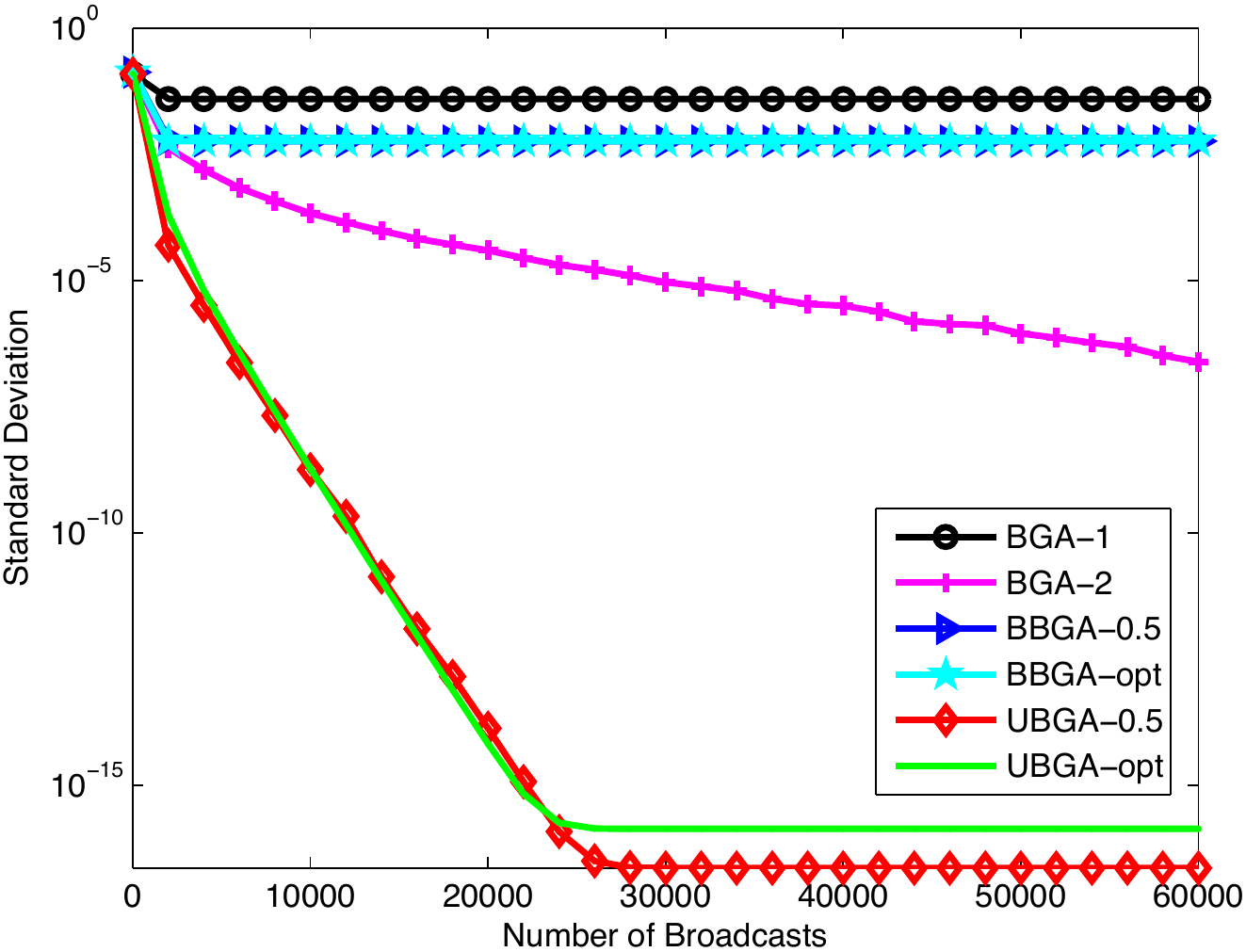}
\label{gaussian:mse for undirected graphs with 50 nodes}}
\hfil
\subfigure[$n=100$]{\includegraphics[width=2in]{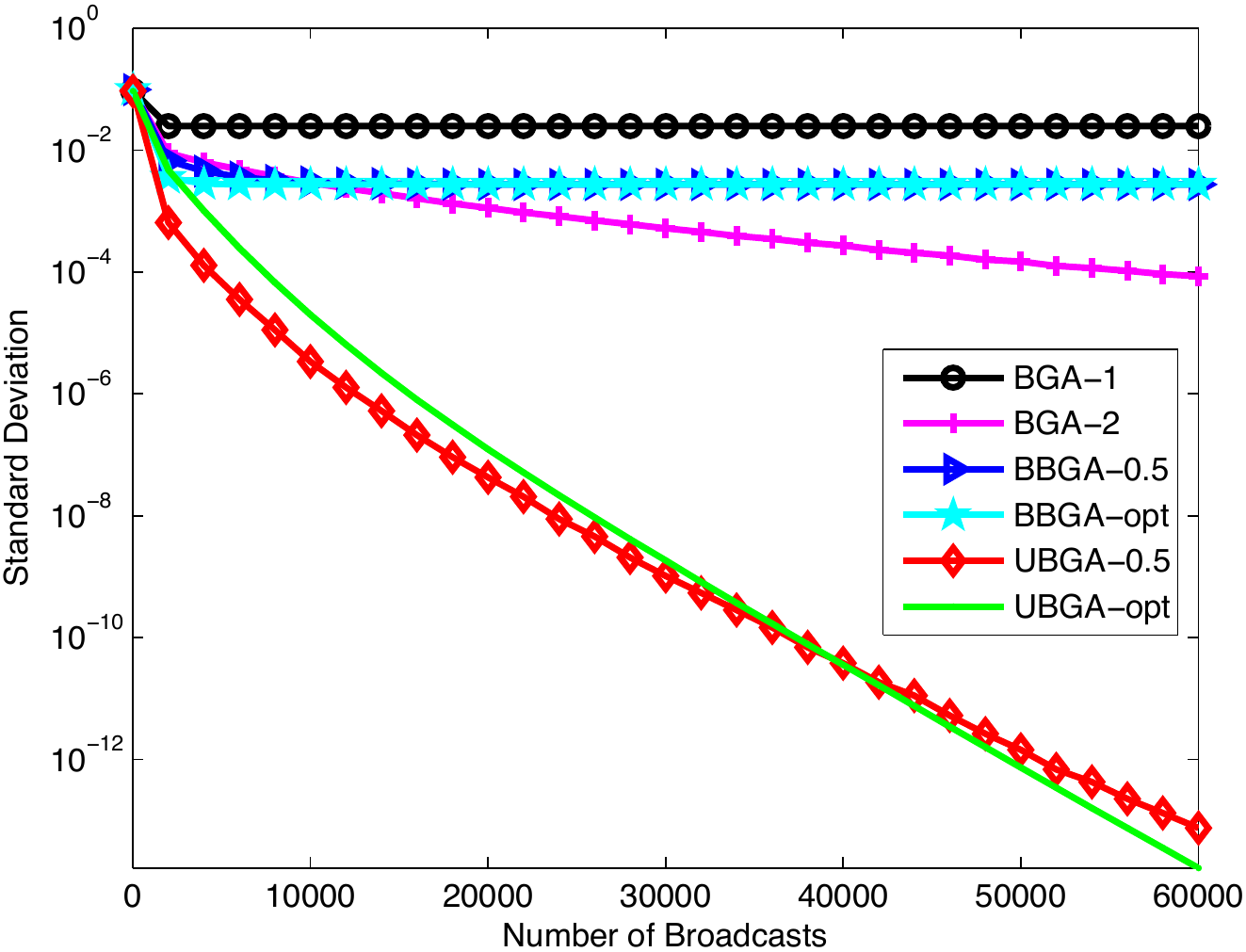}
\label{gaussian:mse for undirected graphs with 100 nodes}}
\hfil
\subfigure[$n=500$]{\includegraphics[width=2in]{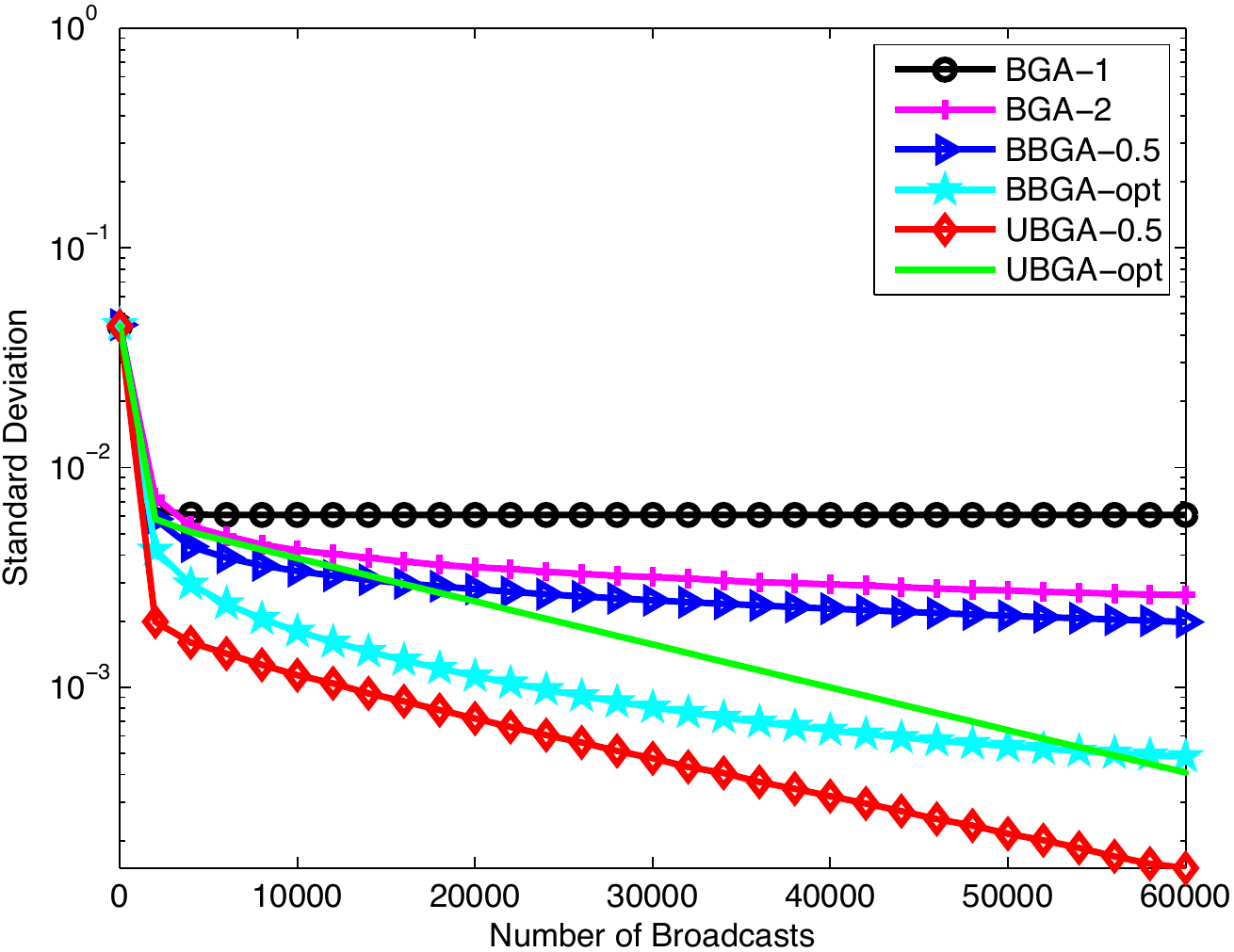}
\label{gaussian:mse for undirected graphs with 500 nodes}}}
\caption{The mean squared error of BGA-1, BGA-2, BBGA and UBGA with respect to the number of broadcasts on undirected random geometric graphs with Gaussian initial values.}
\label{guassian:mse-und}
\end{figure*}

\begin{figure*}[!t]
\centering{\subfigure[$n=50$]{\includegraphics[width=2in]{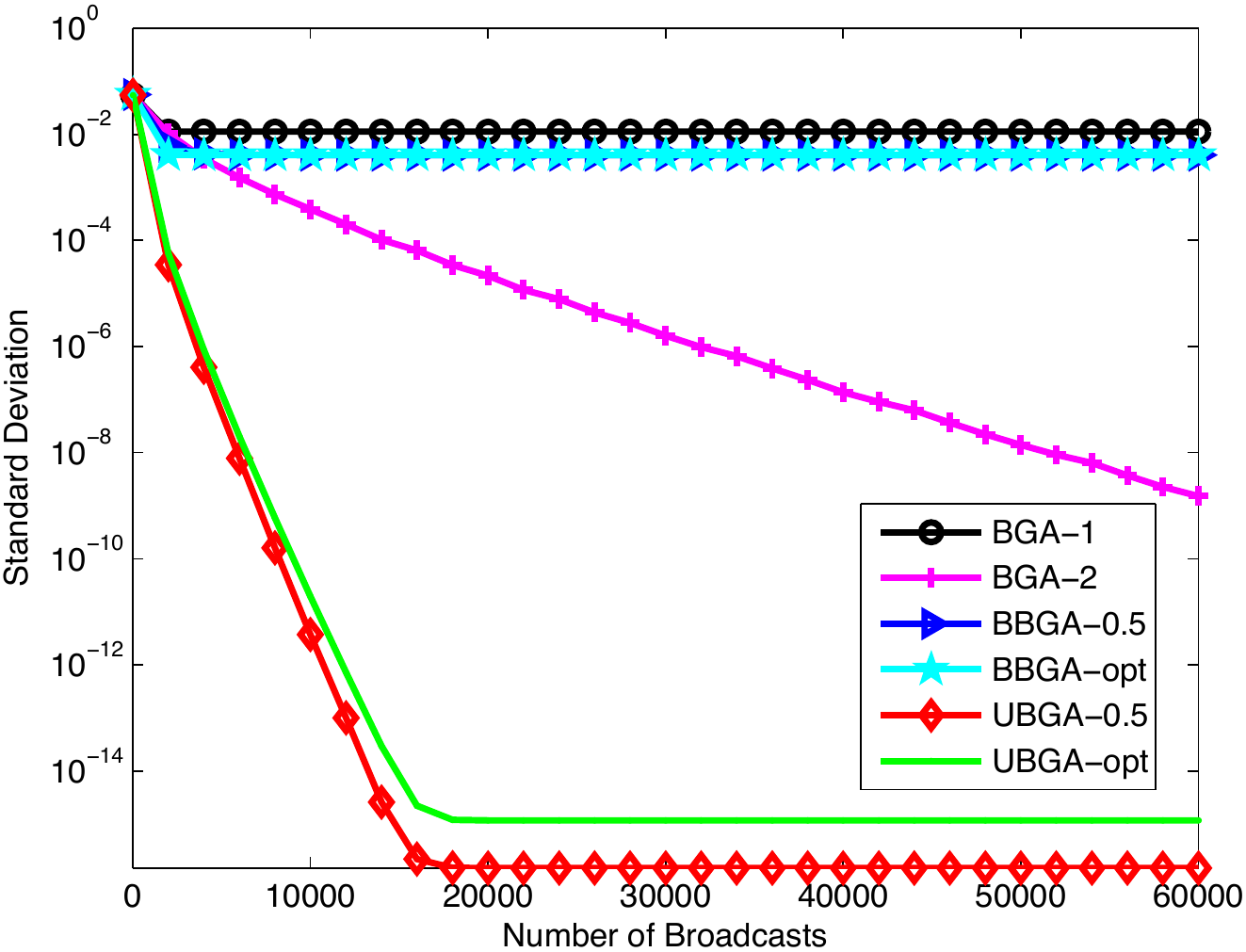}
\label{slope:mse for undirected graphs with 50 nodes}}
\hfil
\subfigure[$n=100$]{\includegraphics[width=2in]{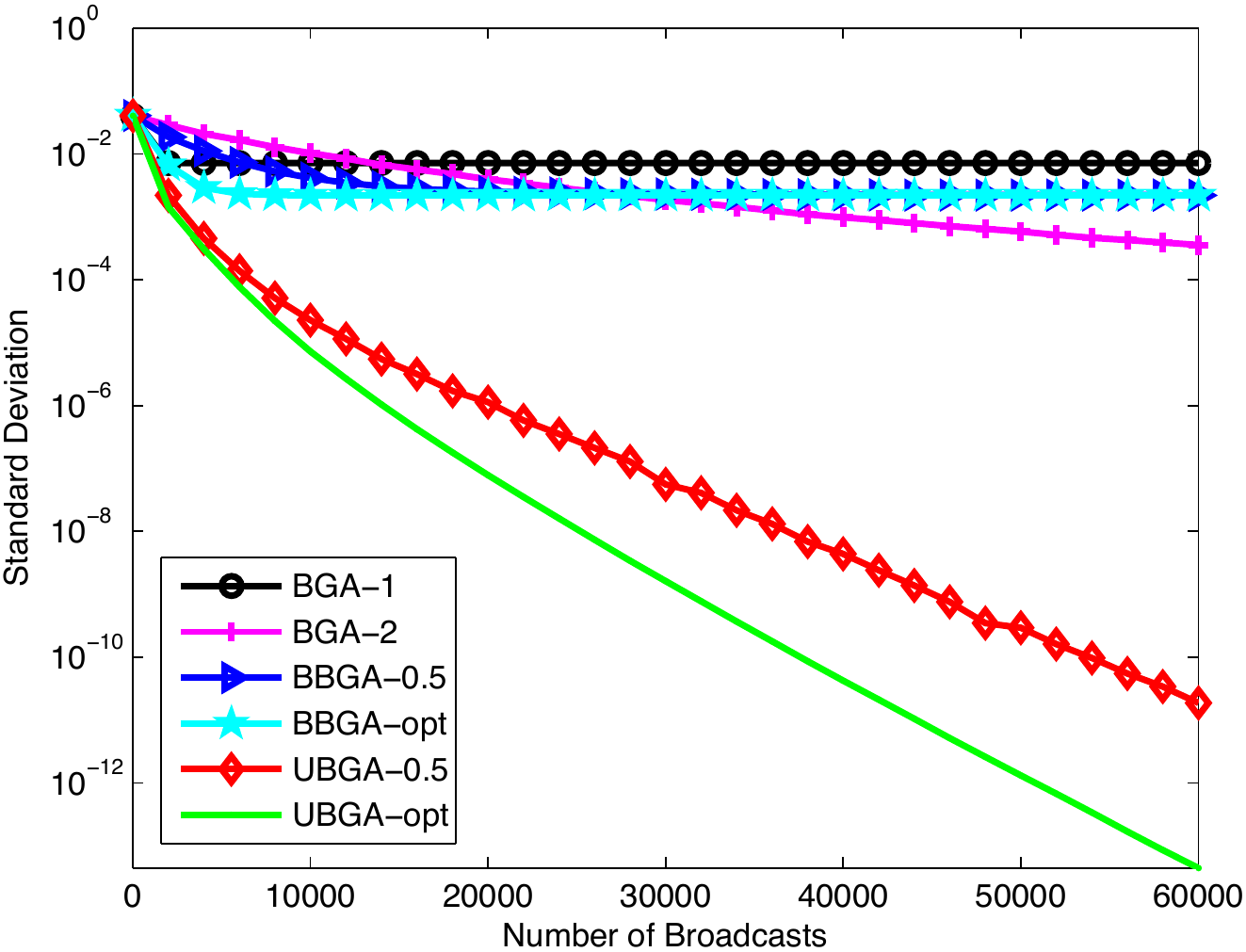}
\label{slope:mse for undirected graphs with 100 nodes}}
\hfil
\subfigure[$n=500$]{\includegraphics[width=2in]{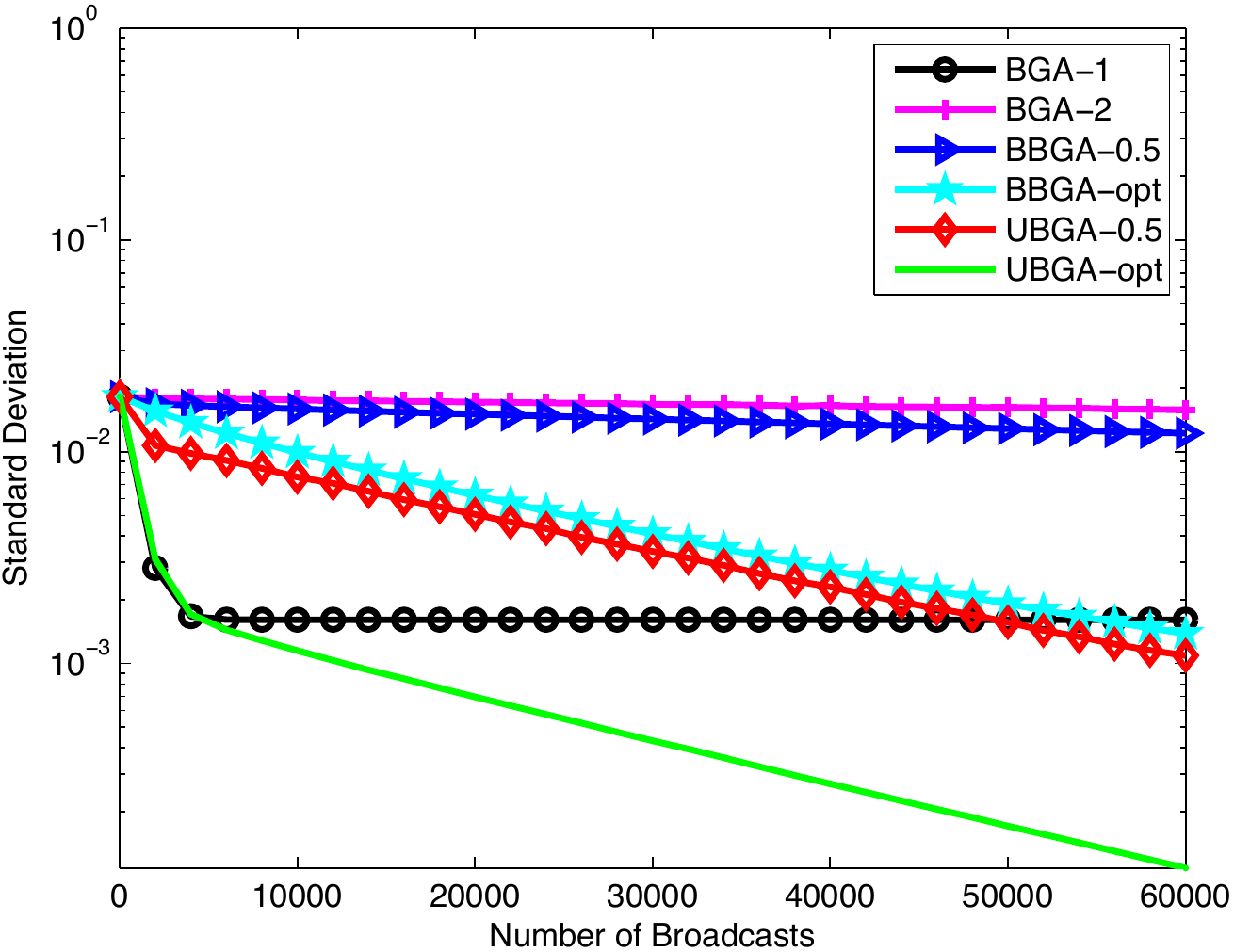}
\label{slope:mse for undirected graphs with 500 nodes}}}
\caption{The mean squared error of BGA-1, BGA-2, BBGA and UBGA with respect to the number of broadcasts on undirected random geometric graphs with slope initial values.}
\label{slope:mse-und}
\end{figure*}

\begin{figure*}[!t]
\centering{\subfigure[$n=50$]{\includegraphics[width=2in]{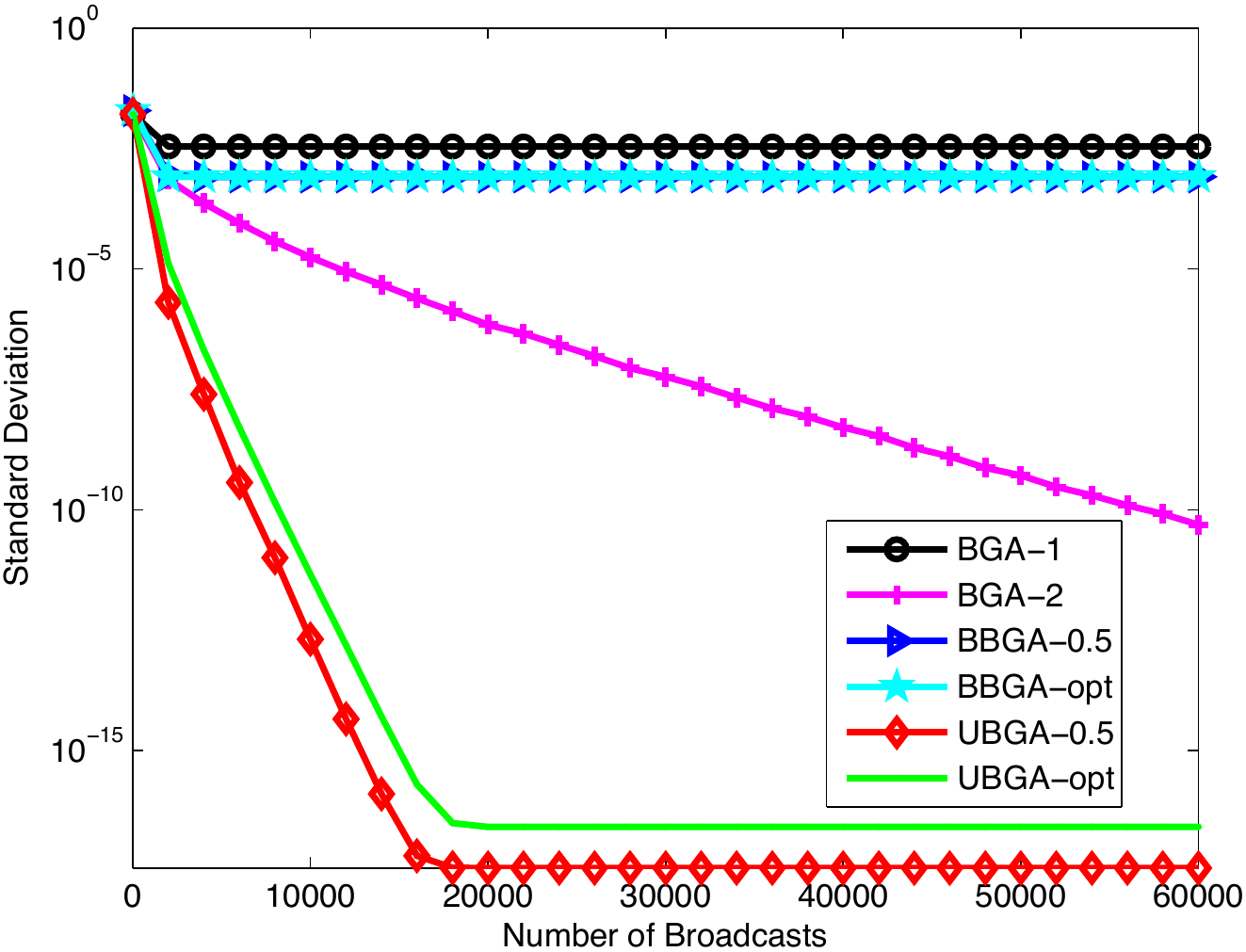}
\label{spike:mse for undirected graphs with 50 nodes}}
\hfil
\subfigure[$n=100$]{\includegraphics[width=2in]{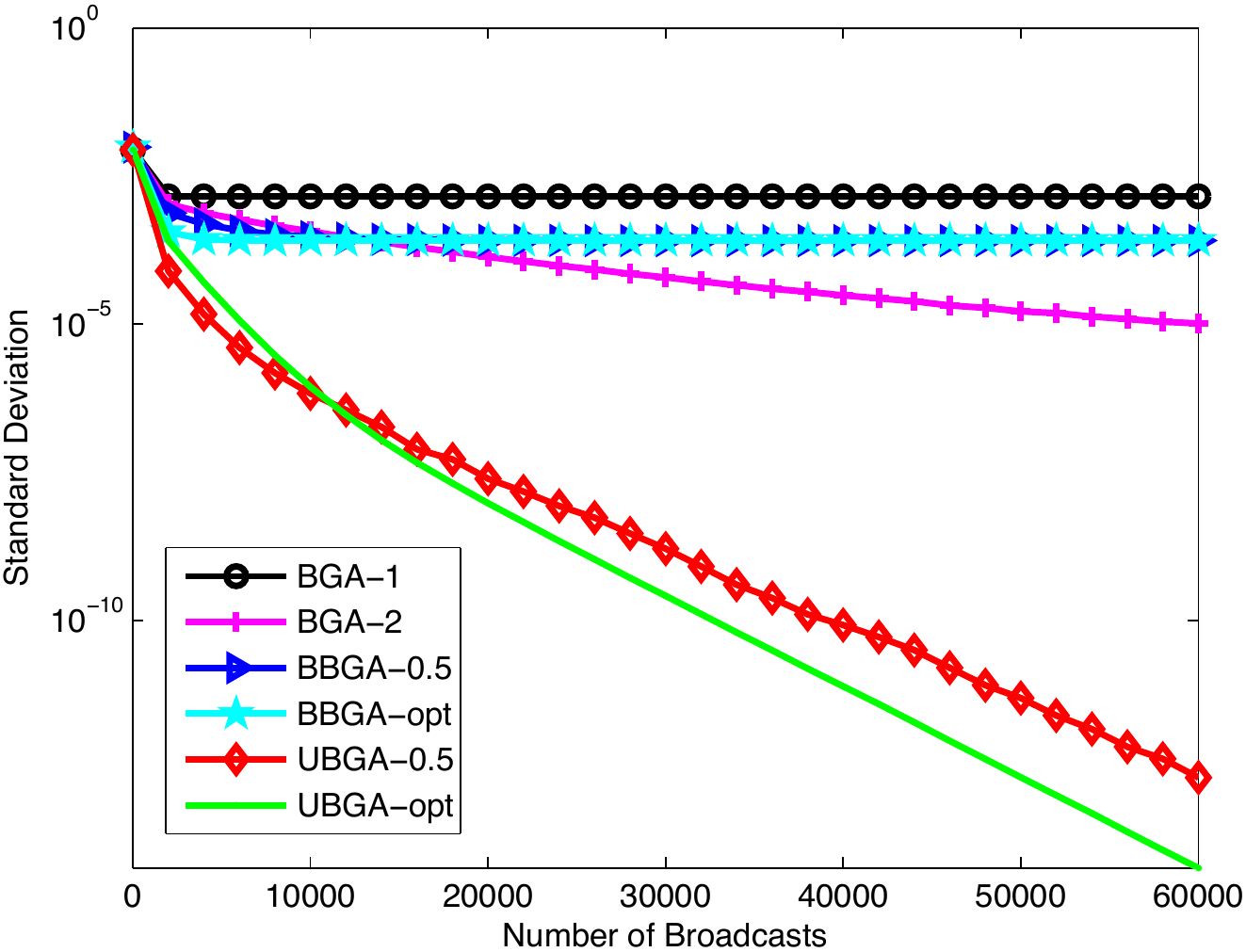}
\label{spike:mse for undirected graphs with 100 nodes}}
\hfil
\subfigure[$n=500$]{\includegraphics[width=2in]{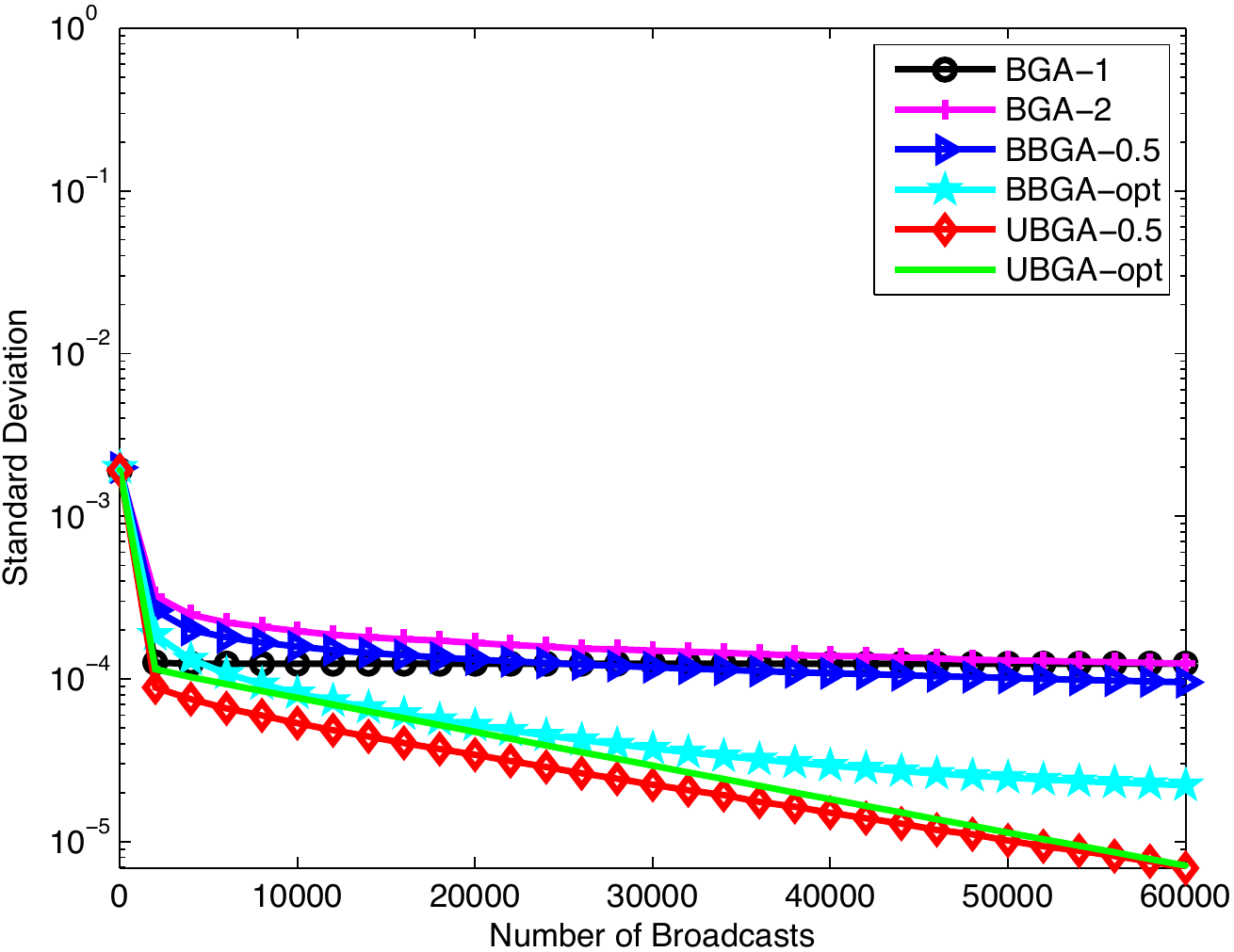}
\label{spike:mse for undirected graphs with 500 nodes}}}
\caption{The mean squared error of BGA-1, BGA-2, BBGA and UBGA with respect to the number of broadcasts on undirected random geometric graphs with spike initial values.}
\label{spike:mse-und}
\end{figure*}

\section{Conclusion and future work} \label{sec:conclusion}

In this paper, we propose a framework for broadcast gossip algorithms and prove that consensus is achieved in both expectation and in the mean squared sense for reasonably chosen coefficients. Then we analyze two particular broadcast gossip algorithms, UBGA and BBGA, where the former preserves the average and the latter is more practical for implementation in non-symmetric broadcast networks. These algorithms have an interpretation from the perspective of matrix perturbation, and we derive an upper bound on the perturbation parameter under which convergence is guaranteed. We also study the optimal value of the perturbation parameter and find it is within the range of allowable values. By numerical analysis, the optimal perturbation parameter obtained from BBGA on undirected digraphs is shown to also work well on digraphs. If the out-degree information is available (as is the case in undirected networks), we demonstrate that UBGA outperforms the existing state-of-the-art broadcast gossip algorithms. When out-degree information is not available, BBGA is a promising alternative because it exhibits an excellent tradeoff between the rate of convergence and the limiting mean squared error.

Interesting future work includes studying convergence properties of broadcast gossip algorithms with quantized transmissions. The broadcast gossip algorithms proposed in this paper involve maintaining and transmitting companion variables, in addition to the state variables which are being averaged, and we are interested in understanding how the number of bits allocated to these two different values impacts the rate of convergence and limiting value.

Finally, since wireless media is shared by nodes within communication radius for each other, broadcast packets are likely to undergo collisions and interference, and it would also be interesting to develop a deeper understanding of how broadcast gossip algorithms behave under more realistic channel models (e.g., accounting for capture effects).

\bibliographystyle{IEEEtran}
\bibliography{broadcast}

\begin{thebibliography}{10}
\providecommand{\url}[1]{#1}
\csname url@samestyle\endcsname
\providecommand{\newblock}{\relax}
\providecommand{\bibinfo}[2]{#2}
\providecommand{\BIBentrySTDinterwordspacing}{\spaceskip=0pt\relax}
\providecommand{\BIBentryALTinterwordstretchfactor}{4}
\providecommand{\BIBentryALTinterwordspacing}{\spaceskip=\fontdimen2\font plus
\BIBentryALTinterwordstretchfactor\fontdimen3\font minus
  \fontdimen4\font\relax}
\providecommand{\BIBforeignlanguage}[2]{{%
\expandafter\ifx\csname l@#1\endcsname\relax
\typeout{** WARNING: IEEEtran.bst: No hyphenation pattern has been}%
\typeout{** loaded for the language `#1'. Using the pattern for}%
\typeout{** the default language instead.}%
\else
\language=\csname l@#1\endcsname
\fi
#2}}
\providecommand{\BIBdecl}{\relax}
\BIBdecl

\bibitem{DI10}
A.~Dimakis, S.~Kar, J.~Moura, M.~Rabbat, and A.~Scaglione, ``Gossip algorithms
  for distributed signal processing,'' \emph{Proceedings of the IEEE}, vol.~98,
  no.~11, pp. 1847--1864, November 2010.

\bibitem{OL07}
R.~Olfati-Saber, J.~Fax, and R.~Murray, ``Consensus and cooperation in
  networked multi-agent systems,'' \emph{Proceedings of the IEEE}, vol.~95,
  no.~1, pp. 215--233, Jan. 2007.

\bibitem{ME10}
M.~Mesbahi and M.~Egerstedt, \emph{Graph Theoretic Methods in Multiagent
  Networks}.\hskip 1em plus 0.5em minus 0.4em\relax Princeton University Press,
  2010.

\bibitem{AY08a}
T.~Aysal, M.~Yildiz, and A.~Scaglione, ``Broadcast gossip algorithms,'' in
  \emph{Proc. IEEE Information Theory Workshop}, Porto, Portugal, May 2008.

\bibitem{AY08b}
T.~Aysal, M.~Yildiz, A.~Sarwate, and A.~Scaglione, ``Broadcast gossip
  algorithms: Design and analysis for consensus,'' in \emph{Proc. IEEE
  Conf.~Decision and Control}, Cancun, Mexico, Dec. 2008.

\bibitem{AY09}
------, ``Broadcast gossip algorithms for consensus,'' \emph{IEEE Transactions
  on Signal Processing}, vol.~57, no.~7, pp. 2748--2761, July 2009.

\bibitem{FR09}
M.~Franceschelli, A.~Giua, and C.~Seatzu, ``Consensus on the average on
  arbitrary strongly connected digraphs based on broadcast gossip algorithms,''
  in \emph{Proc. 1st IFAC Workshop on Estimation and Control of Networked
  Systems}, 2009, pp. 66--71.

\bibitem{TS86}
J.~Tsitsiklis, D.~Bertsekas, and M.~Athans, ``Distributed asynchronous
  deterministic and stochastic gradient optimization algorithms,'' \emph{IEEE
  Trans. Automatic Control}, vol.~31, no.~8, pp. 803--812, Sep. 1986.

\bibitem{BE97}
D.~Bertsekas and J.~Tsitsiklis, \emph{Parallel and Distributed Computation:
  Numerical Methods}.\hskip 1em plus 0.5em minus 0.4em\relax Athena Scientific,
  1997.

\bibitem{FR11}
F.~Fagnani and P.~Frasca, ``Broadcast gossip averaging: Interference and
  unbiasedness in large {Abelian} {Cayley} networks,'' \emph{IEEE J. Selected
  Topics in Signal Processing}, vol.~5, no.~4, pp. 866--875, August 2011.

\bibitem{NE11}
A.~Nedi\'{c}, ``Asynchronous broadcast-based convex optimization over a
  network,'' \emph{IEEE Trans. Automatic Control}, vol.~56, no.~6, pp.
  1337--1351, June 2011.

\bibitem{FA11}
M.~Franceschelli, A.~Giua, and C.~Seatzu, ``Distributed averaging in sensor
  networks based on broadcast gossip algorithms,'' \emph{IEEE Sensors Journal},
  vol.~11, no.~3, pp. 808--817, March 2011.

\bibitem{CA11}
K.~Cai and H.~Ishii, ``Average consensus on general digraphs,'' in \emph{Proc.
  IEEE Conf. Decision and Control}, Orlando, FL, 2011, pp. 1956--1961.

\bibitem{CA12}
------, ``Average consensus on general strongly connected digraphs,'' March
  2012, available at http://arxiv.org/abs/1203.2563.

\bibitem{ME01}
C.~Meyer, \emph{Matrix analysis and applied linear algebra}.\hskip 1em plus
  0.5em minus 0.4em\relax SIAM, 2001.

\bibitem{SE04}
A.~Seyranian and A.~Mailybaev, \emph{Multiparameter stability theory with
  mechanical applications}.\hskip 1em plus 0.5em minus 0.4em\relax World
  Scientific, 2004.

\bibitem{Seneta}
E.~Seneta, \emph{Non-negative Matrices and Markov Chains}.\hskip 1em plus 0.5em
  minus 0.4em\relax Springer, 1981.

\bibitem{XI04}
L.~Xiao and S.~Boyd, ``Fast linear iterations for distributed averaging,''
  \emph{Systems and Control Letters}, vol.~53, no.~1, pp. 65--78, September
  2004.

\bibitem{BO10}
S.~Boyd, A.~Ghosh, B.~Prabhakar, and D.~Shah, ``Randomized gossip algorithms,''
  \emph{IEEE Transactions on Information Theory}, vol.~52, no.~6, pp.
  2508--2530, June 2006.

\bibitem{BH96}
R.~Bhatia, \emph{Matrix Analysis}.\hskip 1em plus 0.5em minus 0.4em\relax
  Springer-Verlag, 1996.

\bibitem{ST90}
G.~Stewart and J.~Sun, \emph{Matrix Perturbation Theory}.\hskip 1em plus 0.5em
  minus 0.4em\relax Academic Press, 1990.

\bibitem{BE00}
A.~Berman and X.~Zhang, ``Lower bounds for the eigenvalues of {Laplacian}
  matrices,'' \emph{Linear Algebra and its Applications}, vol. 316, no. 1--3,
  pp. 13--20, September 2000.

\bibitem{GU98}
P.~Gupta and P.~R. Kumar, ``Critical power for asymptotic connectivity in
  wireless networks,'' in \emph{Stochastic Analysis, Control, Optimization, and
  Applications}, Boston, 1998, pp. 1106--1110.

\bibitem{PE03}
M.~Penrose, \emph{Random Geometric Graphs}.\hskip 1em plus 0.5em minus
  0.4em\relax Oxford University Press, 2003.

\end{thebibliography}

\end{document}